\documentclass{article}
\usepackage[utf8]{inputenc}
\usepackage{authblk}
\usepackage[margin=1in]{geometry}
\usepackage[titletoc,title]{appendix}
\usepackage[dvipsnames,table,xcdraw]{xcolor}
\usepackage{color}
\usepackage{setspace}
\usepackage{amsmath}

\usepackage{amsmath,amsfonts,amssymb,amsthm,mathtools}
\usepackage{algorithm2e}
\RestyleAlgo{ruled}

\usepackage{graphicx,float}
\usepackage{subcaption}  
\usepackage{booktabs}
\usepackage{booktabs,multirow}
\usepackage{epsfig}

\usepackage[colorlinks=true, pdfborder={0 0 0}, linkcolor=red]{hyperref}

\usepackage[style=authoryear,date=year,maxcitenames=2,maxbibnames=20,giveninits,url=false,backend=biber,uniquename=false,uniquelist=false]{biblatex}
\bibliography{bibConfiBandQuant-1}
\usepackage{wasysym}

\allowdisplaybreaks
\title{\textbf{Simultaneous Inference of a Partially Linear Model in Time Series}}
\author{Jiaqi Li\thanks{Department of Mathematics and Statistics, Washington University in St. Louis, USA. Email: lijiaqi@wustl.edu.}\qquad Likai Chen\thanks{Department of Mathematics and Statistics, Washington University in St. Louis, USA. Email: likai.chen@wustl.edu.}\qquad
Kun Ho Kim\thanks{Department of Finance, John Molson School of Business, Concordia University, Canada. Email: kunhokim8@gmail.com}\qquad
Tianwei Zhou\thanks{Department of Mathematics and Statistics, Washington University in St. Louis, USA. Email: tianweizhou@wustl.edu.\\ We thank Richard Baillie, John Cochrane, Wolfgang H\"{a}rdle, Hira Koul, Jeffrey Racine and Harald Uhlig for their helpful comments. All errors belong to the authors.} 
}

\date{\today}

\begin{document}

\maketitle

\begin{abstract}
We introduce a new methodology to conduct simultaneous inference of the nonparametric component in partially linear time series regression models where the nonparametric part is a multivariate unknown function. In particular, we construct a simultaneous confidence region (SCR) for the multivariate function by extending the high-dimensional Gaussian approximation to dependent processes with continuous index sets. Our results allow for a more general dependence structure compared to previous works and are widely applicable to a variety of linear and nonlinear autoregressive processes. We demonstrate the validity of our proposed methodology by examining the finite-sample performance in the simulation study. Finally, an application in time series, the forward premium regression, is presented, where we construct the SCR for the foreign exchange risk premium from the exchange rate and macroeconomic data.

\vspace{1cm}

\vspace{1cm}

\end{abstract}

\setlength{\parindent}{15pt}

\def\One{\mathbf{1}}

\def\III{\text I}

\def\A{\mathcal A}
\def\B{\mathcal B}
\def\C{\mathcal C}
\def\D{\mathcal D}
\def\F{\mathcal F}
\def\G{\mathcal G}
\def\H{\mathcal H}
\def\I{\mathcal I}
\def\K{\mathcal K}
\def\L{\mathcal L}
\def\M{\mathcal M}
\def\P{\mathcal P}
\def\S{\mathcal S}
\def\T{\mathcal T}
\def\V{\mathcal V}
\def\W{\mathcal W}
\def\Z{\mathcal Z}

\def\EE{\mathbb E}
\def\II{\mathbb I}
\def\NN{\mathbb N}
\def\PP{\mathbb P}
\def\RR{\mathbb R}
\def\ZZ{\mathbb Z}

\def\bbeta{\boldsymbol \beta}

\newtheorem{definition}{Definition}
\newtheorem{assumption}{Assumption}
\newtheorem{theorem}{Theorem}
\newtheorem{proposition}{Proposition}
\newtheorem{corollary}{Corollary}
\newtheorem{lemma}{Lemma}
\newtheorem{remark}{Remark}
\newtheorem{example}{Example}

\newpage

\section{Introduction}\label{sec_intro}

Partially linear models are of interest in many practical problems. For example, in econometrics, \textcite{engle_semiparametric_1986} modeled the electricity sales as the combination of a smooth function of temperature and a linear function of price and income; in materials science, \textcite{green_semi-parametric_1985} used a semi-parametric generalized linear model to analyze the bioassay data for the study of flame retardants; in biology, \textcite{liang_empirical_2009} applied generalized partially linear models to investigate the relationship between viral load and CD4$^+$ cell counts to understand AIDS pathogenesis. See other applications in \textcite{hardle_partially_2000}.

In this paper, we consider a partially linear time series regression model
\begin{equation}
    \label{eq_model}
    Y_i=Z_i^{\top}\bbeta+\mu(X_i)+\sigma(X_i)\epsilon_i,\quad i=1,\ldots,n,
\end{equation}
where $(Z_i,X_i,Y_i)$ are observed stationary processes with $Z_i\in\RR^{l}$, $X_i\in\RR^d$ and $Y_i\in\RR$, for $l,d\ge1$. Here $\bbeta\in\RR^{l}$ is a fixed vector of unknown parameters and $\mu(\cdot)$ [resp. $\sigma^2(\cdot)$] is an unknown smooth regression function (resp. conditional variance or volatility function) from $\RR^d$ to $\RR$. In addition, $\epsilon_i\in\RR$ is an unobserved random error with mean zero, independent of the covariates $X_i$ and $Z_i$. In this work, we'll primarily concentrate on the conditional volatility $\sigma(X_i)$ for clarity's sake although our findings can be extended to $\sigma(X_i,Z_i)$. Compared to completely parametric or nonparametric specifications, a partially linear model in (\ref{eq_model}) enjoys a flexible semi-parametric structure. The parametric components can provide easier interpretations of each variable to better characterize the underlying data-generating mechanism, while the additional nonparametric part allows a data-driven approximation with no specific structures imposed on the true regression function, which can avoid inconsistent estimators and faulty inferences due to model mis-specification in purely parametric models.

In the past decades, much attention has been directed to estimating and testing partially linear models. See, for instance, \textcite{engle_semiparametric_1986,rice_convergence_1986,Robinson,speckman_kernel_1988,schick_root-n_1996} on the $\sqrt{n}$-consistent estimators of $\bbeta$; \textcite{gao_convergence_1995,fan_profile_2005,xie_scad-penalized_2009} on the inferences of $\bbeta$. It is crucial to include the parametric component in equation (\ref{eq_model}) because the parametric component is derived out of relevant theories and it is practically useful to identify the parametric component. One notable example is the Phillips curve with a time-varying natural unemployment rate (\cite{KHK:2014}), where the parameter $\bbeta$ captures the effect of the unemployment rate on the price inflation. Moreover, $\bbeta$ can be employed to measure the impact of various demographic variables, such as age, gender, family size, and the residency type, on the household gasoline consumption, as demonstrated in the U.S. case study by \textcite{kim_simultaneous_2021}. In fact, this paper significantly extends the scope of both \textcite{KHK:2014} and \textcite{kim_simultaneous_2021} by introducing a multivariate, time-dependent, and stochastic nonlinear component into (\ref{eq_model}). Furthermore, the parameter $\bbeta$ in equation (\ref{eq_model}) can represent the factor that determines the efficiency of the foreign currency market, as discussed in Section \ref{sec_app}. Considering these diverse roles played by the parametric component, it is essential and useful to include the linear parametric part for the estimation and the statistical inference of model (\ref{eq_model}), instead of relying solely on the nonparametric part. It is worth noting that the purely nonparametric model corresponds to a special case of model (\ref{eq_model}) where $\bbeta$ equals zero.

The estimation of the nonparametric part has also been intensively studied, including the methods based on kernel, local linear and spline smoothers  (\cite{hamilton_local_1997,yu_penalized_2002,fan_kernel-based_2003,aneiros-perez_local_2008}). Several attempts have been made to the inference of $\mu$ in partially linear models, such as 
the consistency and asymptotic normality for the estimator of $\mu$ by \textcite{liang_asymptotic_1997}, the point-wise confidence intervals of $\mu$ based on empirical likelihood by \textcite{liang_empirical_2009}, and simultaneous confidence bands of multivariate function $\mu$ by \textcite{KHK:2014,kim_simultaneous_2021}. However, all the aforementioned literature focused on the {\it independent} or non-stochastic observations. No previous work investigated the simultaneous inference of $\mu$ in a partially linear time series model with {\it dependence} as (\ref{eq_model}) that is commonly encountered in real data (\cite{hardle_partially_2000}). Further, most of the studies on partially linear models assume the error terms in model (\ref{eq_model}) to be homoskedastic, where the error term $\sigma(X_i)\epsilon_i$ is simply reduced to $\epsilon_i$ and the conditional variance is constant over time. This can be a shortcoming since it would rule out most macroeconomic and financial time series, at least for the application to asset pricing (\cite{nelson_conditional_1991}), where returns may be uncorrelated but feature stochastic volatility. The current paper aims to fill in these gaps by providing theory for the simultaneous inference of the mean trend $\mu$ in (\ref{eq_model}) under a general dependency structure, which also allows for conditional heteroskedasticity.

Specifically, we allow both the errors $\epsilon_i$ and the covariates $X_i$ in (\ref{eq_model}) to be dependent over $i$. Let $X_i=(X_{i1},X_{i2},\ldots,X_{id})^{\top}$ be a stationary process of the form
\begin{equation}
    \label{eq_X_structure}
    X_i = H(\ldots,v_{i-1},v_i),
\end{equation}
where $v_i$ are independent and identically distributed (i.i.d.) random vectors in $\RR^{d'}$ for 
$d'\ge1$ and $H=(H_1,H_2,\ldots,H_d)^{\top}$ is a measurable function such that $X_i$ is a well-defined. The nonlinear Wold representation (\ref{eq_X_structure}) allows a very general class of stationary processes, including linear processes such as vector autoregressive models (VAR) and autoregressive moving average (ARMA) models and nonlinear transforms such as bilinear models, Volterra processes, Markov chain models, threshold/exponential autoregressive models (TAR/EAR) and (generalized) autoregressive conditionally heteroscedastic (ARCH/GARCH) type models, etc. Within this framework, $v_i$ can be viewed as independent inputs of a physical system, and all the dependencies among the outputs $X_i$ result from the underlying data-generating mechanism $H(\cdot)$. 

For the identification of model (\ref{eq_model}), we assume that the error $\epsilon_i$ is independent of both covariates $X_i$ and $Z_i$. In particular, we shall proceed with the conditional expectation $\EE(Y_i\mid X_i,Z_i)=Z_i^{\top}\bbeta + \mu(X_i)$. Following the fixed design case in \textcite{hardle_partially_2000}, we let $Z_i$ be a function of $X_i$ plus another noise term (cf. Assumption \ref{asm_iden}). This means $\EE(Y_i\mid X_i,Z_i)$ cannot be reduced to $\EE(Y_i\mid X_i)$ and indicates that it is nontrivial and also challenging to extend the inference of $\mu(\cdot)$ in a purely nonparamteric model to that under a partially linear setting. If $Z_i=0$, then model (\ref{eq_model}) is simply a nonparametric regression process as a special case, that is, 
\begin{equation}
    \label{eq_model_example}
    Y_i= \mu(X_i)+\sigma(X_i)\epsilon_i,\quad i=1,\ldots,n.
\end{equation}
When $X_i=Y_{i-1}$ and $\epsilon_i$ are i.i.d. random noises, model (\ref{eq_model_example}) incorporates many interesting linear and nonlinear autoregressive processes (AR),  such as AR processes if $\mu(x)=ax$ for some real parameter $a$, and autoregressive conditional heteroscedastic (ARCH) processes if  $\mu(x)=0$ and $\sigma^2(x)=\alpha_0+\alpha_1x^2$ for some non-negative real parameters $\alpha_0,\,\alpha_1\in\RR$.

Many contributions have been made to 
constructing the SCR of $\mu(\cdot)$ in completely nonparametric models. For example, concerning independent data, \textcite{johnston_probabilities_1982} was among the first to investigate the inferences of univariate mean regression functions; \textcite{hardle_asymptotic_1989} derived simultaneous confidence bands for one-dimensional kernel M-estimators; \textcite{HS:2010,GH:2012} constructed uniform confidence bands for conditional quantile and expectile functions, respectively. With regard to dependent cases, see inference of trends in a fixed design with $X_i=i/n$ by \textcite{WZ:2007}; confidence bands for the mean function in functional time series with physical dependence by \textcite{CS2015}; nonlinear regression model with nonstationary regressors by \textcite{Li2017} and a time-varying nonlinear regression model by \textcite{ZW2015}. In particular, \textcite{ZW:2008,LW:2010} proposed inferences of the univariate mean and volatility functions in a similar time series regression model in (\ref{eq_model_example}), where they assumed that the error terms $\epsilon_i$ are i.i.d.. Our work can be viewed as a generalization of their results by extending the dependence structure of $\epsilon_i$ and by including an additional parametric part to accommodate a broader class of data-generating mechanisms.
That is, our work is distinct from \textcite{ZW:2008,LW:2010} in that our model framework is semi-parametric with the multivariate covariate $X_i$ and time-dependent $\epsilon_i$, while the framework in \textcite{ZW:2008,LW:2010} is purely nonparametric with an univariate $X_i$ and an i.i.d. noise $\epsilon_i$.

\textbf{Contributions:} Here we summarize our three main contributions to the literature: Firstly, we extend the dependence structure of the error terms to a more general case by allowing $\epsilon_i$ to be {\it dependent over} $i$, while also accounting for the dependence among the covariates $X_i$ and the conditional heteroscedasticity. Secondly, different from the relevant studies relying on the Gumbel convergence to achieve the asymptotics of the statistics (\cite{zhao_kernel_2006,LW:2010}), we provide a new testing methodology based on the multiplier bootstrap enlightened by \textcite{chernozhukov_central_2017}. This allows one to {\it avoid the notoriously slow convergence issue} associated with the Gumbel distribution.
Thirdly, there is no previous work performing simultaneous inference of the {\it multivariate} $\mu(\cdot)$ in (\ref{eq_model}), allowing the multivariate covariate $X_i\in\mathbb{R}^d$ with $d\geq 2$ under some general {\it time dependence} setting. This paper could be a complement to the non-parametric model validation problem with a general dependence structure and conditional heteroscedasticity, applicable in various multivariate scenarios.

\textbf{Notation:} For a vector $v=(v_1,...,v_d)\in\RR^d$ and $q>0$, we denote $|v|_q=(\sum_{i=1}^d|v_i|^q)^{1/q}$ and $|v|_{\infty}=\max_{1\le i\le d}|v_i|$. For $s>0$ and a random vector $X$, we say $X\in\L^s$ if $\lVert X\rVert_s=[\EE(|X|_2^s)]^{1/s}<\infty$. For two positive number sequences $(a_n)$ and $(b_n)$, we say $a_n=O(b_n)$ or $a_n\lesssim b_n$ (resp. $a_n\asymp b_n$) if there exists $C>0$ such that $a_n/b_n\le C$ (resp. $1/C\le a_n/b_n\le C$) for all large $n$, and say $a_n=o(b_n)$ if $a_n/b_n\rightarrow0$ as $n\rightarrow\infty$. We set $(X_n)$ and $(Y_n)$ to be two sequences of random variables. Write $X_n=O_{\PP}(Y_n)$ if for $\forall \epsilon>0$, there exists $C>0$ such that $\PP(|X_n/Y_n|\le C)>1-\epsilon$ for all large $n$, and say $X_n=o_{\PP}(Y_n)$ if $X_n/Y_n\rightarrow 0$ in probability as $n\rightarrow\infty$. We denote the centered random variable $X$ by $\EE_0(X)$, that is, $\EE_0(X)=X-\EE(X)$. 

\textbf{Roadmap:} The rest of the paper is structured as follows. Section \ref{sec_SCR} introduces the overall methodology to perform simultaneous inference of $\mu(\cdot)$ in (\ref{eq_model}) with $d\geq 1$. The asymptotic properties of the proposed statistics and estimators as well as the implementation are provided in Sections \ref{sec_asym} and \ref{sec_est}. Section \ref{sec_simul} is devoted to a simulation study to evaluate the performance of our methods and Section \ref{sec_app} offers an empirical application, the forward premium anomaly, to demonstrate the validity of the proposed methodology in practice. Section \ref{conclusion} concludes the paper and discusses potential extensions for future research. The technical proofs are deferred to the Supplementary Materials.

\section{Simultaneous Confidence Region (SCR)}\label{sec_SCR}
In this section, we first introduce the definition of the simultaneous confidence region (SCR). Then, we shall follow with the estimator of the nonparametric trend function $\mu(\cdot)$ in (\ref{eq_model}). Further, we illustrate our new methodology on constructing the SCR of $\mu(\cdot)$ based on this estimated $\mu(\cdot)$. The theoretical intuition of the proposed simultaneous inference is also provided.

To conduct simultaneous inference of the trend $\mu(\cdot)$ in model (\ref{eq_model}), we shall construct the nonparametric simultaneous confidence region (SCR) for $\mu(\cdot)$. 
In particular, we consider deriving asymptotic SCR for $\mu(\cdot)$ over the region $\T_d=[T_{11},T_{12}]\times [T_{21},T_{22}]\times \cdots \times [T_{d1},T_{d2}] \subset \RR^d$ with confidence level $100(1-\alpha)\%$, $\alpha\in(0,1)$. To this end, we shall find two functions $l_n(\cdot)$ and $r_n(\cdot)$ based on the observations $(Z_i,X_i,Y_i)$, $1\le i\le n$, such that
\begin{equation}
    \label{eq_UCBdef}
    \lim_{n\rightarrow\infty}\PP\Big(l_n(x)\le \mu(x)\le r_n(x), \text{ for all } x\in\T_d\Big)=1-\alpha.
\end{equation}
Given the SCR for $\mu(\cdot)$, we can verify whether $\mu(\cdot)$ is of some certain parametric form by testing the null hypothesis
\begin{equation}
    \label{eq_testnull}
    \H_0:\,\mu(\cdot)=\mu_{\theta}(\cdot),
\end{equation}
against the alternative $\H_{\A}:\, \mu(\cdot)\neq\mu_{\theta}(\cdot)$, where $\theta\in\Theta$ for some parametric space $\Theta$ and $\mu_{\theta}(\cdot)$ is a multivariate parametric function. Specifically, one can test (\ref{eq_testnull}) by checking whether the condition $l_n(x)\le\mu_\theta(x)\le r_n(x)$ holds for all $x\in\T_d$. If this condition does not hold for some $x\in\T_d$, 
then we reject the null hypothesis at level $\alpha$.
The SCR-based inference is more preferred to other standard inferential procedures utilizing mean-integrated-squared-error (MISE) type statistics, for example, since it is more effective in suggesting the right function form of $\mu(\cdot)$ in (\ref{eq_model}). When the null hypothesis in (\ref{eq_testnull}) is rejected via an MISE-type test statistic, it would be rather difficult to figure out the reason for rejection, which, however, can be easily dealt with under our approach by locating graphically where the SCR is violated by $\mu_{\theta}(\cdot)$ under the null hypothesis.

Next, we provide an estimator for $\mu(\cdot)$ given the observed sample $(Z_i,X_i,Y_i)$, $1\le i\le n$. Let $x=(x_1,\ldots,x_d)^{\top}\in\RR^d$ and set $K(\cdot)\ge0$ to be some kernel function with support $[-1,1]^d$. We consider the following optimization problem:
\begin{equation}
    \label{eq_mu_hat_star_goal}
    \hat\mu^*(x)=\text{argmin}_{\theta}n^{-1}\sum_{i=1}^nK_h(x-X_i)\big(Y_i-Z^{\top}_{i}\hat\bbeta-\theta\big)^2,
\end{equation}
where $K_h(\cdot)=K(\cdot/h)/h^d$, $h$ is a bandwidth parameter with $h\rightarrow0$ and $h^dn\rightarrow\infty$, and $\hat\bbeta$ is a consistent estimator of the unknown parameters $\bbeta$ in (\ref{eq_model}). We shall defer the details of $\hat\bbeta$ to Section \ref{sec_est}. Here in (\ref{eq_mu_hat_star_goal}), we adopt the local constant estimator for the simplicity of notation. One can achieve similar results by applying the local linear estimator introduced in \textcite{fan_local_1996}. By solving (\ref{eq_mu_hat_star_goal}), we can obtain the Nadaraya-Watson estimator for $\mu(\cdot)$ which has the expression
\begin{equation}
    \label{eq_mu_hat_star}
    \hat\mu^*(x) =\sum_{i=1}^nw_h(x,X_i)\big(Y_i-Z^{\top}_{i}\hat\bbeta\big),
\end{equation}
where the weight function $w_h(x,X_i)$ is defined as
\begin{equation}
    \label{eq_weight}
    w_h(x,X_i) = \frac{K_h(x-X_i)}{\sum_{i=1}^nK_h(x-X_i)}.
\end{equation}
Moreover, we denote the consistent estimator of the volatility function $\sigma(\cdot)$ in (\ref{eq_model}) by $\hat\sigma(\cdot)$, and we shall provide the detailed definition and consistency results of $\hat\sigma(\cdot)$ in Section \ref{sec_est}. Given $\hat\bbeta$ and $\hat\sigma(\cdot)$, we consider the statistic $\sup_{x\in\T_d}\big|\hat\mu^*(x)- \mu(x)\big|/\hat\sigma(x)$ to construct the SCR of $\mu(x)$. We shall note that, due to the smoothness of $\mu(\cdot)$ and the consistency of $\hat\bbeta$, this statistic can be approximately written into the supremum of a sum of dependent random fields conditioned on the covariates $X_i$, that is
\begin{equation}
    \label{eq_GA_form}
    \sup_{x\in\T_d}\big|\hat\mu^*(x)- \mu(x)\big|/\hat\sigma(x) \approx \sup_{x\in\T_d}\Big|\sum_{i=1}^nw_h(x,X_i)\sigma(X_i)\epsilon_i\Big|/\hat\sigma(x).
\end{equation}
It is non-trivial to investigate the asymptotic properties of (\ref{eq_GA_form}) when $X_i$ and $\epsilon_i$ are dependent over $i$. \textcite{zhao_kernel_2006,LW:2010} have dealt with the case where the covariates $X_i$ are dependent while the errors $\epsilon_i$ are independent by establishing the Gumbel convergence. To address the more general dependency structure in our study, we propose an extension of the high-dimensional Gaussian approximation theorem introduced by \textcite{chernozhukov_central_2017} to dependent processes with continuous index sets. By this generalized high-dimensional Gaussian approximation, we shall expect the limit distribution of our proposed statistic to be approximated by the one of the maximum of a centered Gaussian random vector $\hat\Z=(\hat \Z_1,\ldots,\hat\Z_n)^{\top}\in\RR^n$, that is
\begin{equation}
    \label{eq_GA_intuitive}
    \PP\big(\sup_{x\in\T_d}\sqrt{h^dn}\big|\hat\mu^*(x)- \mu(x)\big|/\hat\sigma(x)<u\big) \approx \PP\big(\max_{1\le j\le n}|\hat\Z_j|<u\big).
\end{equation}
We defer the detailed definition of the covariance matrix for $\hat\Z$ to (\ref{eq_Q_hat}). Intuitively, the result in (\ref{eq_GA_intuitive}) would enable us to find the critical value of our proposed statistic, and consequently facilitates the construction of the simultaneous confidence region for $\mu(\cdot)$. Specifically, we can approximate $l_n(x)$ and $r_n(x)$ in (\ref{eq_UCBdef}) by the estimators
\begin{equation}
    \label{eq_ln_rn}
    \hat l_n(x) = \hat\mu^*(x) - \hat q_{\alpha}\hat\sigma(x), \quad \hat r_n(x) = \hat\mu^*(x) + \hat q_{\alpha}\hat\sigma(x),
\end{equation}
respectively, where $\hat q_{\alpha}$ is the $(1-\alpha)$-th empirical quantile of $\max_{1\le j\le n}|\hat\Z_j|/\sqrt{h^dn}$ given the significance level $\alpha\in(0,1)$, and it can be evaluated by the multiplier bootstrap (\cite{chernozhukov_central_2017}). Based on the SCR in (\ref{eq_ln_rn}), one can test whether the trend function $\mu(\cdot)$ is of any particular parametric form $\mu_\theta(\cdot)$, such as quadratic or cubic patterns, by evaluating whether $l_n(x)\le\mu_\theta(x)\le r_n(x)$ is satisfied for all $x\in\T_d$. If the SCR fails to {\it entirely} contain this parametric form, then we reject the null hypothesis (\ref{eq_testnull}) at level $\alpha$. We shall provide the detailed steps for implementing the SCR construction at the end of Section \ref{sec_est} after we introduce our main theorems.




\section{Asymptotic Properties}\label{sec_asym}
This section is devoted to our main results on the asymptotic properties for the statistic $\sup_{x\in\T_d}\big|\hat\mu^*(x)- \mu(x)\big|/\hat\sigma(x)$, which provides the theoretical foundation for the construction of SCR. In Section \ref{subsec_GA}, we shall first establish the asymptotic distribution of the proposed statistic under the oracle setting, that is, assuming that the unknown parameters $\bbeta$ and $\sigma(\cdot)$ in the statistic are the true ones. The case with $\bbeta$ and $\sigma(\cdot)$ replaced by their consistent estimators $\hat\bbeta$ and $\hat\sigma(\cdot)$, respectively, are dealt with in Section \ref{sec_est}, where we provide the consistency results for $\hat\bbeta$ and $\hat\sigma(\cdot)$ and show the similar asymptotic distribution of the statistic. 

\subsection{Technical Assumptions}\label{subsec_asm}
We shall start with some regularity conditions which will be useful to establish our main theorems. First, we impose the smoothness condition on the trend function $\mu(\cdot)$ and assume that the volatility function $\sigma(\cdot)$ also varies smoothly and is bounded on the support $\T_d$.
\begin{assumption}[Trend and volatility]\ \\
    \label{asm_sigma}
    (i) (Smoothness)
    Assume that the trend function $\mu(\cdot)$ and volatility function $\sigma(\cdot)$ defined in (\ref{eq_model}) are both Lipschitz continuous on $\T_d$. \\
    (ii) (Bounds) Assume that for some constants $c_{\sigma},c_{\sigma}'>0$, $c_{\sigma}\le \inf_{x\in\T_d}\sigma(x)\le \sup_{x\in\T_d}\sigma(x)\le c_{\sigma}'$.
\end{assumption}
\begin{assumption}[Kernel]\ \\
    \label{asm_kernel}
    (i) The kernel function $K(\cdot)$ in (\ref{eq_mu_hat_star_goal}) is defined on $\II=[-1,1]^d$ and is continuously differentiable up to order two. \\
    (ii) Assume that $\sup_{x\in \II}|K(x)|<\infty$ and $\int_{\II} K(x)dx =1$. Also assume that $K(x)$ has first-order derivative with $ \sup_{x \in \II}\max_{1\le i\le d}|\partial K(x)/\partial x_i|<\infty$. \\
    (iii) Assume that the bandwidth parameter $h\rightarrow0$ and $h^dn\rightarrow\infty$.
\end{assumption}
\begin{remark}[Choice of the bandwidth parameter $h$]
\label{remark_bandwidth}
The bandwidth parameter $h$ can be selected by some data-driven method such as the generalized cross-validation (GCV) introduced by \textcite{CG:1977}. Specifically, the bandwidth $h$ is firstly chosen by GCV and then adjusted downward for undersmoothing, which means $h$ is chosen to converge to 0 faster than the optimal rate yielded by GCV. We defer the detailed discussion of bandwidth selection to Remark \ref{rmk_thm1} after we introduce some additional conditions on $h$ in Theorem \ref{thm1_GA}. Furthermore, we show the results of our simulation study in Section \ref{sec_simul} where we benchmark the effects of multiple bandwidths $h$ on the coverage probabilities.
\end{remark}

Next, we shall specify the dependence structure of the error process $\{\epsilon_i\}_{1\le i\le n}$ and the covariates $X_1,\ldots,X_n$ in model (\ref{eq_model}). In particular, throughout this paper, we assume that $\{\epsilon_i\}_{1\le i\le n}$ is MA($\infty$), which can be formulated as follows:
\begin{equation}
    \label{eq_epsilon}
    \epsilon_i = \sum_{k=0}^{\infty}a_k\eta_{i-k},
\end{equation}
where $\eta_i\in\RR$ are i.i.d. random variables with mean zero and unit variance, independent of $(X_i,Z_i)$ in (\ref{eq_model}). The coefficients $a_k$, $k\ge0$, take values in $\RR$ such that $\epsilon_i$ is a proper random variable. We assume that the innovations $\eta_i$ have finite $q$-th moment, for some $q\ge 4$. We define the absolute sum of the coefficients as $S=\big|\sum_{k\ge0}a_k\big|$. Then, the long-run variance of $\{\epsilon_i\}_{1\le i\le n}$ is 
\begin{equation}
    \label{eq_longrun}
    S^2=\Big(\sum_{k\ge0}a_k\Big)^2=\sum_{k=-\infty}^{\infty}\gamma(k),
\end{equation}
where $\gamma(k)=\EE(\epsilon_i\epsilon_{i+k})$ is the autocovariance of innovations with lag $k\in\ZZ$. 
\begin{assumption}[Finite moment]
    \label{asm_moment}
    Assume that innovations $\{\eta_k\}$ defined in (\ref{eq_epsilon}) are i.i.d. with $\lVert \eta_1\rVert_q<\infty$, for some constant $q\ge4$.
\end{assumption}

\begin{assumption}[Dependence of $\epsilon$]\label{asm_dep_epsilon}
Assume that for any integer $l\ge0$, 
$$\sum_{k\ge l}|a_k|/S=O\big\{(1\vee l)^{-\zeta}\big\},$$
where the constant $\zeta>1$ and $S$ is the long-run standard deviation of $\{\epsilon_i\}_{1\le i\le n}$ defined in (\ref{eq_longrun}). 
\end{assumption}
Assumptions \ref{asm_moment} and \ref{asm_dep_epsilon} post conditions on the moment and dependency structure of errors $\{\epsilon_i\}_{1\le i\le n}$. Specifically, the moment condition in Assumptions \ref{asm_moment} depends on $q$ which characterizes the heavy-tailedness of the noise, and a larger $q$ means a thinner tail. Assumption \ref{asm_dep_epsilon} requires that the dependency strength of $\{\epsilon_i\}_i$ decays at a polynomial rate, which also ensures that the long-run variance of $\{\epsilon_i\}_i$ is finite. 

For the covariates $X_1,\ldots,X_n$, we denote the density function of $X_{i}=(X_{i1},\ldots, X_{id})^{\top}$ by $g(x_1,\ldots, x_d)$. For $s=(s_1,s_2,\ldots,s_d)^{\top} \in \mathbb{R}^d$, we define
\begin{equation}
    \label{eq_density_x}
    g(s\mid \mathcal{F}_{i-1})=\frac{\partial^d\PP(X_{i}\leq s\mid \mathcal{F}_{i-1})}{\partial s_1\ldots \partial s_d},
\end{equation}
where the filtration $\mathcal{F}_i = (\ldots ,v_{i-1},v_{i})$, $i\in \ZZ$.
We let $v_i'$ be an i.i.d. copy of $v_i$ and $\mathcal{F}_{i,\{k\}}$ be $\mathcal{F}_{i}$ with $v_k$ therein replaced by $v_k'$. We define $X_{ij,\{i-k\}}$ as a coupled version of $X_{ij}$ with the form
$$X_{ij,\{i-k\}}=H_j(\ldots ,v_{i-k-1},v'_{i-k},v_{i-k+1},\ldots ,v_i).$$
In addition, we denote $X_{i,\{i-k\}}=(X_{i1,\{i-k\}},\ldots ,X_{id,\{i-k\}})^\top$. We impose assumptions on the moments and dependency structures of covariates $X_1,\ldots,X_n$ as follows.
\begin{assumption}[Covariates $X_i$] \ \\
\label{asm_dep}
(i) (Finite moment). Assume that the conditional density function $g(x\mid\F_{i-1})$ defined in (\ref{eq_density_x}) has finite $s$-th moment, i.e. $\|g(x\mid\F_{i-1})\|_s<\infty$, for some constant $s>2$. \\ 
(ii) (Dependence strength). Following \textcite{Wu:2005}, for all $k\ge1$, we define the physical dependence measure
$$\theta_{k,s} =\sup_{x \in \T_d} \big\lVert g(x\mid \F_{i-1})-g(x\mid  \F_{i-1,\{i-k\}})\big\rVert_s.$$
We assume that for some constant $\xi > 0$ and positive integer $m$,
$$\sum_{k\geq m} \theta_{k,s} =O(m^{-\xi}).$$
\end{assumption}
Roughly speaking, the functional dependence measure $\theta_{k,s}$ defined in Assumption \ref{asm_dep} quantifies the dependence of $X_i$ on $v_{i-k}$ by measuring the distance between $g(x\mid\F_{i-1})$ and its coupled version $g(x\mid\F_{i-1,\{i-k\}})$. The physical dependence measure accounts for any measurable function of cumulative i.i.d. noises, which makes it applicable to a wide range of linear and nonlinear time series processes. To elaborate such conditional dependency structure, we shall consider a simple moving average (MA) example as follows. 

\begin{example}
Suppose that the covariates $(X_1,X_2,\ldots,X_n)$ is a linear process, which takes the form
$$X_{i} = \sum_{l\geq 0}A_{l}v_{i-l},$$
where $v_i \in \mathbb{R}^d$ are i.i.d random vectors with continuous density function $f(\cdot): \RR^d\mapsto\RR$, and $A_l\in\RR^{d\times d}$ are coefficient matrices, $l\ge0$. Without loss of generality, we assume that $A_0$ is an identity matrix, and then, $\theta_{k,s}$ in this case can be written into
$$\theta_{k,s} = \sup_{x \in \T_d} \Big\| f\Big(x-\sum_{l\ge 1}A_l v_{i-l}\Big)-f\Big(x-\sum_{l\ge 1}A_l v_{i-l}+A_k(v_{i-k}-v'_{i-k})\Big)\Big\|_s,$$
which, by the continuity of the density function $f(\cdot)$, can be bounded as follows
\begin{align*}
\theta_{k,s}\lesssim  \max_{1\le j\le d}|A_{k,j,\cdot}|_2\|v_{i-k}-v'_{i-k}\|_s,   
\end{align*}
where $A_{k,j,\cdot}$ is the $j$-th row of matrix $A_k$. In Assumption \ref{asm_dep}(ii), this condition actually indicates an algebraic decay rate of the temporal dependence of $X_i$. Namely, for any $m\geq 1,$ there exists some constant $\xi^*>0,$ such that
  \begin{equation*}
  \max_{1\le j\le d}\sum_{k\geq m} |A_{k,j,\cdot}|_2 \lesssim m^{-\xi^*}.
  \end{equation*}

\end{example}

\begin{assumption}[Bounds and smoothness]
    \label{asm_smooth}
    For some constants $c_g,c_g'>0$, assume that
    $$c_g\leq \inf_{x \in \T_d} g(x)\leq \sup_{x \in \T_d} g(x) \leq c_g',$$
    and 
    $$\sup_{x \in \T_d}\max_{1\leq j\leq d} \big|\partial g(x)/\partial x_j\big| <\infty.$$
\end{assumption}
Assumption \ref{asm_smooth} imposes conditions on the boundary and smoothness of the density function $g(x)$ on the support $\T_d$. In the literature, conditions similar to Assumption \ref{asm_smooth} have been commonly used; see, for instance, \textcite{aneiros-perez_local_2008,kim_simultaneous_2021}.

\subsection{Gaussian Approximation}\label{subsec_GA}

This subsection is devoted to our main results, the asymptotic distribution of the proposed statistic derived by Gaussian approximation. To explicitly illustrate our theory on the simultaneous inference of trend $\mu(\cdot)$, we shall first assume that $\bbeta$ and $\sigma(\cdot)$ in model (\ref{eq_model}) are both known. In particular, we define 
\begin{equation}
    \label{eq_mu_hat_sol}
    \hat\mu(x) =\sum_{i=1}^nw_h(x,X_i)\big(Y_i-Z^{\top}_{i}\bbeta\big),
\end{equation}
and consider the statistic $\sup_{x\in\T_d}\big|\hat\mu(x)- \mu(x)\big|/\sigma(x)$. By Slutsky's theorem, the asymptotic distribution of the statistic still holds when replacing the true $\bbeta$ and $\sigma(\cdot)$ therein by their consistent estimators $\hat\bbeta$ and $\hat\sigma(\cdot)$, respectively. Hence, in this subsection, we shall first proceed under the oracle setting, and we refer to the implementation with theoretical guarantees to Section \ref{sec_est}.

As mentioned in Section \ref{sec_SCR}, we aim to provide the limiting distribution for the statistic $\sup_{x\in\T_d}\big|\hat\mu(x)- \mu(x)\big|/\sigma(x)$ by applying the high-dimensional Gaussian approximation. Specifically, we shall generate a centered Gaussian random field $\{\Z_t\}_{t\in\T_d}$ and utilize the distribution of $\sup_{t\in\T_d}|\Z_t|$ for the approximation. We denote the covariance matrix of $\sqrt{h^dn}\big|\hat\mu(x)- \mu(x)\big|/\sigma(x)$ conditioned on the covariates $X_i$, $1\le i\le n$, by $Q=(Q_{t,s})_{t,s\in\T_d}$, where $Q_{t,s}$ takes the form
\begin{equation}
    \label{eq_cov_Z}
    Q_{t,s}= h^dn\sum_{k=-\infty}^{\infty}\sum_{i=1\vee (1-k)}^{n\wedge (n-k)}c_{t,s,i,k}w_h(x_t,X_i)w_h(x_s,X_{i+k})\gamma(k),
\end{equation}
where $c_{t,s,i,k} = \sigma(x_t)^{-1}\sigma(x_s)^{-1}\sigma(X_i)\sigma(X_{i+k})$. Further, in practical usage, to evaluate the covariance matrix $Q$, we would need to estimate the autocovariance $\gamma(k)$, for each $k\ge0$, while this is unrealistic when $k$ goes to infinity. Therefore, to establish a consistent estimator for $Q$, we define a truncated version of $Q$ as $Q^{(L)}=(Q_{t,s}^{(L)})_{t,s\in\T_d}$, for some large positive integer $L$, where
$Q_{t,s}^{(L)}$ is defined as
\begin{equation}
    \label{eq_cov_Z_truncate}
    Q^{(L)}_{t,s}= h^dn\sum_{k=1-L}^{L-1}\sum_{i=1\vee (1-k)}^{n\wedge (n-k)}c_{t,s,i,k}w_h(x_t,X_i)w_h(x_s,X_{i+k})\gamma(k).
\end{equation}
When applied to real data, the autocovariance $\gamma(k)$ can be estimated by the notable existing methods and we defer the details of this estimation to Section \ref{sec_est}. Now we let $\{\Z_t\}_{t\in\T_d}$ be a Gaussian random field with mean zero and conditional covariance matrix $Q^{(L)}$ given the covariates $X_1,X_2,\ldots,X_n$. The first main theorem is stated as follows.
\begin{theorem}[Gaussian approximation]
    \label{thm1_GA}
    Suppose that Assumptions \ref{asm_sigma}-\ref{asm_smooth} hold. Then, for $$\Delta_1= (h^dn)^{-q/2}n\log^q(n)+1/n, \quad \Delta_2=(h^dn)^{-1/6} \log^{7/6} (n) + (n^{2/q}/(h^dn))^{1/3} \log(n),$$
    $$\Delta_3=h\log(n)+h^{2+d/2}n^{1/2
    }\sqrt{\log(n)}+n^{-\zeta}\log(n), \quad \Delta_4=L^{-\zeta/3}\log^{2/3}(n),$$
    we have
    $$\sup_{u\in\RR}\Big|\PP\big(\sup_{x\in\T_d}\sqrt{h^dn}\big|\hat\mu(x)- \mu(x)\big|/\sigma(x)<u\big) - \PP\big(\sup_{t\in\T_d}|\Z_t|<u\big) \Big| \lesssim \Delta_1+\Delta_2+\Delta_3 +\Delta_4,$$
    where the constants in $\lesssim$ are independent of $n$ and $h$. 
    If in addition,
    \begin{align}
        \label{thm1_o1}
        &(h^d)^{2-q}n^{2-q}\log^{3q}(n)\rightarrow0, \quad h\log(n)\rightarrow0, \quad h^{4+d}n\log(n) \rightarrow0,\nonumber \\ 
        & \qquad n^{-\zeta}\log(n) \rightarrow 0 \quad \text{and} \quad L^{-\zeta}\log^2(n) \rightarrow0,
    \end{align}
    then we have
    $$\sup_{u\in\RR}\Big|\PP\big(\sup_{x\in\T_d} \sqrt{h^dn}\big|\hat\mu(x)- \mu(x)\big|/\sigma(x)<u\big) - \PP\big(\sup_{t\in\T_d}|\Z_t|<u\big) \Big| \rightarrow 0.$$
\end{theorem}
\begin{remark}[Comments on the convergence rate]
    \label{rmk_thm1}
    In Theorem \ref{thm1_GA}, condition (\ref{thm1_o1}) imposes conditions on the strength of dependency and the bandwidth parameter $h$. The first part $(h^d)^{2-q}n^{2-q}\log^{3q}(n)\rightarrow0$ requires that $h$ should not be too small, while the second and third parts, $h\log(n)\rightarrow0$ and $h^{4+d}n\log(n)$, aim to keep $h$ from being too large. Therefore, one can apply cross-validation to choose the bandwidth $h$, such as the order of $n^{-1/3}$ or $n^{-1/5}$. The last two parts $n^{-\zeta}\log(n) \rightarrow 0$ and $L^{-\zeta}\log^2(n) \rightarrow0$ suggest that the dependency strength of $\epsilon_i$ cannot be too strong, and the integer $L$ in the conditional covariance matrix $Q^{(L)}$ should be large to approximate the true covariance matrix of the proposed statistic, respectively. For a general moving average model MA$(\infty)$ defined in (\ref{eq_epsilon}), we could choose $L$ to be some large positive integer such as $L=\sqrt{n}$. For some particular MA$(p)$ model, one could simply let $L=p$. We refer to Proposition \ref{prop_longrun} and the corresponding discussion for more details.
\end{remark}
As shown by Theorem \ref{thm1_GA}, our approach differs from the one based on the Gumbel convergence in extreme value theory adopted by \textcite{zhao_kernel_2006,LW:2010}. We extend the high-dimensional Gaussian approximation in \textcite{chernozhukov_central_2017} to dependent processes and provide a Gaussian limit distribution for $\sup_{x\in\T_d}|\hat\mu(x)-\mu(x)|/\sigma(x)$. Our method does not need the density estimate of covariate $X_i$ to build the test statistic, and thus offers easier implementation (see details in Section \ref{sec_est}) than the competing ones.

\section{Implementation of SCR}\label{sec_est}
In the previous sections, we assumed that $\bbeta$, $\sigma(\cdot)$ in model (\ref{eq_model}) and the conditional long-run covariance matrix of $\sup_{t\in\T_d}\big|\hat\mu(x)- \mu(x)\big|/\sigma(x)$ are all known, which, however, is not realistic in practice. Hence, we shall introduce the estimators $\hat\bbeta$, $\hat\sigma(\cdot)$ and $\hat Q^{(L)}$ in this subsection, followed by the consistency results and the limit distribution of the statistic built on these estimates. The detailed instructions on how to construct SCR in practice are also provided in Section \ref{subsec_scr_construct}.

\subsection{Consistent Estimators}
The literature on estimating the parametric components in model (\ref{eq_model}) has a long history. \textcite{Robinson} was among the first contributing to this problem, where he derived a least square estimator of $\bbeta$ based on a Nadaraya-Waston kernel estimator of $\mu$ and provided the $\sqrt{n}$-consistency result of the estimate. In this study, we shall consider the same estimator for $\bbeta$, and we will show that the $\sqrt{n}$-consistent rate still holds under our setting. Define the Robinson's estimator $\hat\bbeta$ as follows,
\begin{equation}
    \label{eq_betahat}
    \hat\bbeta=\Big(\sum_{i=1}^n(Z_i-\tilde Z_i)(Z_i-\tilde Z_i)^{\top}\One_{X_i\in\T_d}\Big)^{-1}\Big(\sum_{i=1}^n(Y_i-\tilde Y_i)(Z_i-\tilde Z_i)^{\top}\One_{X_i\in\T_d}\Big),
\end{equation}
where $\tilde Y_i$ and $\tilde Z_i$ are the kernel estimators of $Y_i$ and $Z_i$, respectively, which are
\begin{equation}
    \label{eq_tilde_YZ}
    \tilde Y_i=\sum_{t=1}^nw_h(X_i,X_t)Y_t, \quad \text{and} \quad \tilde Z_i=\sum_{t=1}^nw_h(X_i,X_t)Z_t.
\end{equation}
To establish the consistency result of the Robinson's estimator $\hat\bbeta$, we shall introduce some regularity conditions on the covariates $Z_i$ of the parametric part in (\ref{eq_model}). Specifically, following \textcite{hardle_asymptotic_1989}, we consider $Z_i$ as random design points and related to $X_i$ in the following way.
\begin{assumption}[Random design]\ \\
    \label{asm_iden}
    (i) Assume that there exists some bounded function $h(\cdot)$: $\mathbb{R}^d \rightarrow \mathbb{R}^l$, such that
    $$Z_{i} = h(X_i)+u_i,$$
    where $h(\cdot)$ is Lipschitz continuous on $\T_d$ and $u_i$ are centered random vectors in $\mathbb{R}^l$ independent of $X_i$. The minimum eigenvalue of the covariance matrix $$\Sigma_u = \EE(u_iu_i^{\top})$$ 
    is lower bounded, that is  $\lambda_{\text{min}}(\Sigma_u) \ge c$, for some constant $c>0$. \\
    (ii) Assume that for each $1\le i\le n$ and $1\le j\le l$, $0<\|u_{ij}\|_p<\infty$, for some constant $p\ge4$, where $u_{ij}$ is the $j$-th element of $u_i$.  \\
    (iii) We suppose that for some measurable function $f$, 
    $$u_i=f(\ldots,e_{i-1},e_i), \quad \text{and} \quad u_{i,\{i-k\}}=f(\ldots,e_{i-k-1},e_{i-k}',e_{i-k+1},\ldots,e_i),$$
    where $e_i$ are i.i.d. random vectors in $\RR^l$ independent of $X_i$ and $\epsilon_i$, and $e_i'$ is an i.i.d. copy of $e_i$. We assume that for $p\ge 4$,
    $$\sum_{k\geq 0} \Big\| \max_{1\le j\le l}\big|u_{ij}- u_{ij,\{i-k\}}\big|\Big\|_p <\infty,$$ 
    where $u_{ij,\{k\}}$ is the $j$-th element of $u_{i,\{k\}}$. 
\end{assumption}

Assumption \ref{asm_iden} can be considered as a generalization of the conditions with weaker requirements compared to the ones introduced in \textcite{hardle_asymptotic_1989}, \textcite{gao_convergence_1995} and \textcite{Sun}, where they assumed $u_i$ to be i.i.d. random vectors, while in this paper, we allow $u_i$ to be weakly dependent over $i$. In fact, Assumption \ref{asm_iden}(iii) implies the short-range dependence (SRD) of $u_i$. The following proposition asserts that, under such dependence conditions, we can consistently estimate $\bbeta$ in the time series regression model (\ref{eq_model}) by the Robinson's estimator at the rate of $\sqrt{n}$.
\begin{proposition}[Consistency of $\hat\bbeta$]
    \label{prop1}
    Under Assumptions \ref{asm_sigma}-\ref{asm_iden}, if $h^4n\rightarrow0$, we have
    $$|\hat\bbeta-\bbeta|_{\infty}=O_{\PP}\big\{1/\sqrt{n}\big\}.$$
\end{proposition}

Concerning the conditional volatility function $\sigma(X_i)$ in (\ref{eq_model}), for each $1\le i\le n$, we define the local constant estimator
\begin{equation}
    \label{eq_sigmax_est}
    \hat\sigma^2(X_i)=\sum_{t=1}^nw_h(X_i,X_t)(Y_t-Z_t^{\top}\hat\bbeta  - \hat\mu^*(X_i))^2.
\end{equation}
Note that the bandwidth parameter $h$ in (\ref{eq_sigmax_est}) can be selected differently from the one used in the estimator for $\mu(\cdot)$. Similar methods have been applied in the existing studies; see, for example, 
\textcite{fan_local_1996,zhao_kernel_2006}. In this paper, we adopt the same bandwidth parameters in both $\hat\mu^*(\cdot)$ and $\hat\sigma(\cdot)$ for brevity. The consistency result of $\hat\sigma(\cdot)$ is stated as follows.
\begin{proposition}[Consistency of $\hat\sigma$]
    \label{prop2}
    Assume that the conditions in Theorem \ref{thm1_GA} hold. Then,
    we have
    $$\sup_{x\in\T_d}|\hat\sigma^2(x)-\sigma^2(x)|=O_{\PP}\Big\{h+\frac{1}{n}+\sqrt{\frac{\log(n)}{h^dn}}\Big\}.$$
\end{proposition}
Recall that in Theorem \ref{thm1_GA}, $\Z_t$, $t\in\T_d$, is a Gaussian random field with mean zero and conditional covariance matrix $Q^{(L)}$ as defined in (\ref{eq_cov_Z}). In practice, one shall estimate $Q^{(L)}$ by $\hat Q^{(L)}=(\hat Q_{j,j'}^{(L)})_{1\le j,j'\le n}$, that is 
\begin{equation}
    \label{eq_Q_hat}
    \hat Q_{j,j'}^{(L)}=h^dn\sum_{k=1-L}^{L-1}\sum_{i=1\vee(1-k)}^{n\wedge (n-k)}\hat c_{j,j',i,k}w_h(X_j,X_i)w_h(X_{j'},X_{i+k})\hat\gamma(k),
\end{equation}
where $\hat c_{j,j',i,k} = \hat\sigma^{-1}(X_j)\hat\sigma^{-1}(X_{j'})\hat\sigma(X_i)\hat\sigma(X_{i+k})$, and $\hat\gamma(k)$ is a consistent estimate of $\gamma(k)$ which takes the form
\begin{equation*}
    \hat\gamma(k)=\sum_{i=1}^{n-k}\hat\epsilon_i\hat\epsilon_{i+k}/n,
\end{equation*}
with $\hat\epsilon_i=(Y_i-Z_i^{\top}\hat\bbeta-\hat\mu^*(X_i))/\hat\sigma(X_i)$. The precision of the estimated autocovariance $\hat\gamma(k)$ has been well investigated in the literature; see, for example, \textcite{shumway}. Also, note that we adopt the estimator $\hat Q_{j,j'}^{(L)}$ which includes the term $\hat c_{j,j',i,k}$ as the estimator of $c_{j,j',i,k}$ appeared in the true $Q_{t,s}^{(L)}$ defined in (\ref{eq_cov_Z}). Since $\hat\sigma(\cdot)$ is a consistent estimator of $\sigma(\cdot)$ by Proposition \ref{prop2} and $\sigma(\cdot)$ is bounded on the support $\T_d$ according to Assumption \ref{asm_sigma}, it can be shown that $\hat Q_{j,j'}^{(L)}$ is a consistent estimator of $Q_{j,j'}^{(L)}$. We state this result and provide the exact consistency rate as follows.

\begin{proposition}[Precision of long-run covariance matrix]
    \label{prop_longrun}
    Suppose that Assumptions \ref{asm_sigma}-\ref{asm_iden} hold. Then, for some large positive integer $L\rightarrow\infty$ satisfying $hL^2/n\rightarrow0$, $L^3/n^2\rightarrow0$ and $L^4\log(n)/(h^dn^3)\rightarrow0$, we have
    $$\max_{1\le j,j'\le n}\big|\hat Q_{j,j'}^{(L)} - Q_{j,j'}\big| = O_{\PP}\Big\{\frac{L^2}{n}\Big(h+\sqrt{\log(n)/h^dn} + 1/\sqrt{n}\Big) + L/n + L^{-\zeta}\Big\},$$
    where the constant $\zeta>1$ is defined in Assumption \ref{asm_dep_epsilon}.
\end{proposition}
In real-data applications, one can simply take $L=\sqrt{n}$ to obtain a consistent estimator $\hat Q^{(L)} =(\hat Q_{j_1,j_2}^{(L)})_{1\le j_1,j_2\le n}$ to estimate the covariance matrix of the statistic $\sqrt{h^dn}|\hat\mu^*(x)-\mu(x)|/\hat\sigma(x)$ conditioned on the covariates $X_i$, $1\le i\le n$. In next section, we shall provide the Gaussian approximation for the test statistic $\sup_{x\in\T_d}|\hat\mu^*(x)-\mu(x)|/\hat\sigma(x)$ with the estimated covariance matrix applied.

\subsection{Gaussian Approximation for Implementation}
Provided the consistent estimated long-run covariance matrix $\hat Q^{(L)}$ in Proposition \ref{prop_longrun}, we now introduce a sequence of Gaussian random variables $\hat\Z_j$, $1\le j\le n$, with mean zero and the conditional covariance matrix $\hat Q^{(L)}$, i.e., $\hat Q_{j,j'}^{(L)} = \EE(\hat\Z_j\hat\Z_{j'})$, for each $1\le j,j'\le n$. We shall show that the Gaussian approximation result in Theorem \ref{thm1_GA} still holds when we replace $Q_{t,s}^{(L)}$ therein by its estimator $\hat Q_{j,j'}^{(L)}$. In addition, Propositions \ref{prop1} and \ref{prop2} ensure that even when $\bbeta$ and $\sigma(x)$ in model (\ref{eq_model}) are both unknown, one shall still achieve the similar asymptotic distribution in Theorem \ref{thm1_GA} for the test statistic $\sup_{x\in\T_d}\big|\hat\mu^*(x)- \mu(x)\big|/\hat\sigma(x)$ by employing the proposed estimators $\hat\bbeta$ and $\hat\sigma(x)$. As a result, combining Propositions \ref{prop1}--\ref{prop_longrun} facilitates the construction of SCR for practical use. For the theoretical guarantee, we show the Gaussian approximation result with $\hat\bbeta$, $\hat\sigma(x)$ and $\hat Q^{(L)}$ plugged in as follows.
\begin{theorem}[Gaussian approximation for implementation]
    \label{thm2_GA}
    Under the conditions in Theorem \ref{thm1_GA} and Propositions \ref{prop1}--\ref{prop_longrun}, for the same $\Delta_1$, $\Delta_2$ and $\Delta_3$ defined in Theorem \ref{thm1_GA} and
    $$\Delta_4'=\Big(hL^2/n+L^2\sqrt{\log(n)/(h^dn^3)} + L/n + L^{-\zeta}\Big)^{1/3}\log^{2/3}(n),$$
    we have
    $$\sup_{u\in\RR}\Big|\PP\big(\sup_{x\in\T_d}\sqrt{h^dn}\big|\hat\mu^*(x)- \mu(x)\big|/\hat\sigma(x)<u\big) - \PP\big(\max_{1\le j\le n}|\hat\Z_j|<u\big) \Big| \lesssim \Delta_1+\Delta_2+\Delta_3+\Delta_4',$$
    where the constants in $\lesssim$ are independent of $n$ and $h$.
    If in addition, the conditions in expression (\ref{thm1_o1}) hold, and
    \begin{align}
        \label{thm2_o1}
        hL^2n^{-1}\log^2(n)\rightarrow0, \quad L^4h^{-2d}n^{-6}\log^5(n)\rightarrow0, \quad Ln^{-1}\log^2(n)\rightarrow0,
    \end{align}
    then we have
    $$\sup_{u\in\RR}\Big|\PP\big(\sup_{x\in\T_d}\sqrt{h^dn}\big|\hat\mu^*(x)- \mu(x)\big|/\hat\sigma(x)<u\big) - \PP\big(\max_{1\le j\le n}|\hat\Z_j|<u\big) \Big| \rightarrow 0.$$
\end{theorem}

\begin{remark}[Comparison of convergence rates in Theorems \ref{thm1_GA} and \ref{thm2_GA}]
    Note that the convergence rate in Theorem \ref{thm2_GA} differs from the one in Theorem \ref{thm1_GA} only in the last part $\Delta_4'$ (appeared as $\Delta_4$ in Theorem \ref{thm1_GA}). This difference results from the replacement of the covariance matrix $Q^{(L)}$ by its consistent estimator $\hat Q^{(L)}$, which leads to an additional approximation error due to the Gaussian comparison as shown in Lemma \ref{lemma_comparison} in the Supplementary Material. The other approximation errors introduced by the estimators $\hat\bbeta$ and $\hat\sigma(x)$ can be shown to be negligible compared to $\Delta_1$--$\Delta_3$ and $\Delta_4'$. In particular, when we let $L=\sqrt{n}$, the convergence rates in Theorems \ref{thm1_GA} and \ref{thm2_GA}
    are of the same magnitude up to a logarithmic factor. 
We refer to a rigorous proof of these convergence rates to the Supplementary Material. Based on Theorem \ref{thm2_GA}, we can construct the $(1-\alpha)$-th percentile SCR defined in (\ref{eq_ln_rn}) by using the $(1-\alpha)$-th quantile of $\max_{1\le j\le n}|\hat\Z_j|/\sqrt{h^dn}$.
\end{remark}

\begin{remark}[The necessity of partially linear structure]
Generally, the linear part in (\ref{eq_model}) is derived out of its underlying theories, instead of being added to the model in an ad-hoc fashion; see, for example, the demand equation in \textcite{kim_simultaneous_2021} and the forward premium regression in (\ref{fama}). Secondly, the purely nonparametric approach
suffers from the curse of dimensionality with a large number of model predictors, while our partially linear model is obviously more efficient in dealing with the issue  (\cite{wang_conditional_2016}) than the purely nonparametric counterpart. Thirdly, it is nontrivial to extend the asymptotic theory of the purely nonparametric model to the partially linear case, because $Z_i$ in (\ref{eq_model}) is a function of $X_i$ plus a random noise $u_i$ (i.e. Assumption \ref{asm_iden}). As a consequence, establishing the consistency results for $\hat\bbeta$ and $\hat\sigma$ is different than in the purely nonparametric case.  
\end{remark}

\subsection{SCR Construction}\label{subsec_scr_construct}
Recall the upper and lower bounds of the simultaneous confidence region, i.e., $\hat l_n(x)$ and $\hat r_n(x)$ defined in (\ref{eq_ln_rn}). For the convenience of implementation, we provide the detailed steps of our proposed method for the construction of SCR as follows. Suppose that we have observed a sample $(Z_i,X_i,Y_i)$, $1\le i\le n$. Given the significance level $\alpha\in(0,1)$, we aim to construct the $(1-\alpha)$-th percentile SCR for $\mu(\cdot)$ in model (\ref{eq_model}).
\begin{itemize}
    \item \textbf{Step 1.} Calculate the weight function $w_h(x,X_t)$ at each covariate $X_i$:
    $$w_h(X_i,X_t) = \frac{K_h(X_i-X_t)}{\sum_{t=1}^nK_h(X_i-X_t)},$$
    where $K_h(x-X_t)=K\big((x-X_t)/h\big)/h$. One can choose the kernel function $K(\cdot)$ to be any commonly used kernels, such as the Gaussian kernel $K(u)=(2\pi)^{-1/2}e^{-u^2/2}$ and the Epanechnikov kernel $K(u)=3/4(1-u^2)\One_{|u|\le1}$.
    \item \textbf{Step 2.} Estimate the coefficients $\bbeta$ of the linear part: 
    $$\hat\bbeta= \Big(\sum_{i=1}^n(Z_i-\tilde Z_i)^2\Big)^{-1}\Big(\sum_{i=1}^n(Y_i-\tilde Y_i)(Z_i-\tilde Z_i)\Big),$$
    where $\tilde Y_i=\sum_{t=1}^nw_h(X_i,X_{t})Y_{t}$, and $\tilde Z_i=\sum_{t=1}^nw_h(X_i,X_{t})Z_{t}$.
    \item \textbf{Step 3.} Calculate the point estimator of the nonparametric function $\mu(\cdot)$ at each covariate $X_i$:
    $$\hat\mu^*(X_i) = \sum_{t=1}^nw_h(X_i,X_t)(Y_t-Z_t^{\top}\hat\bbeta).$$
    \item \textbf{Step 4.} Estimate the conditional variance/volatility function $\sigma(\cdot)$ at each covariate $X_i$:
    $$\hat\sigma^2(X_i)=\sum_{t=1}^nw_h(X_i,X_t)(Y_t-Z_t^{\top}\hat\bbeta  - \hat\mu^*(X_i))^2.$$

    \item \textbf{Step 5.} Let $N=1000$ (or some other large positive number). Apply the multiplier bootstrap to generate a sequence of centered Gaussian random variables $\hat\Z_j$, $1\le j\le N$, with covariance matrix $\hat Q^{(L)}=(\hat Q_{j,j'}^{(L)})_{1\le j,j'\le N}$, where
    $$\hat Q_{j,j'}^{(L)}=h^dn\sum_{j-j'=-L}^L\sum_{i=1}^nc_{j,j',i}\hat\gamma(|j-j'|),$$
    with $c_{j,j',i}=w_h(X_j,X_i)w_h(X_{j'},X_{i+j-j'})$, $\hat\gamma(k)=\sum_{i=1}^{n-k}\hat\epsilon_i\hat\epsilon_{i+k}/n$, and $\hat\epsilon_i=(Y_i-Z_i^{\top}\hat\bbeta-\hat\mu^*(X_i))/\hat\sigma(X_i)$.
    Regarding the choice of $L$, see Proposition \ref{prop_longrun}, and we can choose $L$ to be some large positive integer such as $L=\sqrt{n}$. In some special case such as MA(2), one can simply choose $L$ to be the maximum lag, i.e., $L=2$.
    \item \textbf{Step 6.} Construct the $(1-\alpha)$-th percentile SCR for $\mu(\cdot)$:
    \begin{equation}
        \label{eq_SCR_final}
        \hat\mu^*(X_i)-\hat q_{\alpha}\hat\sigma(X_i) \le \mu(X_i)\le \hat\mu^*(X_i) + \hat q_{\alpha}\hat\sigma(X_i), \quad 1\le i\le n,
    \end{equation}
    where $\hat q_{\alpha}$ is the $(1-\alpha)$-th empirical quantile of $\max_{1\leq j\leq n}|\hat\Z_j|/\sqrt{h^dn}$.
\end{itemize}
\begin{remark}[Dynamic width of SCR]
Our proposed SCR in (\ref{eq_SCR_final}) takes into account the conditional heteroskedasticity in model (\ref{eq_model}). The conditional heteroskedasticity allows the conditional variance of the error to be time-varying, and also accounts for the dependence between the covariate $X_i$ and the model error $\sigma(X_i)\epsilon_i$ (\cite{zhao_inference_2013}). Additionally, the conditional heteroskedasticity allows the width of the confidence band to vary over the covariate term, and thus, the relative width can indicate the different amounts of information used for different regions of the band. For example, the lack of information at each tail of the curve, due to fewer observations, is implied by the relatively wide confidence intervals around the tails. This dynamic-width SCR is adopted by \textcite{kim_simultaneous_2021} for the i.i.d. data. We extend it to the time series setting, so that one can accommodates macroeconomic and/or financial time series models such as the forward premium regression; see Section \ref{sec_app}. 
\end{remark}

\section{Simulation Studies}\label{sec_simul}
We present the simulation study with four different time series models in this section to illustrate the performance of our proposed SCR. Consider the following data-generating process:
\begin{eqnarray}\label{smodel}
y_i=\beta z_i+\mu(x_i)+\sigma(x_i)\epsilon_i,~~~~~~~~i=1,\cdots,n
\end{eqnarray}
where $\beta=0.5$, $\mu(x)=0.3+0.4\,x$ and $\sigma(x)=\big(0.1+0.1\,x^2\big)^{1/2}$. We shall compute the {\it coverage probability} of the SCR for $\mu(\cdot)$ in (\ref{smodel}) to evaluate our method. We set the covariate terms $x_i$ and $z_i$ to be
\begin{eqnarray*}
x_i&=&\sum_{k=0}^{\infty}0.1^k\,\delta_{i-k}\\
z_i&=&0.2+0.4 x_i+u_i,~~~~~~~~i=1,\cdots,n,
\end{eqnarray*}
where $\delta_i$ and $u_i$ are both i.i.d. random variables following the standard normal distribution, and $\delta_i$ are independent of $u_i$. Note that this setting is in accordance with Assumption \ref{asm_iden}. For the specifications of $\epsilon_i$ in (\ref{smodel}), we consider four commonly used models: (i) i.i.d. standard normal, (ii) an autoregressive (AR) process, (iii) a moving-average (MA) process and (iv) an autoregressive moving-average (ARMA) process, which take the forms of:
\begin{eqnarray*}
&&(i)\,\,\mbox{Std. Normal}:\,\epsilon_i=\eta_i,\\
&&(ii)\,\,\mbox{AR(1)}:\,\epsilon_i=0.1\,\epsilon_{i-1}+\eta_i,\\
&&(iii)\,\,\mbox{MA(1)}:\,\epsilon_i=\eta_i+0.2\eta_{i-1},\\
&&(iv)\,\,\mbox{ARMA(1,1)}:\,\epsilon_i=0.1\,\epsilon_{i-1}+\eta_i+0.2\,\eta_{i-1},~~~~~~~~i=1,\cdots,n
\end{eqnarray*} 
where $\eta_i$ are i.i.d. random variables following the standard normal distribution. For a range of bandwidths (see Table \ref{tbl_cov_prob}), we follow Steps 1--6, the procedures described in the previous section, to construct the 95\% SCR of $\mu(\cdot)$ in (\ref{smodel}). The process is repeated $M$ times for each chosen bandwidth parameter. Given these $M$ different SCRs, we count how many of those SCRs contain the true $\mu(\cdot)$ in them, which gives us the desired coverage probabilities. Here we let $M=1,000$ (i.e. number of iterations) and $n=200$ (i.e. sample size). 
The simulation results under (i)--(iv) are summarized by Table \ref{tbl_cov_prob}.  

\begin{table}[tbp]
\centering
\caption[] {Coverage probability of 95\% SCR for $\mu(\cdot)$} 
\label{tbl_cov_prob}
\renewcommand{\arraystretch}{1.3} 
\begin{tabular}{cccccccccc}
\hline\hline
bandwidth && \multicolumn{1}{c}{(i) Std. Normal} && \multicolumn{1}{c}{(ii) AR(1)}  && \multicolumn{1}{c}{(iii) MA(1)}  && \multicolumn{1}{c}{(iv) ARMA(1,1)} &  \\ \hline 
0.30      && 0.874      && 0.858   && 0.860   && 0.852    &  \\ 
0.32      && 0.898      && 0.890   && 0.884   && 0.888    &  \\ 
0.34      && 0.918      && 0.898   && 0.920   && 0.902    &  \\ 
0.36      && 0.938      && 0.914   && 0.932   && 0.918    &  \\ 
0.38      && 0.948      && 0.924   && 0.934   && 0.936    &  \\ 
0.40      && 0.950      && 0.934   && 0.942   && 0.948    &  \\ 
0.42      && 0.958      && 0.950   && 0.956   && 0.954    &  \\ 
0.44      && 0.963      && 0.951   && 0.958   && 0.958    &  \\ 
0.46      && 0.970      && 0.960   && 0.960   && 0.960    &  \\ 
0.48      && 0.970      && 0.958   && 0.964   && 0.964    &  \\ 
0.50      && 0.972      && 0.964   && 0.964   && 0.964    &  \\ \hline
\end{tabular}
\end{table}

From Table \ref{tbl_cov_prob}, we see that the coverage probabilities under various bandwidths are reasonably close to $0.95$, which is the nominal coverage rate at a 5 percent level. For some bandwidths, we have the coverage probabilities that are almost identical to the nominal level $95\%$. However, as the bandwidth gets extreme, the coverage probability deviates from the nominal $0.95$, as we can observe from Table \ref{tbl_cov_prob}.

\section{Application: Forward Premium Regression}\label{sec_app}

One of the useful applications for (\ref{eq_model}) in time series analysis is the {\it forward premium anomaly}. Consider the celebrated forward premium regression (\cite{Fm:1984}):
\begin{equation}\label{fama}
s_{t+1}-s_t=\mu +\beta\cdot(f_{1,t}-s_{t})+u_{t+1},
\end{equation}
where $s_{t}$ is log of monthly {\it spot} exchange rate at time $t$ and $f_{1,t}$ is log of monthly {\it forward} exchange rate with one-month maturity at time $t$, respectively. Here $u_t$ is a model error. A vast amount of literature illustrates that the constant term $\mu$ in (\ref{fama}), the foreign exchange risk premium, is related to {\it fundamentals} (\cite{DH:1985,HS:1986,Hod:1989,Mark:1995,BK:2006,Alv:2009,LRV:2014,BK:2015,KKK:2022}), such as money growth, output growth, interest rates, conditional variance of money growth and so on. Given this, one can possibly model (\ref{fama}) in the following flexible manner:
\begin{equation}\label{fama2}
s_{t+1}-s_t=\mu(x_{1,t},\,x_{2,t}) + \beta\cdot(f_{1,t}-s_{t}) + \sigma(x_{1,t},\,x_{2,t})\,\epsilon_{t+1},
\end{equation}
where
\begin{equation*}
x_{1,t}=m_t-m^*_t,\,\,\,\,\,\,\,x_{2,t}=y_t-y^*_t.
\end{equation*}
Here, $m_t$ and $y_t$ are log of domestic money stock and log of domestic production, respectively, and $m^*_t$ and $y^*_t$ are their foreign counterparts. Note here that $\mu$ in (\ref{fama}) is replaced by $\mu(x_{1,t},\,x_{2,t})$, and that the conditional heteroscedasticity is introduced in (\ref{fama2}). Interestingly, Assumption \ref{asm_iden} for (\ref{fama2}) implies, quite intuitively, that the forward premium $f_{1,t}-s_t$ does depend on the fundamentals $x_{1,t}$ and $x_{2,t}$.

Although the theory of Uncovered Interest Parity (UIP) in international finance implies that $\mu(\cdot)=0$, numerous empirical studies actually show $\mu(\cdot)\not=0$. We refer to \textcite{FT:1990} and \textcite{Eg:1996} for an excellent review of the literature on this. Recently, \textcite{BDKK:2023} revisited the issue and tested the UIP hypothesis to show the effect of modeling lagged spot returns and lagged forward premiums in (\ref{fama}). Interestingly, (\ref{fama2}) is a specification of our partially linear model framework in (\ref{eq_model}). Hence, the methodology developed in this study can be readily applied to decide whether or not the UIP condition holds for (\ref{fama2}).

The exchange rate data used in this study are monthly spots and 30-day forward exchange rate data for Australian Dollar (AUD), British Pound (GBP) and U.S. Dollar (USD), where USD is the numeraire currency. That is, we employ the AUD/USD and GBP/USD currency pairs in this empirical study. For the money and the output variables, we employ monthly M1 and monthly industrial production index home (i.e. U.S.) and abroad (i.e. Australia and U.K.). Note that the usual GDP series cannot be used here because the sampling frequency is monthly in this application. End-of-month observations from December 1988 through December 2016, including the 2008 financial crisis, are used in this study.

The estimation and inference results, including the simultaneous confidence region, are reported by Figures \ref{fig_aud_mean_vol}--\ref{fig_gbp_y}. First of all, the semi-parametric Robinson estimate of the UIP coefficient $\beta$ in (\ref{fama2}) and its standard error are $\hat\beta_R=0.362$ and $s.e.(\hat\beta_R)=0.753$ for AUD/USD, and $\hat\beta_R=0.802$ and $s.e.(\hat\beta_R)=0.776$ for GBP/USD. In contrast, the traditional OLS estimates of $\beta$ in (\ref{fama}) are $\hat\beta_{OLS}=-0.029$ and $s.e.(\hat\beta_{OLS})=0.718$ for AUD/USD, and $\hat\beta_{OLS}=0.689$ and $s.e.(\hat\beta_{OLS})=0.641$ for GBP/USD. Interestingly, our semi-parametric estimates of the UIP coefficient for both AUD/USD and GBP/USD are closer to one, the value implied by UIP, with lower $t$-statistics than their OLS counterparts.

Figure \ref{fig_aud_mean_vol} shows the estimates of $\hat{\mu}(\cdot,\cdot)$ and $\sigma(\cdot,\cdot)$ (i.e. the nonlinear surfaces) for AUD/USD, while Figure \ref{fig_gbp_mean_vol} shows the mean and volatility estimates for GBP/USD. As we can see from Figures \ref{fig_aud_mean_vol} and \ref{fig_gbp_mean_vol}, the estimates of both the mean and the volatility functions for each currency pair are {\it highly nonlinear}, which implies that (\ref{fama}) may not be an appropriate parametric form for the underlying processes.

Figure \ref{fig_aud_m} represents the two-dimensional relationship between $\mu(\cdot,\cdot)$ and the relative output growth for AUD/USD when the relative money growth is fixed at a certain percentile. Similarly, Figure \ref{fig_aud_y} represents the relationship between $\mu(\cdot,\cdot)$ and the relative money growth for AUD/USD when the relative output growth is fixed at a certain percentile. Each panel in Figures \ref{fig_aud_m} and \ref{fig_aud_y} also includes the 95\% SCR of the corresponding function estimate. The horizontal line represents the null hypothesis of $\mu(\cdot,\cdot)=0$.

In order to not reject the null hypothesis of zero risk premium at a 95\% confidence level, the horizontal line fixed at zero must be {\it entirely} contained by $all$ SCRs in Figures \ref{fig_aud_m} and \ref{fig_aud_y}. Clearly, Panels (a)--(c) in Figure \ref{fig_aud_m} and Panels (a)--(c) in Figure \ref{fig_aud_y} fail to contain the horizontal line entirely. That is, all of the panels in Figures \ref{fig_aud_m} and \ref{fig_aud_y} fail to contain the horizontal line entirely. Hence, the null hypothesis gets rejected at a 5\% level, which is the key implication of Figures \ref{fig_aud_m} and \ref{fig_aud_y}. Even a more general null hypothesis of {\it constant} risk premium is also rejected clearly at a 5 percent level because the SCRs in Figure \ref{fig_aud_y} cannot entirely contain any horizontal line in them.

The main reason for the rejection is because of the tendency that the local-linear estimate of $\mu(\cdot,\cdot)$ changes $nonlinearly$ in both the output growth and the money growth. Similarly, Figures \ref{fig_gbp_m} and \ref{fig_gbp_y} show the results corresponding to the GBP/USD pair. Again, the null hypothesis gets rejected at a 5\% level because the 95\% SCRs in all panels of Figures \ref{fig_gbp_m} and \ref{fig_gbp_y}, except for Panel (b) of Figure \ref{fig_gbp_m}, fail to entirely contain the horizontal line fixed at zero.

Interestingly, these findings highlight the relative advantage of our SCR-based inference because, with the usual $p$-value of a MISE-type test statistic (\cite{HM:1993}), it is rather difficult to figure out why the null gets rejected. Under our approach, however, we can easily locate exactly where the SCR fails to contain the null, so that we can propose an appropriate function for the underlying process more easily.

Intuitively, the trend estimates in this study make sense because either a relative increase in the domestic money supply (i.e. an increase in $x_{1,t}$) or a relative decrease in the domestic production (i.e. a decrease in $x_{2,t}$) will cause the foreign currency to appreciate, yielding a decrease in the risk premium for holding the appreciating currency. Hence the risk premium, $\mu(\cdot,\cdot)$, should increase with a relative increase in the domestic production or a relative decrease in the domestic money supply. This intuition appears to be generally in line with the results in Figures \ref{fig_aud_mean_vol}--\ref{fig_gbp_y}.

\section{Concluding Remarks}
\label{conclusion}				
This paper proposed a new methodology to conduct simultaneous inference of trend in a semi-parametric partially linear time series model and the number of covariate terms in the nonlinear part is allowed to be two or higher. We generalized the high-dimensional Gaussian approximation theory (\cite{chernozhukov_central_2017}) to construct the simultaneous confidence region (SCR) for the multivariate unknown trend in time series. This work can be viewed as an extension of the two-dimensional uniform confidence band (\cite{johnston_probabilities_1982,KHK:2014}) and a generalization to the {\it time series} setting compared to \textcite{kim_simultaneous_2021}. The relating asymptotic properties of the introduced methodology are investigated and the finite-sample properties are studied through a simulation experiment. The developed methodology is applied to the forward premium regression (\cite{Fm:1984}). The empirical analysis based on simultaneous inference confirms that the zero-risk-premium hypothesis for the AUD/USD and GBP/USD currency pairs is rejected at a 5$\%$ level, mainly due to the underlying nonlinear nature, as shown by Figures \ref{fig_aud_mean_vol}--\ref{fig_gbp_y}.    			

The current study  can be extended to the case where all the model parameters are functions of random processes. For instance, the UIP coefficient in (\ref{fama2}) can be modeled as a function of fundamentals as well. Similarly, both the pricing error and the beta coefficient from a standard factor pricing model in finance can be modeled as functions of relevant state variables. The methodology developed in this study can be easily extended to handle such a framework. Additionally, one can conduct simultaneous inference of the unknown function in (\ref{eq_model}) when the covariate terms are non-stationary processes. This would significantly generalize the results in this study. Further insight can be gained by extending the current work in these and other directions, and the authors are currently working on these issues.

\begin{figure}[tbp]
\centering
\begin{tabular}{ccc}
\setlength{\itemsep}{-1.0cm} \psfig{file = 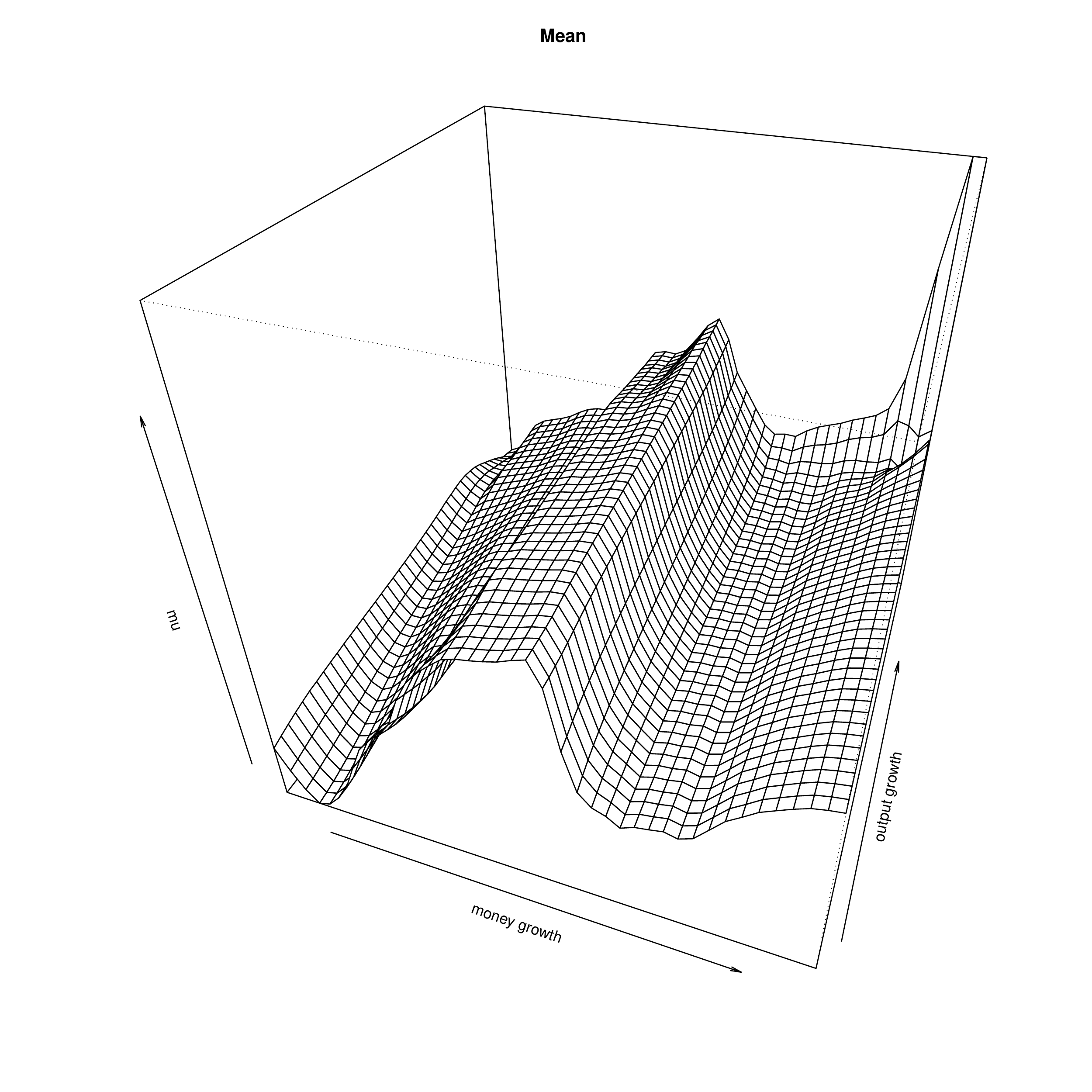, width = 8.0cm,
angle=0} & \psfig{file = 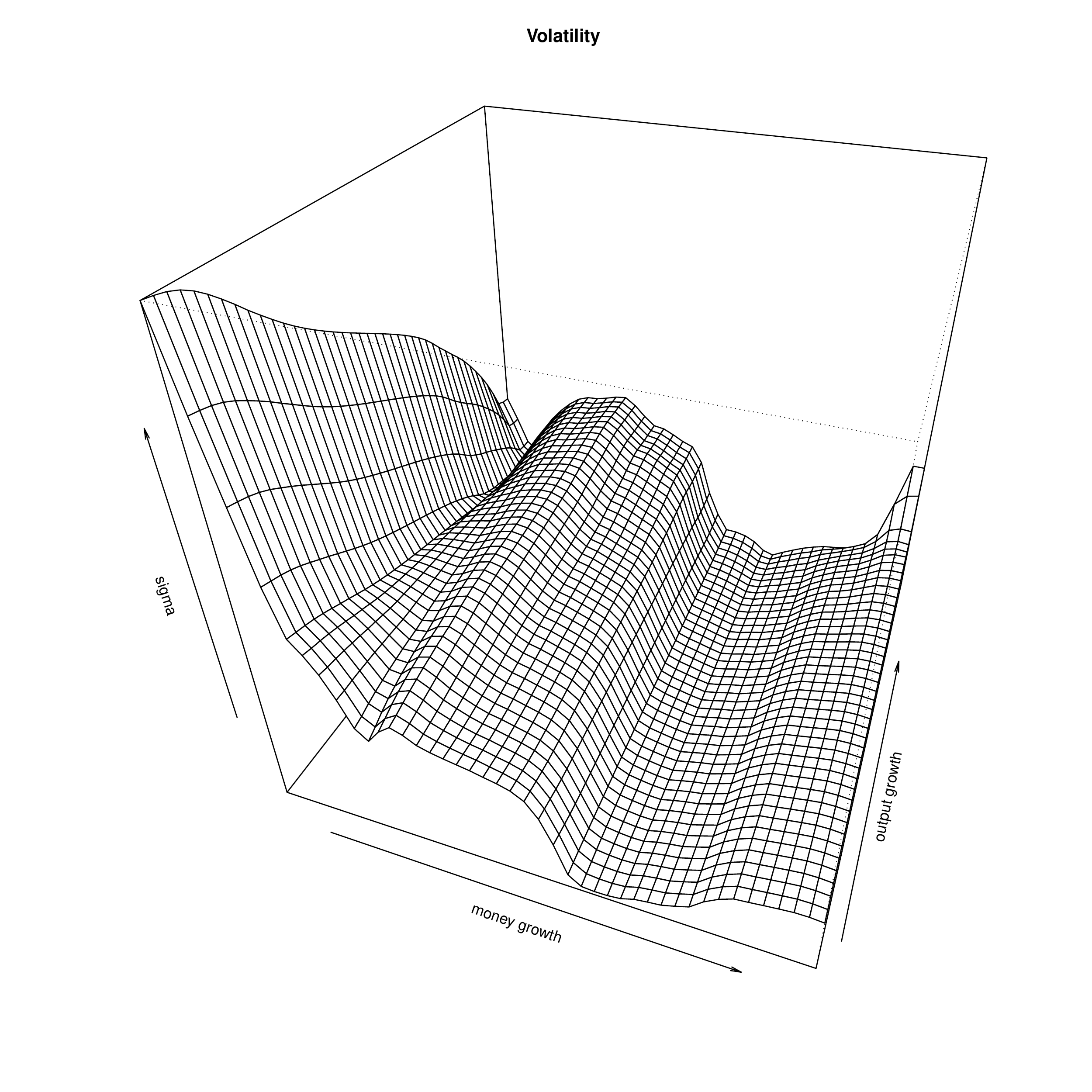, width = 8.0cm, angle=0} &  \\
(a) Conditional mean & (b) Conditional volatility &  \\
&  &
\end{tabular}
{\small {} }
\caption{Conditional mean and volatility in the money growth and the output growth for Australian Dollar (AUD)--US Dollar(USD); The bandwidth is obtained through under-smoothing of the GCV-chosen one.}
\label{fig_aud_mean_vol}
\end{figure}

\begin{figure}[tbp]
\centering
\begin{tabular}{ccc}
\setlength{\itemsep}{-1.0cm} \psfig{file = 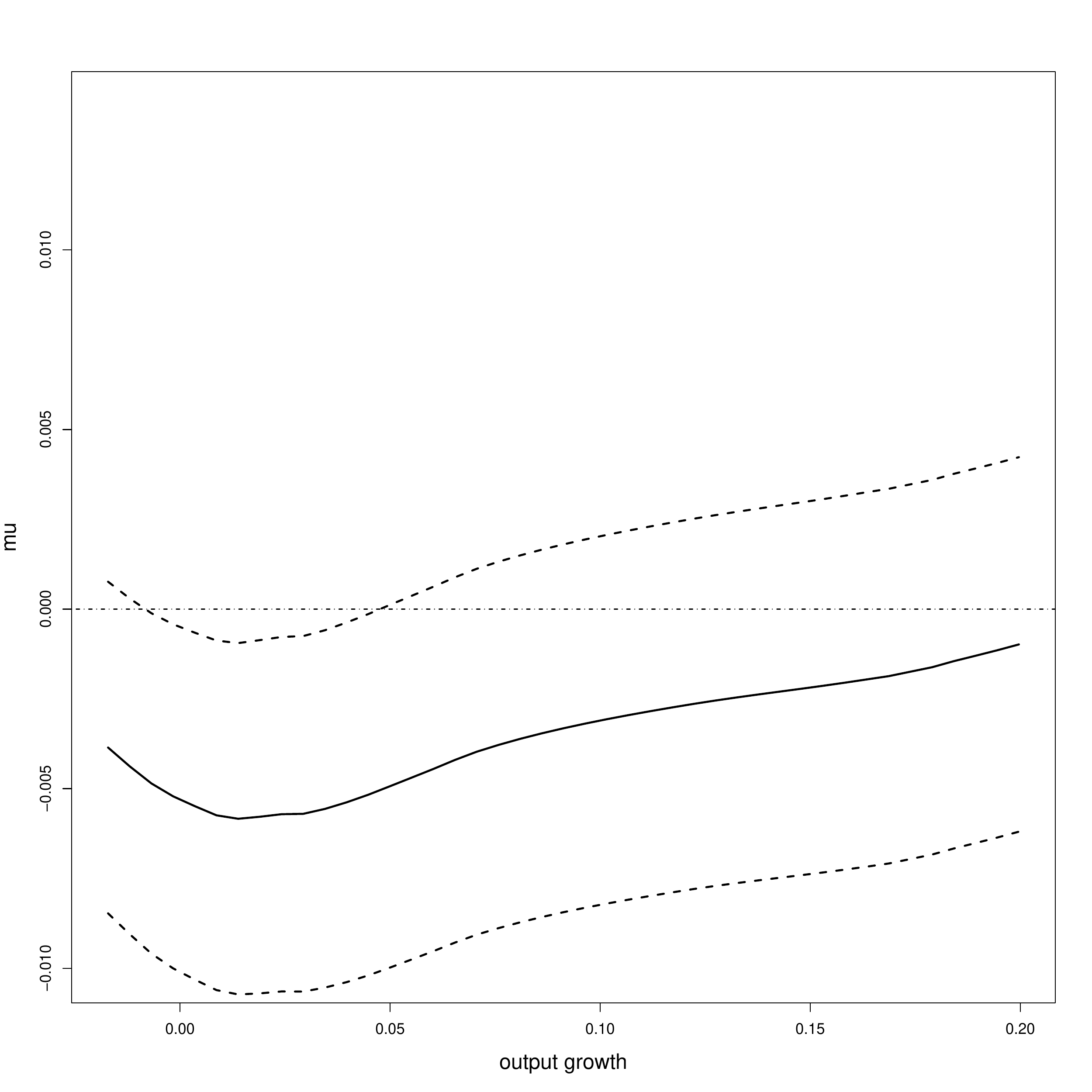, width = 6.0cm,
angle=0} & \psfig{file = 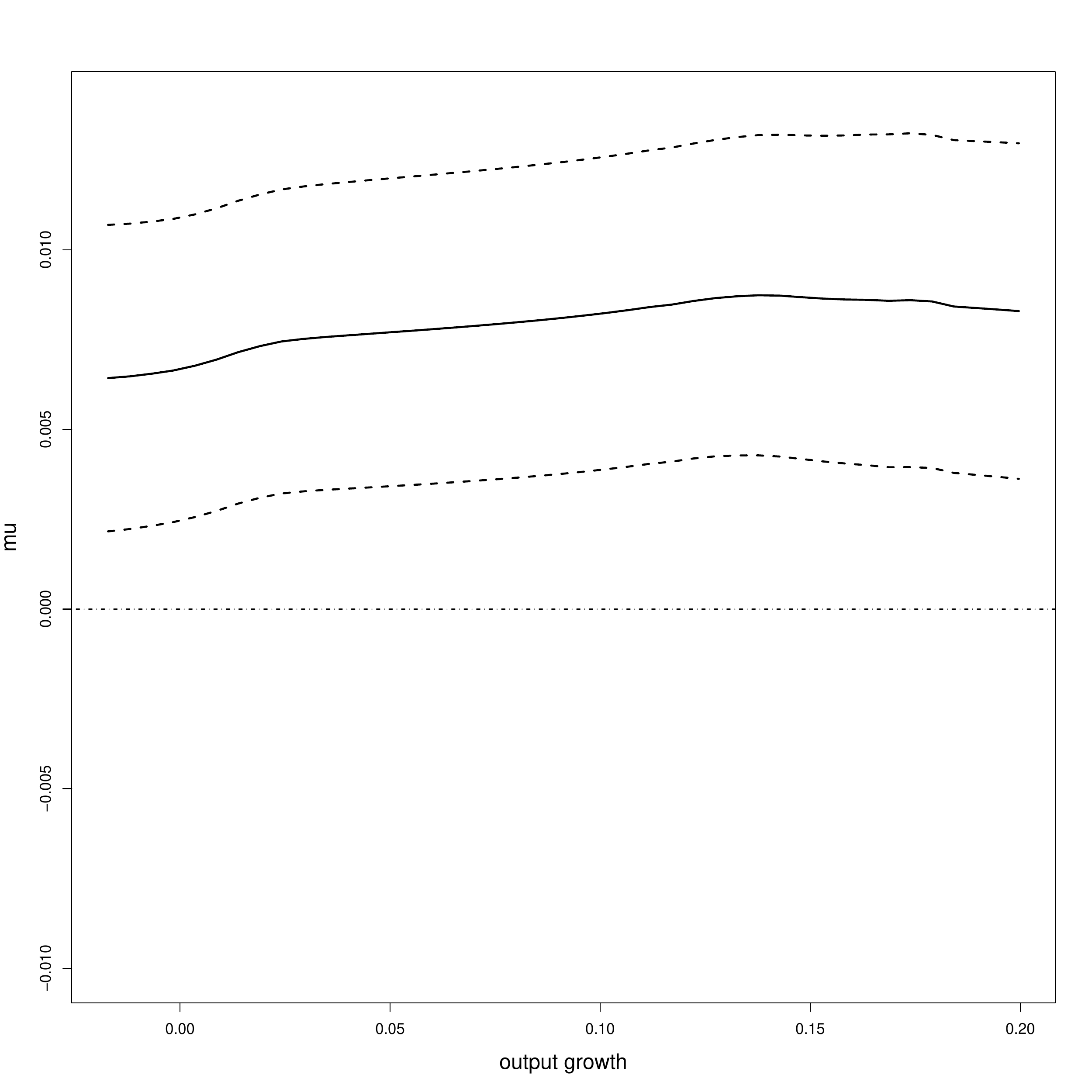, width = 6.0cm, angle=0} &  \\
(a) 25th-percentile of money growth & (b) 50th-percentile of money growth &  \\
\psfig{file = 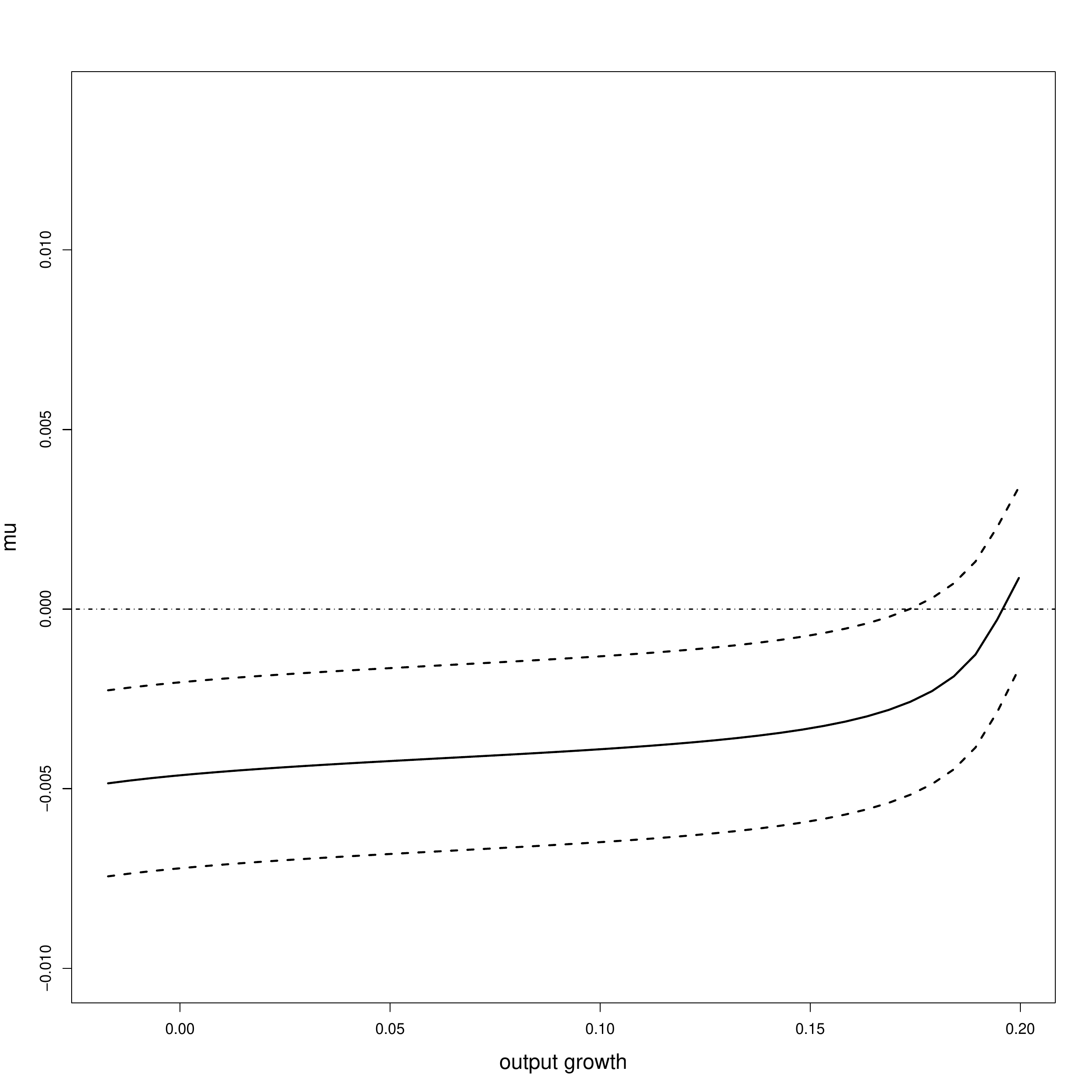, width = 6.0cm, angle=0} &  \\
(c) 75th-percentile of money growth &  \\
&  &
\end{tabular}
{\small {} }
\caption{The $solid$ curve is the local-linear regression estimate of $\mu(\cdot)$ for AUD/USD. The $dashed$ band is the 95\% SCR of $\mu(\cdot)$ and the dotted horizontal line is $H_0:\,\mu(\cdot)=0$. The bandwidth is obtained through under-smoothing of the GCV-chosen one.}
\label{fig_aud_m}
\end{figure}

\begin{figure}[tbp]
\centering
\begin{tabular}{ccc}
\setlength{\itemsep}{-1.0cm} \psfig{file = 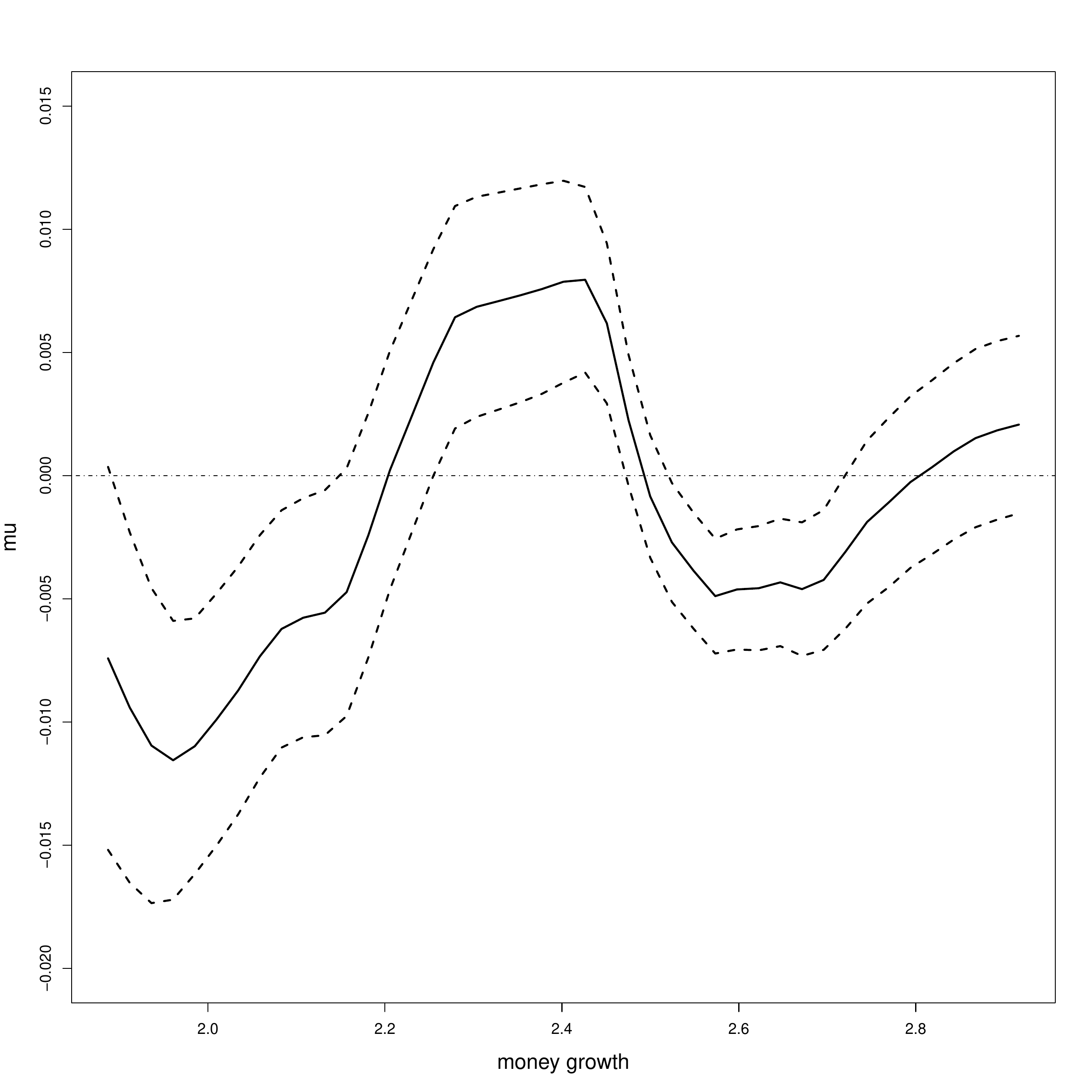, width = 6.0cm,
angle=0} & \psfig{file = 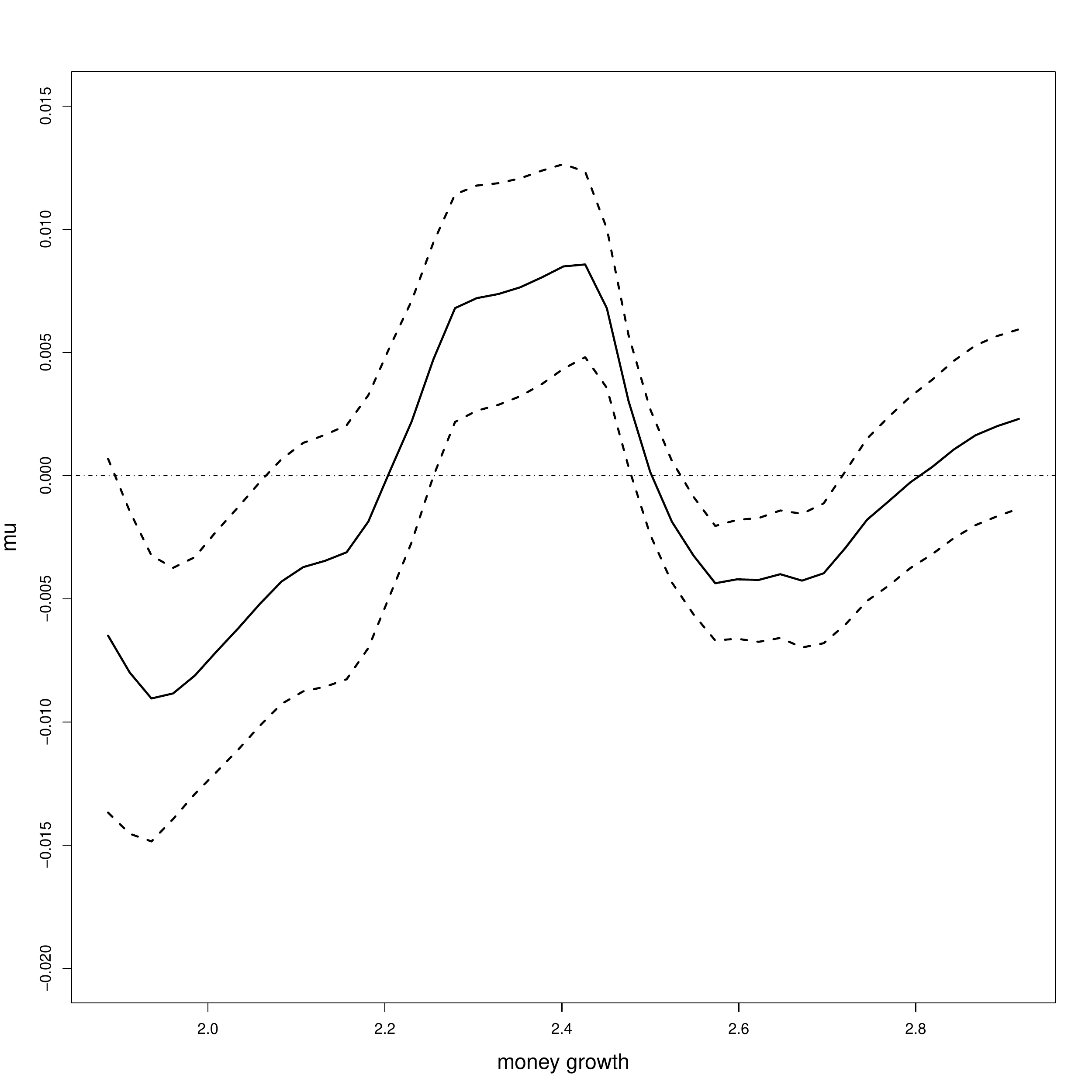, width = 6.0cm, angle=0} &  \\
(a) 25th-percentile of output growth & (b) 50th-percentile of output growth &  \\
\psfig{file = 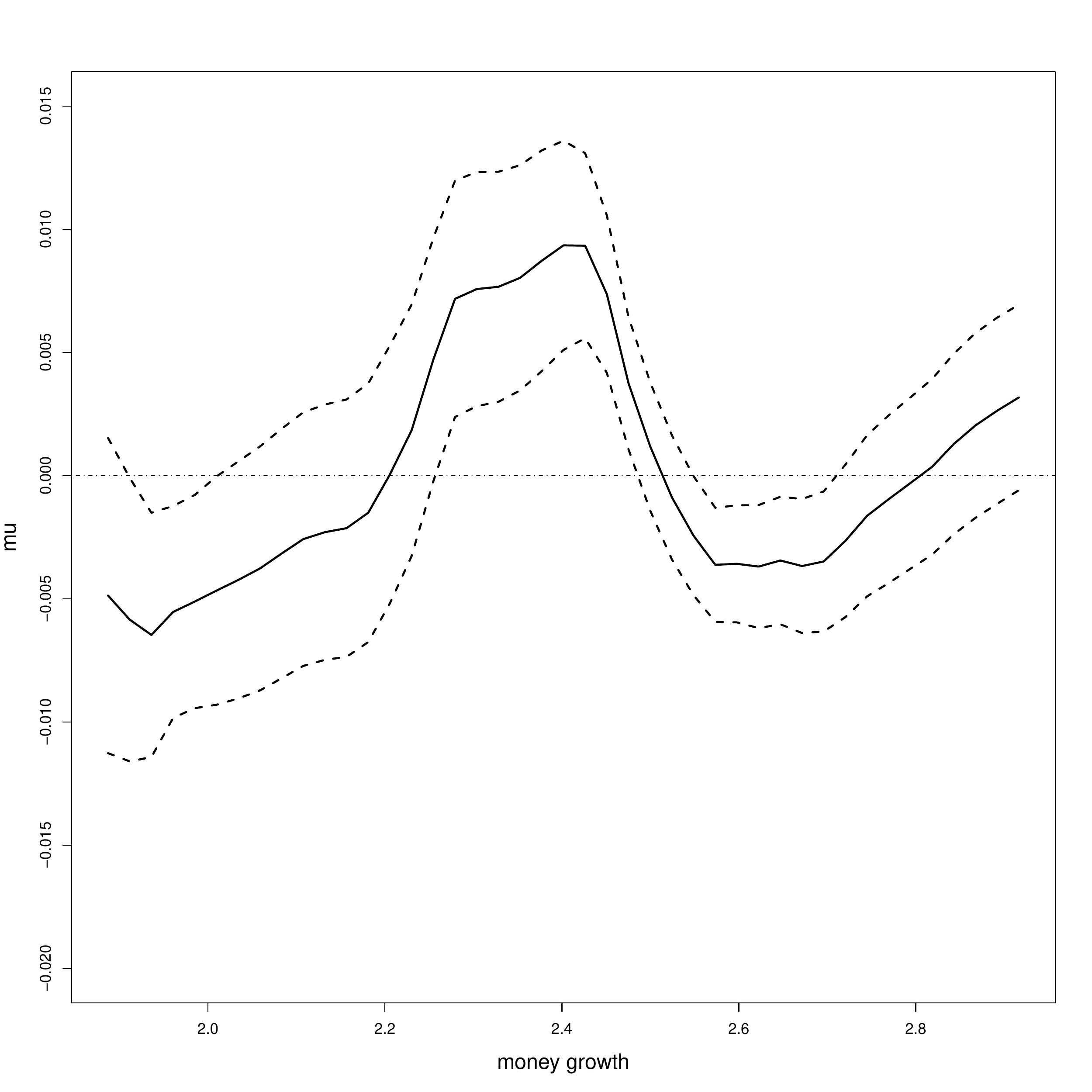, width = 6.0cm, angle=0} &  \\
(c) 75th-percentile of output growth &  \\
&  &
\end{tabular}
{\small {} }
\caption{The $solid$ curve is the local-linear regression estimate of $\mu(\cdot)$ for AUD/USD. The $dashed$ band is the 95\% SCR of $\mu(\cdot)$ and the dotted horizontal line is $H_0:\,\mu(\cdot)=0$. The bandwidth is obtained through under-smoothing of the GCV-chosen one.}
\label{fig_aud_y}
\end{figure}

\begin{figure}[tbp]
\centering
\begin{tabular}{ccc}
\setlength{\itemsep}{-1.0cm} \psfig{file = 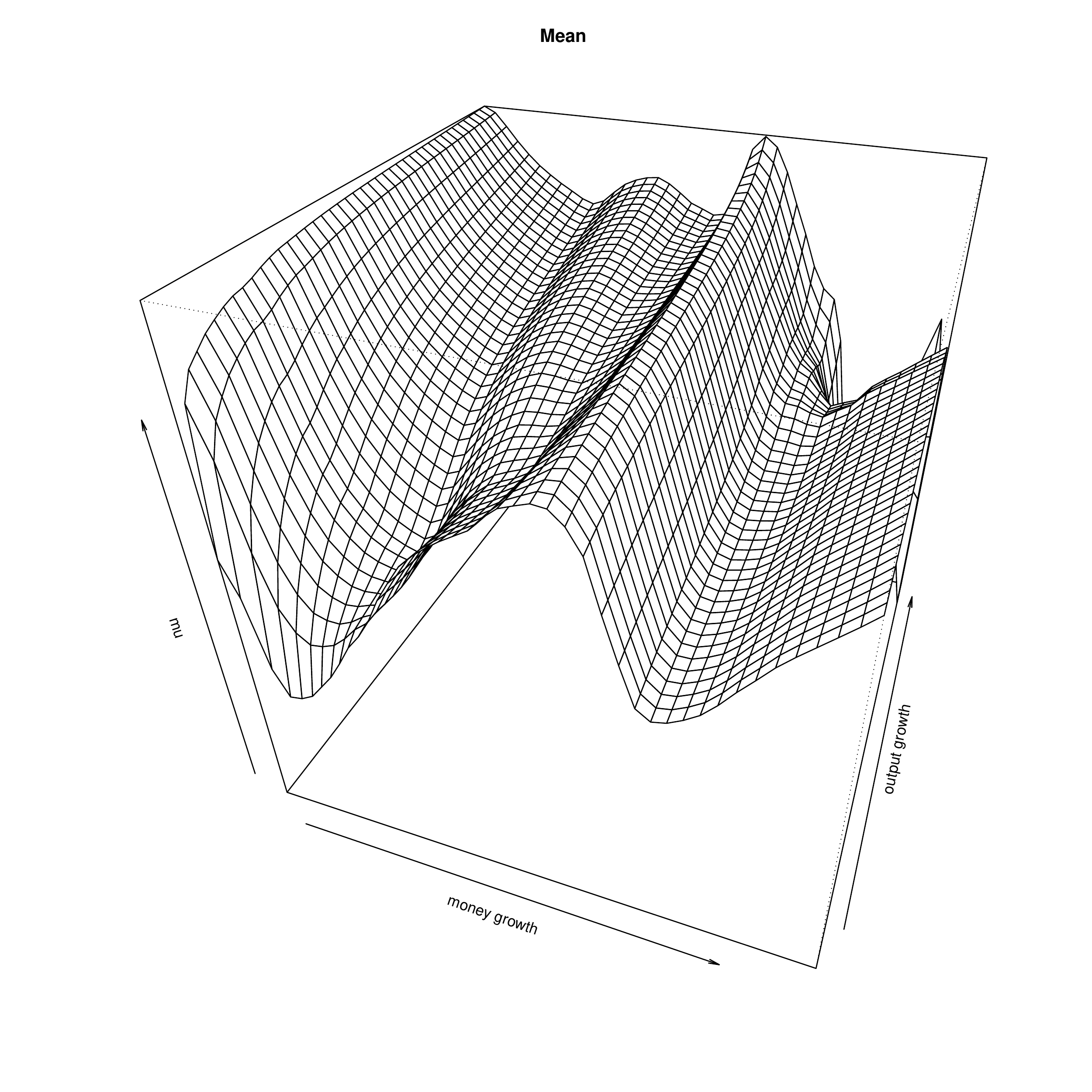, width = 8.0cm,
angle=0} & \psfig{file = 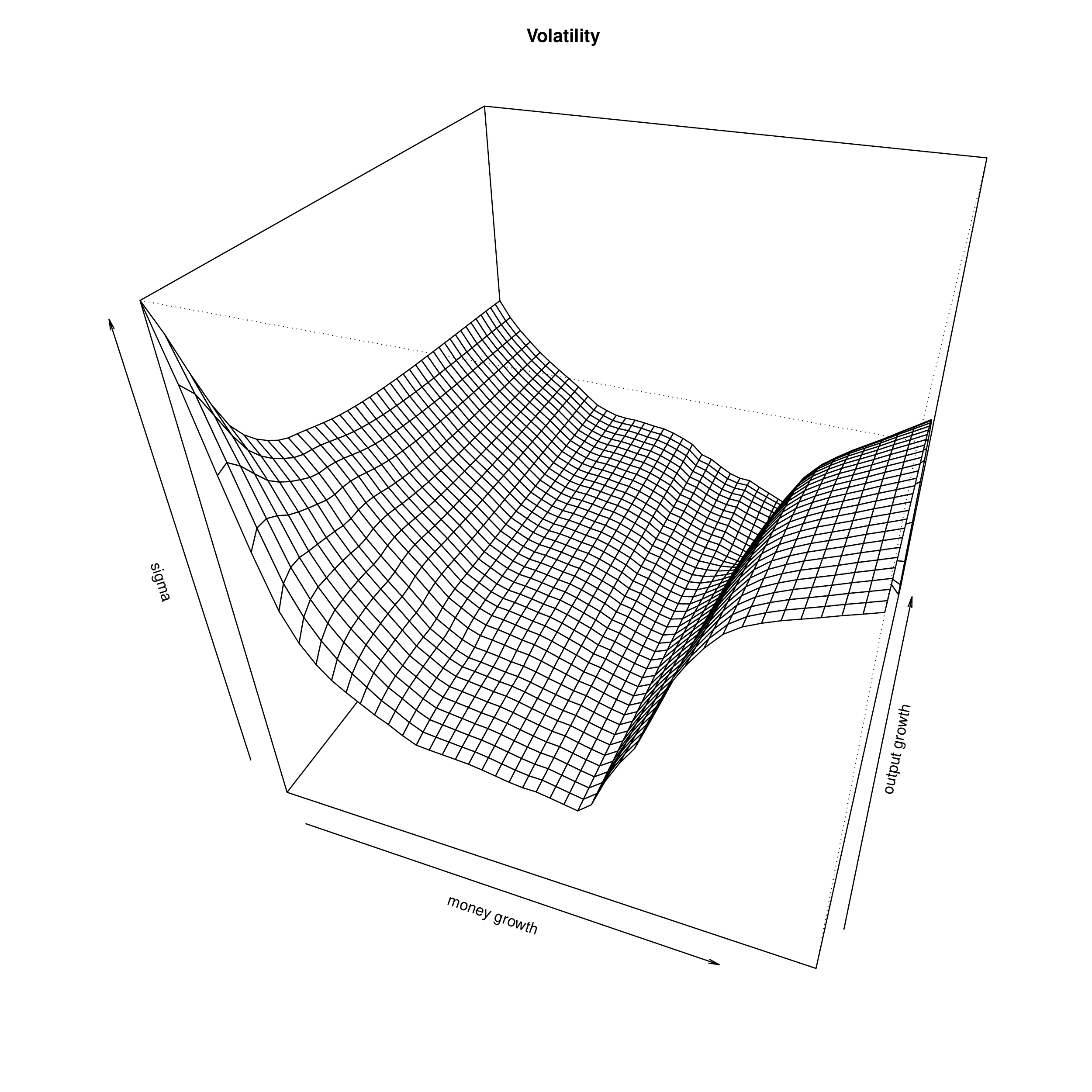, width = 8.0cm, angle=0} &  \\
(a) Conditional mean & (b) Conditional volatility &  \\
&  &
\end{tabular}
{\small {} }
\caption{Conditional mean and volatility in the money growth and the output growth for British Pound (GBP)--US Dollar (USD); The bandwidth is obtained through under-smoothing of the GCV-chosen one.}
\label{fig_gbp_mean_vol}
\end{figure}

\begin{figure}[tbp]
\centering
\begin{tabular}{ccc}
\setlength{\itemsep}{-1.0cm} \psfig{file = 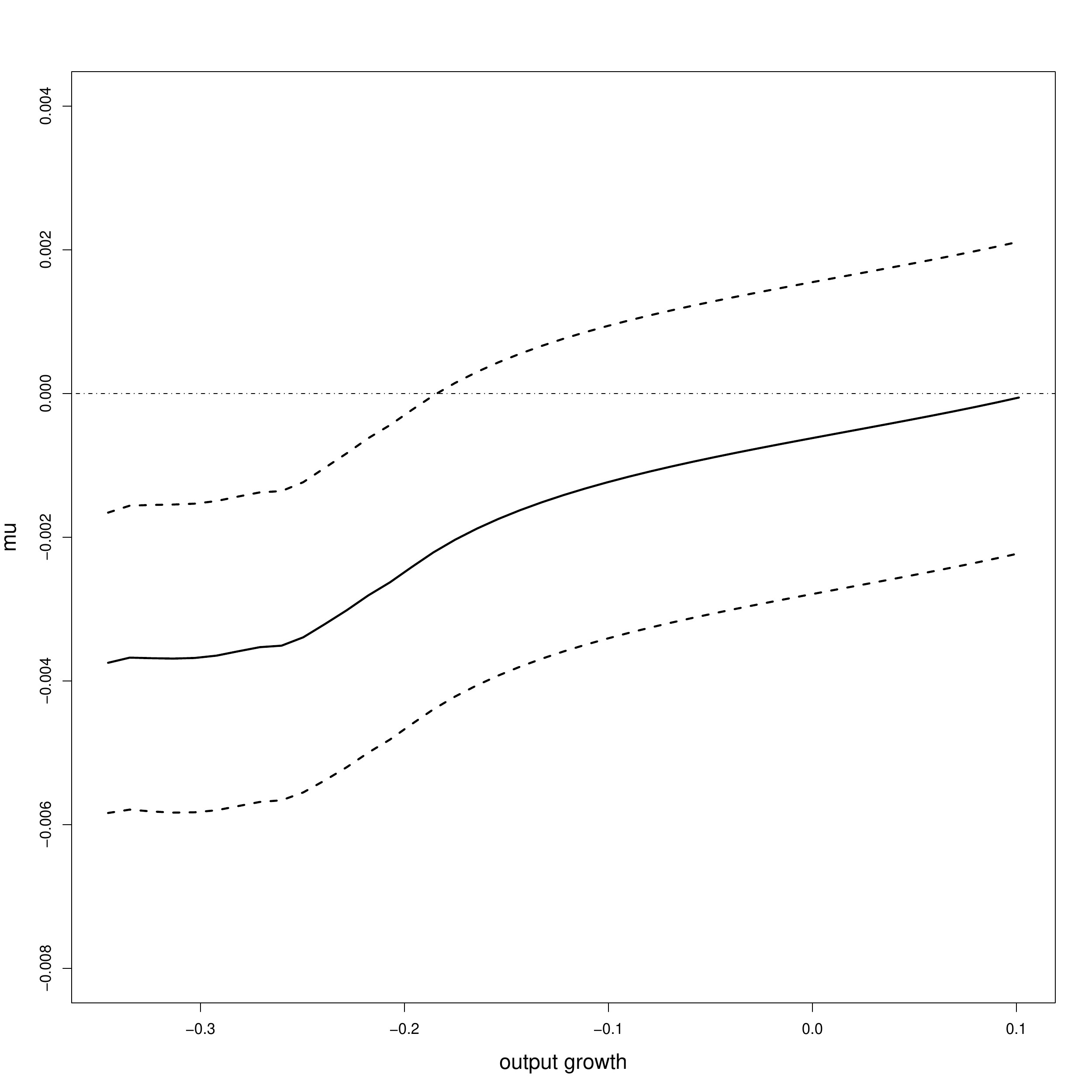, width = 6.0cm,
angle=0} & \psfig{file = 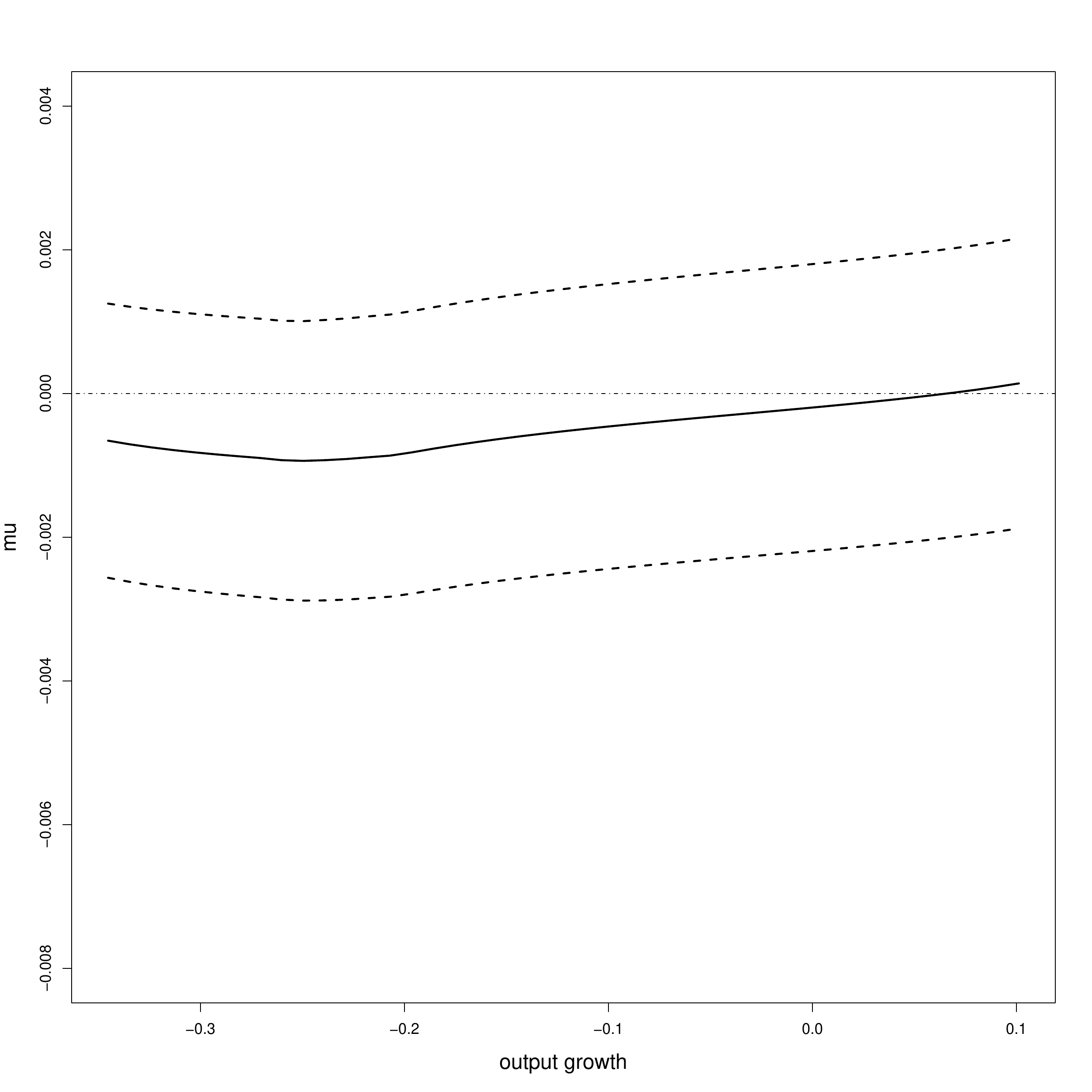, width = 6.0cm, angle=0} &  \\
(a) 25th-percentile of money growth & (b) 50th-percentile of money growth &  \\
\psfig{file = 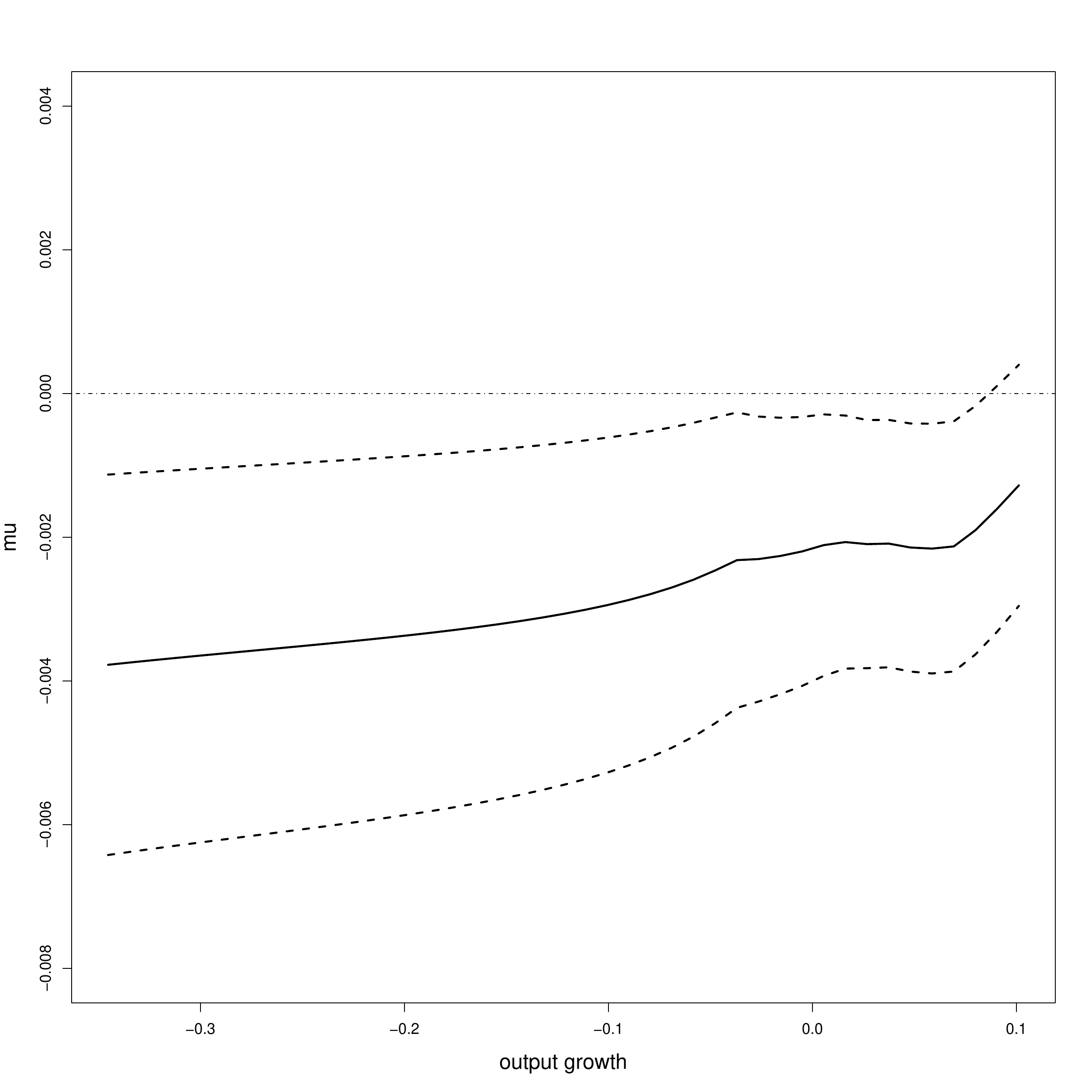, width = 6.0cm, angle=0} &  \\
(c) 75th-percentile of money growth &  \\
&  &
\end{tabular}
{\small {} }
\caption{The $solid$ curve is the local-linear regression estimate of $\mu(\cdot)$ for GBP/USD. The $dashed$ band is the 95\% SCR of $\mu(\cdot)$ and the dotted horizontal line is $H_0:\,\mu(\cdot)=0$. The bandwidth is obtained through under-smoothing of the GCV-chosen one.}
\label{fig_gbp_m}
\end{figure}

\begin{figure}[tbp]
\centering
\begin{tabular}{ccc}
\setlength{\itemsep}{-1.0cm} \psfig{file = 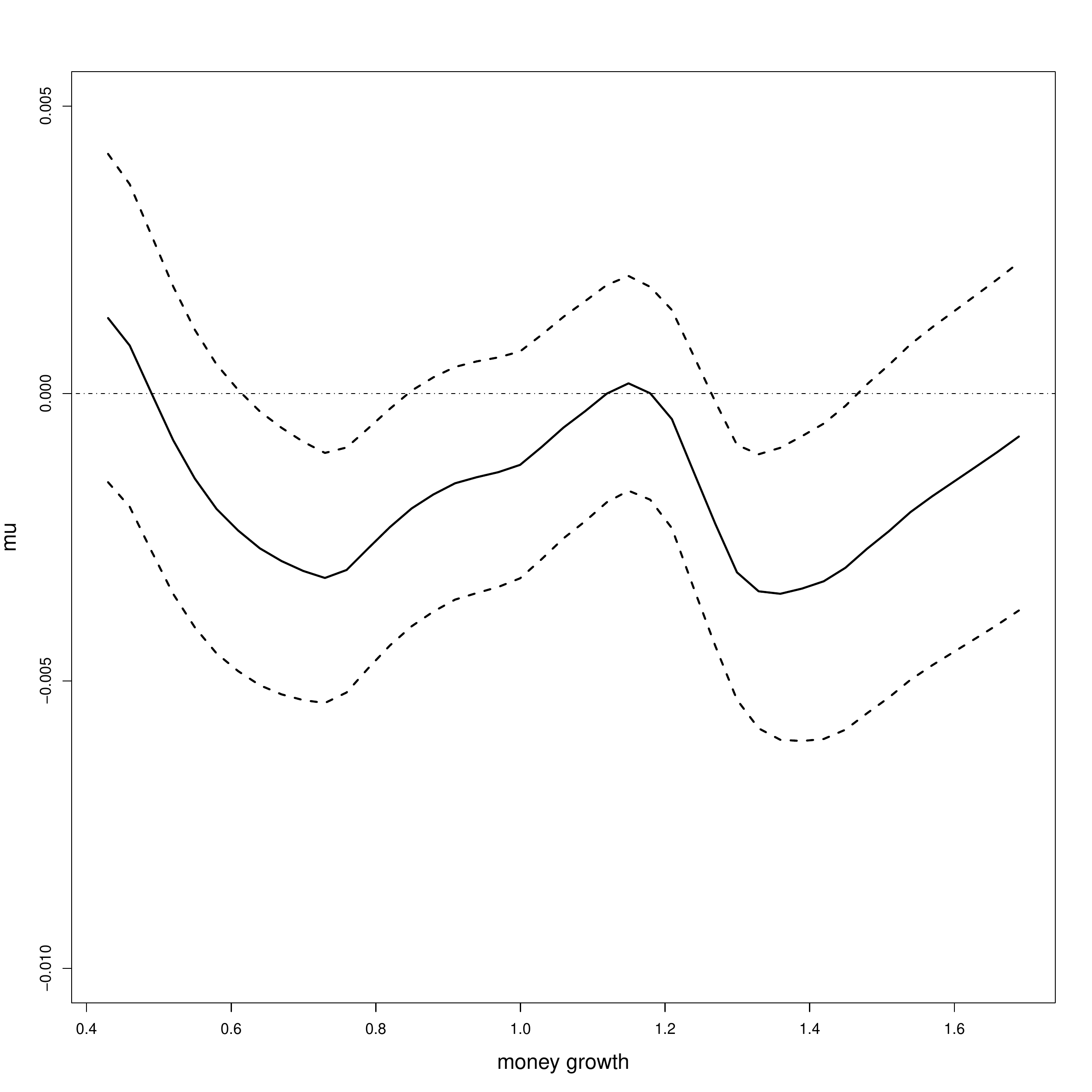, width = 6.0cm,
angle=0} & \psfig{file = 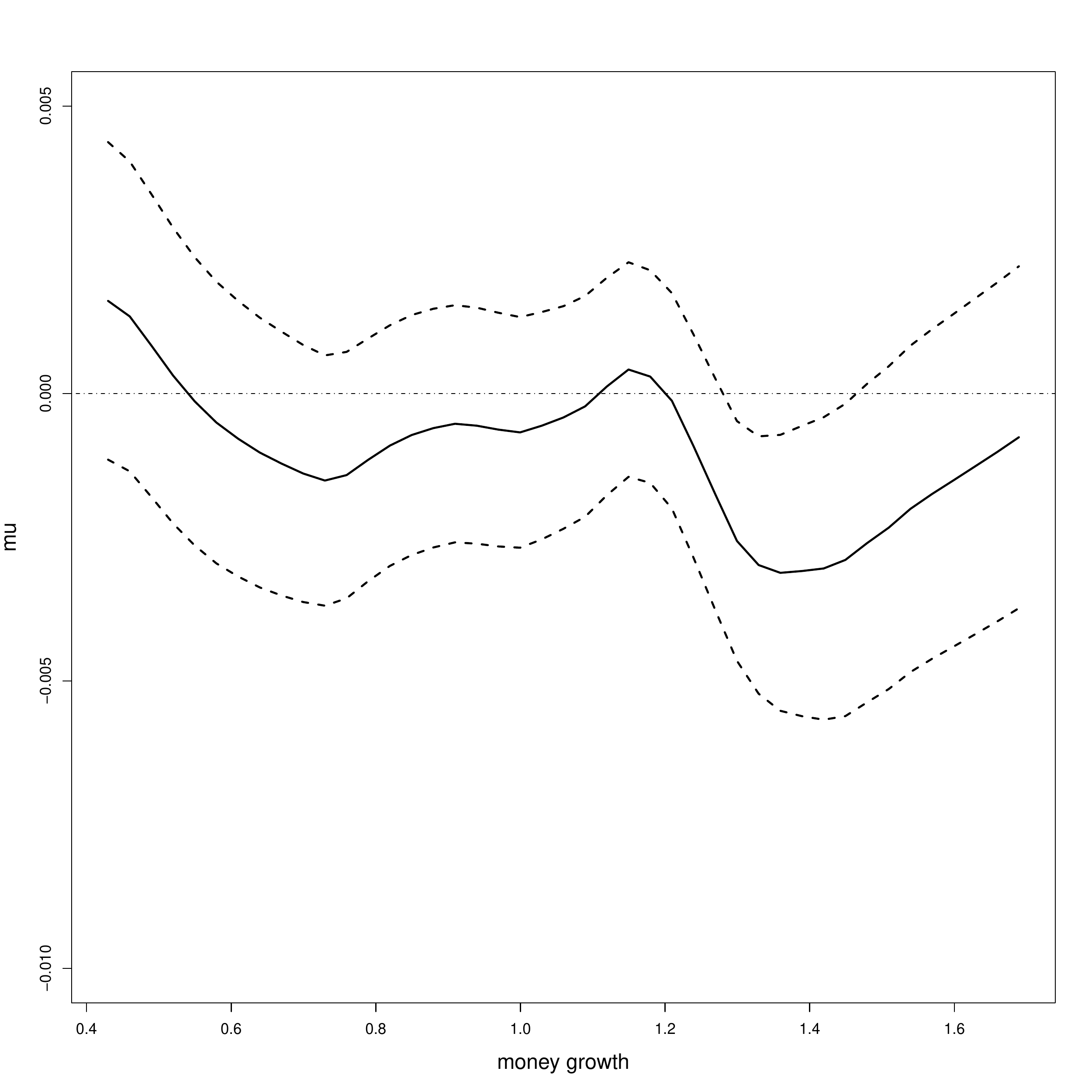, width = 6.0cm, angle=0} &  \\
(a) 25th-percentile of output growth & (b) 50th-percentile of output growth &  \\
\psfig{file = 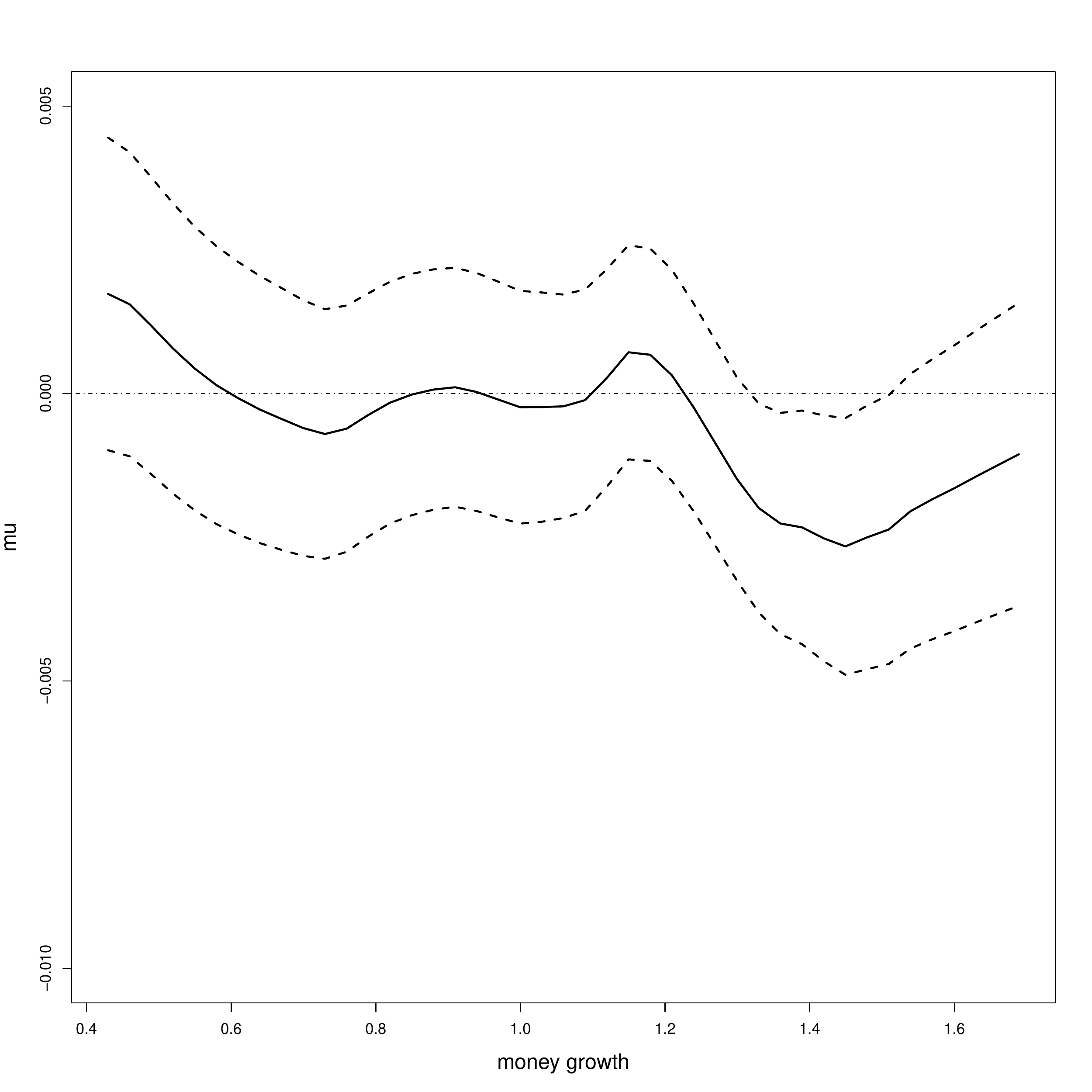, width = 6.0cm, angle=0} &  \\
(c) 75th-percentile of output growth &  \\
&  &
\end{tabular}
{\small {} }
\caption{The $solid$ curve is the local-linear regression estimate of $\mu(\cdot)$ for GBP/USD. The $dashed$ band is the 95\% SCR of $\mu(\cdot)$ and the dotted horizontal line is $H_0:\,\mu(\cdot)=0$. The bandwidth is obtained through under-smoothing of the GCV-chosen one.}
\label{fig_gbp_y}
\end{figure}

\clearpage

\printbibliography

\clearpage

\section*{Supplementary Materials}

\subsection*{S.1 Some Useful Lemmas}
\begin{lemma}[\cite{Burkholder,Rio}]
    \label{lemma_burkholder}
    Let $q>1,q'=\min\{q,2\}$, and $M_T=\sum_{t=1}^{T}\xi_t$, where $\xi_t\in\L^q$ are martingale differences. Then 
    $$\lVert M_T\rVert_q^{q'}\le K_{q}^{q'}\sum_{t=1}^{T}\lVert \xi_T\rVert_q^{q'},\quad\text{where } K_q=\max\{(q-1)^{-1},\sqrt{q-1}\}.$$
\end{lemma}

\begin{lemma}[\cite{Nazarov2003}]
    \label{lemma_nazarov}
    Let $X=(X_1,\ldots ,X_s)^{\top}$ be a centered Gaussian vector in $\RR^s$. Assume $\EE(X_i^2)\ge b$ for some $b>0$, $1\le i\le s$. Then for any $e>0$ and $d\in\RR^s$, 
    \begin{equation}
        \label{eq_nazarov}
        \sup_{x\in \mathbb{R}}\PP\Big(\big||X+d|_{\infty}-x\big|\le e\Big) \le  ce\big(\sqrt{2\log(s)}+2\big),
    \end{equation}
    where $c$ is a constant depending only on $b$.
\end{lemma}

\begin{lemma}[Comparison of Gaussian vectors (\cite{chen_inference_2022})]
    \label{lemma_comparison}
    Let $X=(X_1,X_2,\ldots,X_s)^{\top}$ and $Y=(Y_1,Y_2,\ldots,Y_s)^{\top}$ be two centered Gaussian vectors in $\RR^s$ and let $d=(d_1,d_2,\ldots,d_s)^{\top}$ be a fixed vectors in $\RR^s$. Define $\Delta=\max_{1\le i,j\le s}|\sigma_{i,j}^X-\sigma_{i,j}^Y|$, where $\sigma_{i,j}^X$ (resp. $\sigma_{i,j}^Y$) is set to be $\EE(X_iX_j)$ (resp. $\EE(Y_iY_j)$). Assume that $Y_i$'s have the same variance $\sigma^2>0$. Then,
    $$\sup_{x\in\RR}\Big|\PP\big(|X+d|_{\infty}\le x\big) - \PP\big(|Y+d|_{\infty}\le  x\big)\Big| \lesssim  \Delta^{1/3}\log^{2/3}(s),$$
    where the constants in $\lesssim$ only depends on $\sigma$.
\end{lemma}

The following lemma extends the concentration inequality in \textcite{freedman1975} to a high-dimensional version.
\begin{lemma}[Freedman's inequality]
    \label{lemma_freedman}
    Suppose that $\A$ is an index set with $|\A|<\infty$, and for each $a\in\A$, we let $\{\xi_{a,i}\}_{i=1}^n$ be a sequence of martingale differences with respect to the filtration $\{\F_i\}_{i=1}^n$. Let $M_a=\sum_{i=1}^n\xi_{a,i}$ and $V_a=\sum_{i=1}^n\EE[\xi_{a,i}^2\mid \F_{i-1}]$. Then, for any $z,u,v>0$,
    \begin{align*}
        \PP\big(\max_{a\in\A}|M_a|\ge z\big) \le \sum_{i=1}^n\PP\big(\max_{a\in\A}|\xi_{a,i}|\ge u\big) + 2\PP\big(\max_{a\in\A}V_a\ge v\big) + 2|\A|e^{-z^2/(2zu+2v)}.
    \end{align*}
\end{lemma}
We introduce some basic properties of the weights $w_h(x,X_i)$ defined in (\ref{eq_weight}) in the following lemma.
\begin{lemma}[Basic properties of weights]
    \label{lemma_weight}
    Under Assumptions \ref{asm_sigma}-\ref{asm_dep}, on some set $\A_n$ with
\begin{align*}
    \PP(\A_n) \geq \begin{cases}
    1-O(n^{-q/2+1}) & \text{if}\ \xi>1/2-1/q\\
    1-O(n^{-\xi q}) & \text{if}\ 0<\xi <1/2-1/q,
    \end{cases}
\end{align*}
we have
\begin{equation}
    \label{lm51}
    \sup_{x\in\T_d}|\hat g(x) - g(x)| \lesssim \sqrt{\log(n)/(h^dn)}+h^2,
\end{equation}
where the constants in $\lesssim$ here are independent of $h$ and $n$, and $\hat g(x)$ is defined as the nonparametric estimate of the density function $g(x)$, that is
\begin{equation}
    \label{eq_hat_g}
    \hat g(x)=n^{-1}\sum_{i=1}^nK_h(x-X_i).
\end{equation}
This result also indicates that on the set $\A_n$,
\begin{equation}
    \label{lm52}
    \sup_{x\in\T_d}\Big| \frac{1}{ng(x)}K_h(x-X_i) - w_h(x,X_i)\Big| \lesssim \Big(\sqrt{\log(n)/(h^dn)}+h^2\Big)/(h^dn),
\end{equation}
and further,
\begin{equation}
    \label{lm53}
    h^dn \max_{1\leq i\leq n}|w_h(x,X_i)| \leq  c_w \quad \text{and} \quad \sum_{i=1}^n|w_h(x,X_i)|=O(1),
\end{equation}
uniformly over $x$ and $h$, where $c_w$ is a constant independent of $h,n,N$.
\end{lemma}

\begin{proof}[Proof of Lemma \ref{lemma_weight}]
By the definition of $\hat g(x)$ in (\ref{eq_hat_g}), we have
\begin{align}
    \label{eq_lemma3_goal1}
    \sup_{x\in\T_d}|\hat g(x) - g(x)| = & \sup_{x\in\T_d}\frac{1}{h^dn}\Big|\sum_{i=1}^nK\Big(\frac{x-X_i}{h}\Big) - h^dn g(x)\Big|.
\end{align}
For brevity, we define
\begin{equation}
    \label{eq_lemma3_M}
    M_i(x) = K\Big(\frac{x - X_i}{h}\Big) - \EE\Big\{K\Big(\frac{x - X_i}{h}\Big)\mid
    \F_{i-1}\Big\}
\end{equation}
and
\begin{equation}
    \label{eq_lemma3_R}
    R_i(x) =\EE\Big\{K\Big(\frac{x - X_i}{h}\Big\}\mid
    \F_{i-1}\Big)-\EE\Big\{K\Big(\frac{x - X_i}{h}\Big)\Big\},
\end{equation}
which can decompose the difference between $\sum_{i=1}^nK\Big(\frac{x-X_i}{h}\Big)$ and its expectation into two parts, that is
\begin{equation}
    \label{eq_lemma3_MR}
    \EE_0\Big[\sum_{i=1}^nK\Big(\frac{x-X_i}{h}\Big)\Big] = M_n(x)+R_n(x), \quad \text{where } M_n(x)=\sum_{i=1}^n M_i(x) \,\text{ and }\, R_n(x)=\sum_{i=1}^n R_i(x).
\end{equation}
Recall $\EE_0(X)=X-E(X)$. Hereafter, we shall call expression (\ref{eq_lemma3_MR}) the M/R decomposition as introduced in \textcite{ZW:2008}. Since $\{M_i(x)\}_{i=1}^n$ are martingale differences with respect to $\F_i$ and $|M_i(x)|\le \sup_{x\in\T_d}|K(x)|$, it follows that
\begin{equation}
    \label{eq_lemma3_part11}
    \EE M_i^2(x) \leq \int_{\RR^{d}} K^2\Big(\frac{x-y}{h}\Big) g(y)dy = h^d\int_{\II^{d}} K^2(z)g(x-hz)dz \lesssim h^d\big(g(x)+O(h)\big).
\end{equation}
Therefore, by Freedman's inequality in Lemma \ref{lemma_freedman} and the chain argument in Lemma 4 in \textcite{zhao_kernel_2006}, it can be shown that
\begin{equation}
    \label{eq_lemma3_part11_result}
    \PP\Big(\sup_{x\in\T_d}|M_n(x)|\ge z \Big) \lesssim Ne^{-z^2/(z+h^dn)},
\end{equation}
where the constants in $\lesssim$ and $O(\cdot)$ here and the remaining proofs are independent of $n$ and $h$. As a direct consequence, with probability greater than $1-O(n^{-q})$, we have
\begin{equation}
    \label{eq_lemma3_part11_result2}
    \sup_{x\in\T_d}|M_n(x)| \lesssim \sqrt{h^dn\log(n)}.
\end{equation}
For the part $R_n(x)$, note that by Taylor's expansion, we have
\begin{equation}
    \label{eq_lemma3_part12}
    \EE\Big\{K\Big(\frac{x - X_i}{h}\Big) \mid \F_{i-1}\Big\} = h^d\int_{\II^d} K(z)g(x-hz \mid \F_{i-1})dz \lesssim h^d\big(g(x\mid \F_{i-1})+O(h^2)\big).
\end{equation}
Under Assumption \ref{asm_kernel}, it follows from the chain argument in Lemma 4 in \textcite{zhao_kernel_2006} and Lemma 5.8 in \textcite{zhang2015gaussian} that, for $z\gtrsim \sqrt{n\log(n)}$ and $\xi>0$, 
\begin{equation}
    \label{eq_lemma3_part21}
    \PP\Big(\sup_{x\in\T_d}\Big|\sum_{t=1}^n  g(x|\mathcal{F}_{t-1})\Big|\geq z\Big) \lesssim n^{1\vee (q/2-\xi q)}\log^{q/2}(n)/z^q + e^{-z^2/n},
\end{equation}
which yields, for $\xi>0$, with probability greater than $1-O(n^{-(\xi q)\wedge(-q/2+1)})$,  
\begin{equation}
    \label{eq_lemma3_part21_result}
    \sup_{x\in\T_d}|R_n(x)|\lesssim h^d\sqrt {n\log (n)}.
\end{equation}
Combining expressions (\ref{eq_lemma3_part11_result2}) and (\ref{eq_lemma3_part21_result}), with probability greater than $1-O(n^{-(\xi q)\wedge(-q/2+1)})$, we have
\begin{equation}
    \label{eq_lemma3_result1}
    \sup_{x\in\T_d}|M_n(x)+R_n(x)|\lesssim \sqrt{h^dn\log(n)} +h^d\sqrt {n\log (n)}.
\end{equation}
Note that for the expectation part, by Assumption \ref{asm_kernel} and Taylor's expansion, we have the approximation as follows
$$\sum_{i=1}^n\EE\Big\{K\Big(\frac{x - X_i}{h}\Big)\Big\}-h^dng(x) = O(h^{d+2}n).$$
This, along with expression (\ref{eq_lemma3_result1}) gives the inequality (\ref{lm51}). Utilizing the fact that $\sup_{x\in\T_d}|K(x)|<\infty$, we can obtain the first inequality in (\ref{lm53}). 

For inequality (\ref{lm52}), by the definition of the weight in expression (\ref{eq_weight}), we can write
\begin{align}
    \Big| K\Big(\frac{x - X_i}{h}\Big)/\big(h^dn\cdot g(x)\big) - w_h(x,X_i)\Big|
    &=\Bigg|w_h(x,X_i)\sum_{i=1}^n K\Big(\frac{x - X_i}{h}\Big)/\big(h^dng(x)\big) - w_h(x,X_i)\Bigg| \nonumber \\ 
    &\leq |w_h(x,X_i)|\cdot\Big|\sum_{i=1}^n K\Big(\frac{x - X_i}{h}\Big)/\big(h^dng(x)\big) -1\Big|.
\end{align}
This, together with \eqref{lm51} and the first inequality in \eqref{lm53} yields inequality (\ref{lm52}). By the similar arguments for the inequality (\ref{lm51}) and (\ref{lm52}), we can achieve the second inequality in (\ref{lm53}). 
\end{proof}

\subsection*{S.2 Proof of Theorem \ref{thm1_GA}} \label{subsec_thm1proof}

Since the proof of Theorem \ref{thm1_GA} is quite involved, we first sketch the outline of the proof strategies. For brevity, we define
\begin{equation}
    \label{eq_thm1_T}
    T_n=\sup_{x\in\T_d}\sqrt{h^dn}\big|\hat\mu(x)-\mu(x)\big|/\sigma(x).
\end{equation}
To achieve the result in Theorem \ref{thm1_GA}, we aim to derive a limit distribution for $T_n$. Since each dimension of $\T_d$ is a bounded interval in $\RR$ which is complicated to directly deal with, we shall first start with a $\delta$-net and then extend the results to $\T_d$. Specifically, for any $\delta>0$, we define a $\delta$-net 
\begin{equation}
    \label{eq_net}
    \{x_j\}_{j=1}^N \subset\T_d, 
\end{equation}
such that for any $y\in\T_d$, there exists $x\in\{x_j\}_{j=1}^N$ satisfying $|x-y|_{\infty} < \delta$. 
We define $T_{n,\delta}$ as the counterpart of $T_n$ on a $\delta$-net, that is
\begin{equation}
    \label{eq_thm1_T_discrete}
    T_{n,\delta}= \max_{1\le j\le N} \sqrt{h^dn}|\hat\mu(x_j)-\mu(x_j)\big|/\sigma(x_j).
\end{equation}
We shall first show that $T_n\approx T_{n,\delta}$ and then it suffices to work on the asymptotic properties of $T_{n,\delta}$. Further, we define
\begin{equation}
    \label{eq_thm1_I_epsilon_discrete}
    I_{\epsilon}=\max_{1\le j\le N} \sqrt{h^dn}\Big|\sum_{i=1}^nw_h(x_j,X_i)\sigma(X_i)\epsilon_i\Big|/\sigma(x_j).
\end{equation}
We would prove that $T_{n,\delta}$ and $I_{\epsilon}$ are close with high chance. Therefore, to study $T_{n,\delta}$, we only need to investigate $I_{\epsilon}$. 
For the innovations $\epsilon_i$ defined in expression (\ref{eq_epsilon}), we define a truncated version of $\epsilon_i$ as
$$\epsilon_{i,m} = \sum_{k=0}^{m-1} a_k \eta_{i-k},$$ 
for some integer $m>0$. We consider the $m$-dependent approximation $I_{\epsilon,m}$ of $I_{\epsilon}$, where $I_{\epsilon,m}$ is $I_{\epsilon}$ with $\epsilon_i$ therein replaced by $\epsilon_{i,m}$.
Then, we shall expect that for large $m$, $I_{\epsilon} \approx I_{\epsilon,m}$. Set $I_{z,m}$ to be $I_{\epsilon,m}$ with $\eta_k$ therein replaced by $z_k$, where $z_k\in\RR$, $k\in\ZZ$, are i.i.d. Gaussian random variables with mean zero and unit variance. Note that when the covariates $X_1,\ldots,X_n$ are observed, $I_{\epsilon,m}$ can be written into the maximum of a sequence of independent random variables. Hence, by the Gaussian approximation result in Proposition 2.1 by \textcite{chernozhukov_central_2017}, we could approximate the distribution of $I_{\epsilon,m}$ by the one of $I_{z,m}$. Then, we would prove that the distributions of $I_{z,m}$ and $\max_{1\le j\le N}|\Z_{t_j}|$ are close. Finally, we complete the proof by showing the continuity of the maximum of a non-centered Gaussian distribution and extending all the results to the region $\T_d$.

\begin{proof}[Proof of Theorem \ref{thm1_GA}]
We now provide the rigorous proof of Theorem \ref{thm1_GA}. First, recall that $\Z_t$, $t\in\T_d$, is a centered Gaussian random field with conditional covariance matrix $Q = (Q_{t,s})_{t,s\in\T_d}$, where $Q_{t,s}$ is defined in (\ref{eq_cov_Z}). Let $\Z_{t_j}$, $1\le j\le N$, be the discretized version of $\Z_t$ on a $\delta$-net as defined in Lemma \ref{lemma_deltanet}, with mean zero and the conditional covariance matrix $Q^{(L)} = (Q_{j_1,j_2}^{(L)})_{1\le j_1,j_2\le N}$, where $Q_{t,s}^{(L)}$ is defined in (\ref{eq_cov_Z_truncate}). Also, recall the definition of the test statistic $T_n$ in expression (\ref{eq_thm1_T}). 
Let $m=n/2$. Then, for any $\alpha_1>0$, we have
\begin{align}
    \PP\big(T_n\le u\big) & \le  \PP\big( I_{\epsilon,m}\le u+\alpha_1\big) + \PP\big(|T_n-I_{\epsilon,m}|\ge \alpha_1\big) \nonumber \\
    & = \PP\big(\max_{1\le j\le N}|\Z_{t_j}|\le u+\alpha_1\big) + \Big[\PP\big(I_{\epsilon,m}\le u+\alpha_1\big) - \PP\big(\max_{1\le j\le N}|\Z_{t_j}^{\dagger}|\le u+\alpha_1\big)\Big] \nonumber \\
    & \quad  + \Big[\PP\big(\max_{1\le j\le N}|\Z_{t_j}^{\dagger}| \le u\big) - \PP\big(\max_{1\le j\le N}|\Z_{t_j}| \le u\big)\Big] + \PP\big(|T_n-I_{\epsilon,m}|\ge \alpha_1\big).
\end{align}
Let $\Z_{t_j}^\dagger$ be $\Z_{t_j}$ with conditional covariance matrix replaced by $Q = (Q_{j_1,j_2})_{1\le j_1,j_2\le N}$ (i.e., the non-truncated, but discretized version), where $Q_{t,s}$ is defined in (\ref{eq_cov_Z}). Then, we further note that, for any $\alpha_2>0$, 
\begin{align}
    \PP\big(\max_{1\le j\le N}|\Z_{t_j}|\le u+\alpha_1\big) 
    & \le \PP\big(\max_{1\le j\le N}|\Z_{t_j}|\le u-\alpha_2\big) + \PP\Big(\big|\max_{1\le j\le N}|\Z_{t_j}|-u\big|\le \alpha_1+\alpha_2\Big) \nonumber \\
    & \le \PP\big(\sup_{t\in\T_d}|\Z_t|\le u\big) +\PP\Big(\big|\sup_{t\in\T_d}|\Z_t|-\max_{1\le j\le N}|\Z_{t_j}|\big|\ge \alpha_2\Big) \nonumber \\
    & \quad + \PP\Big(\big|\max_{1\le j\le N}|\Z_{t_j}|-u\big|\le \alpha_1+\alpha_2\Big).
\end{align}
Hence, by combining the two inequalities above, we obtain, for any $\alpha_3>0$,
\begin{align}
    \label{eq_thm1_goal2}
    & \quad \sup_{u\in\RR}\Big[\PP\big(T_n\le u\big) - \PP\big(\sup_{t\in\T_d}|\Z_t|\le u\big) \Big] \nonumber \\
    & \le \PP\big(|T_n-I_{\epsilon,m}|\ge \alpha_1\big) + \sup_{u\in\RR}\Big|\PP\big(I_{\epsilon,m}\le u\big) - \PP\big(I_{z,m}\le u\big)\Big| + \PP\Big(\big|I_{z,m}-\max_{1\le j\le N}|\Z_{t_j}^{\dagger}|\big|\ge \alpha_3\Big) \nonumber \\
    & \quad +\PP\Big(\big|\sup_{t\in\T_d}|\Z_t|-\max_{1\le j\le N}|\Z_{t_j}^{\dagger}|\big|\ge \alpha_2\Big)  + \sup\Big|\PP\big(\max_{1\le j\le N}|\Z_{t_j}^{\dagger}| \le u\big) - \PP\big(\max_{1\le j\le N}|\Z_{t_j}| \le u\big)\Big| \nonumber \\
    & \quad + \PP\Big(\big|\max_{1\le j\le N}|\Z_{t_j}|-u\big|\le \sum_{s=1}^3\alpha_s\Big) =: \sum_{k=1}^6\III_k.
\end{align}
We shall investigate the parts $\III_1$-$\III_6$ separately. First, we study the $\III_1$ part. Note that 
\begin{equation}
    \label{eq_thm1_part1_3parts}
    |T_n-I_{\epsilon,m}| \le |T_n-T_{n,\delta}| + |T_{n,\delta}-I_{\epsilon}| + |I_{\epsilon}-I_{\epsilon,m}|.
\end{equation}
Let $\alpha_{11}=ch$, where $c>0$ is a constant. Then, as a direct consequence of Lemma \ref{lemma_deltanet}, we achieve
\begin{equation}
    \label{eq_thm1_cont2}
    \PP\big(|T_n-T_{n,\delta}|\ge \alpha_{11} \big)\lesssim 1/n + \PP(\A_n^{\text{c}}),
\end{equation}
where $\PP(\A_n^{\text{c}})=O\big(n^{(-q/2+1)\vee(-\xi q)}\big)$ as indicated by Lemma \ref{lemma_weight}. Since we consider $\delta=h^{3d/2+1}n^{-(1/2+1/q)}$, we have $N\lesssim 1/\delta=O(h^{-3d/2-1}n^{1/2+1/q})$. Let $\alpha_{12}=c_1'h\sqrt{\log(n)}+c_2'h^2\sqrt{h^dn}$, where $c_1',c_2'>0$ are some large constants. Then, by Lemma \ref{lemma_TN_Iepsilon}, we obtain
\begin{equation}
    \label{eq_thm1_approx_mu}
    \PP\big(|T_{n,\delta}-I_{\epsilon}| \ge \alpha_{12}\big) \lesssim 1/n.
\end{equation}
Now we define $\alpha_{13}=c''n^{-\zeta}\log^{1/2}(N)$, where the constant $c''>0$. It follows from Lemma \ref{lemma_m} that
\begin{equation}
    \label{eq_lemma_m_conclusion}
    \PP\big(|I_{\epsilon}-I_{\epsilon,m}|\ge \alpha_{13}\big)\lesssim (h^dn)^{-q/2}n\log^q(n).
\end{equation}
Hence, for $\alpha_1=\alpha_{11}+\alpha_{12}+\alpha_{13}$, by inserting the results in expressions (\ref{eq_thm1_cont2}), (\ref{eq_thm1_approx_mu}) and (\ref{eq_lemma_m_conclusion}) into expression (\ref{eq_thm1_part1_3parts}), we have
\begin{equation}
    \label{eq_thm1_part1}
    \III_1 = O\big\{ (h^dn)^{-q/2}n\log^q(n)+1/n\big\}.
\end{equation}
For the $\III_2$ part, as a direct consequence of Lemma \ref{lemma_m_GA}, we obtain
\begin{equation}
    \label{eq_thm1_part21}
    \III_2\lesssim (h^dn)^{-1/6} \log^{7/6} (Nn) + (n^{2/q}/(h^dn))^{1/3} \log(Nn).
\end{equation}
Concerning the part $\III_3$, we let $\alpha_3=c_3n^{-\zeta}\log^{1/2}(N)$ where the constant $c_3>0$. It follows from a similar argument
in expression (\ref{eq_lemma_m_conclusion}) that $\III_3=O\big\{(h^dn)^{-q/2}n\log^q(n)\big\}$. Similar to expression (\ref{eq_thm1_cont2}), we let $\alpha_2=c_2h^{d/2}$ with some constant $c_2>0$. Then, we could achieve $\III_4=O(1/n)$. 

Finally, for the parts $\III_5$ and $\III_6$, recall the centered Gaussian process $\Z_{t_j}$ and $\Z_{t_j}^\dagger$ defined in the first paragraph of this proof. We shall first show that the distributions of $\Z_{t_j}^\dagger$ and $\Z_{t_j}$ are close and bound $\III_5$. To this end, we note that
\begin{align}
    \label{eq_thm1_longrun}
    |Q_{j_1,j_2}-Q_{j_1,j_2}^{(L)}| & = h^dn\Big|\sum_{|k|\ge L}\sum_{i=1\vee(1-k)}^{n\wedge(n-k)}c_{j_1,j_2,i,k}w_h(x_{j_1},X_i)w_h(x_{j_2},X_{i+k})\gamma(k)\Big| \nonumber \\
    & \lesssim h^dn\max_{k}\big|w_h(x_{j_2},X_k)\sigma(X_k)/\sigma(x_{j_2})\big|\cdot\Big|\sum_{k\ge L}\gamma(k)\sum_{i=1}^{n-k}w_h(x_{j_1},X_i)\sigma(X_i)/\sigma(x_{j_1})\Big| \nonumber \\
    & \lesssim L^{-\zeta}.
\end{align}
Then, it follows from Lemma \ref{lemma_comparison} that 
$$\III_5 = \sup_{u\in\RR}\Big|\PP\big(\max_{1\le j\le N}|\Z_{t_j}^\dagger|\le u\big) - \PP\big(\max_{1\le j\le N} |\Z_{t_j}|\le u\big)\Big| \lesssim L^{-\zeta/3}\log^{2/3}(n).$$

Secondly, we show a continuity result for $\Z_{t_j}$ and bound $\III_6$. We refer to the definition that $M_{ji}=w_h(x_j,X_i)-\EE(w_h(x_j,X_i)\mid\F_{i-1})$ in expression (\ref{eq_lm5_w_decompose}), and we shall apply a similar decomposition technique here. In particular, we have
\begin{align}
    \label{eq_thm1_part5}
    & \sum_{i_1=1}^nw_h(x_j,X_{i_1})\sum_{1\le i_2<i_1}^{n}w_h(x_j,X_{i_2})\gamma(i_1-i_2) \nonumber \\
    = & \sum_{i_1=1}^n\big[M_{ji_1}+O(1/n)\big]\sum_{1\le i_2<i_1}^{n}\big[M_{ji_2}+O(1/n)\big]\gamma(i_1-i_2) \nonumber \\
    = & \sum_{i_1=1}^n\sum_{1\le i_2<i_1}^{n}M_{ji_1}M_{ji_2}\gamma(i_1-i_2) + O(1/n),
\end{align}
which along with expression (\ref{eq_cov_Z}) gives
\begin{align}
    \label{eq_thm1_part5_Q}
    Q_{j,j} & =   h^dn\sum_{i_1=1}^n\Big(c_{j,j,i_1,i_1}\EE_0(M_{ji_1}^2)\gamma(0)+2\sum_{1\le i_2<i_1}c_{j,j,i_1,i_2}M_{ji_1}M_{ji_2}\gamma(i_1-i_2)\Big) \nonumber \\
    & \quad + h^dn\sum_{i_1=1}^nc_{j,j,i_1,i_1}\EE M_{ji_1}^2\gamma(0) + O(h^d) \nonumber \\ 
    & =: h^dn\sum_{i_1=1}^n\tilde D_{ji_1}+ h^dn\sum_{i_1=1}^nc_{j,j,i_1,i_1}\EE M_{ji_1}^2\gamma(0) + O(h^d).
\end{align}
Since the long-run covariance of $\{\epsilon_i\}$ is bounded, by Lemma \ref{lemma_weight} and Assumption \ref{asm_sigma}, we shall achieve that $h^dn\sum_{i_1=1}^n\EE c_{j,j,i_1,i_1}M_{ji_1}^2\gamma(0) \asymp c$, for some constant $c>0$.
Note that $\{\tilde D_{ji_1}\}$ are martingale differences with respect to the filtration $F_{i_1-1}$, which along with Freedman's inequality in Lemma \ref{lemma_freedman} yields
\begin{align}
    \label{eq_thm1_part5_freedman}
    \PP\Big(h^dn\sum_{i_1=1}^n\tilde D_{ji_1} \ge z\Big) \le e^{-z^2/(zv+u)},
\end{align}
where $u$ is the upper bound for $h^dn|\tilde D_{ji_1}|= O\{1/(h^dn)\}$ and $v$ is the upper bound for $\sum_{i_1=1}^{n}\EE[(h^dn\tilde D_{ji_1})^2 \mid \F_{i_1-1}]=O\{1/(h^dn)\}$. Therefore, we have $\max_{1\le j\le N}h^dn\sum_{i_1=1}^n\tilde D_{ji_1} = O\{\sqrt{\log(n)/(h^dn)}\}$ with probability greater than $1-n^{-q}$. This, along with expression (\ref{eq_thm1_part5_Q}) gives
\begin{align}
    \label{eq_thm1_part5_result}
    \min_{1\le j\le N} Q_{j,j} \gtrsim c,
\end{align}
with probability greater than $1-n^{-q}$. Hence, by applying the anti-concentration inequality in Lemma \ref{lemma_nazarov}, we achieve
\begin{align}
    \label{eq_thm1_part4}
    \III_6 & \lesssim  (\alpha_1+\alpha_2+\alpha_3)\sqrt{\log(N)}+n^{-q}  \nonumber \\
    & \lesssim \Big(h\sqrt{\log(n)}+h^2\sqrt{h^dn}+n^{-\zeta}\sqrt{\log(n)} \Big)\sqrt{\log(n)}+n^{-q}.
\end{align}
By combining the results from $\III_1$--$\III_6$ and a similar argument for the other side of the inequality (\ref{eq_thm1_goal2}), we complete the proof.
\end{proof}

\begin{lemma}[$\delta$-net approximation]
    \label{lemma_deltanet}
    For a compact region $\T_d\subset\RR^d$ and a $\delta$-net $\{x_j\}_{j=1}^N \subset\T_d$ such that for any $y\in\T_d$, there exists $x\in\{x_j\}_{j=1}^N$ satisfying $|x-y|_{\infty} < \delta$, we have, for some constant $c>0$,
    
    $$\PP\big(|T_n-T_{n,\delta}| \ge ch \big)= O(1/n)+ \PP(\A_n^{\text{c}}),$$
    where $\PP(\A_n^{\text{c}})=O\big(n^{(-q/2+1)\vee(-\xi q)}\big)$ as indicated by Lemma \ref{lemma_weight}.
\end{lemma}
\begin{proof}[Proof of Lemma \ref{lemma_deltanet}]
By the definition of $\delta$-net and Assumption \ref{asm_sigma}, we have
\begin{align}
    \label{eq_thm1_part11_twoparts}
    |T_n-T_{n,\delta}| & = \sqrt{h^dn}\Big|\sup_{x\in\T_d}\big|\hat\mu(x)-\mu(x)\big|/\sigma(x) - \max_{1\le j\le N}\big|\hat\mu(x_j)-\mu(x_j)\big|/\sigma(x_j)\Big| \nonumber \\
    & \lesssim  \sqrt{h^dn}\max_{1\le j\le N}\sup_{|x-x_j|_{\infty}<\delta}\Big(\big|\hat\mu(x)-\hat\mu(x_j)\big|+\big|\mu(x)-\mu(x_j)\big|\Big).
\end{align}
Since the smooth function $\mu(\cdot)$ is Lipschitz continuous, it follows that
\begin{equation}
    \label{eq_thm1_part11_mu}
    \max_{1\le j\le N}\sup_{|x-x_j|_{\infty}<\delta}\big|\mu(x)-\mu(x_j)\big|\lesssim \delta.
\end{equation}
Also note that on the set $\A_n$ defined in Lemma \ref{lemma_weight}, $\hat\mu(x)$ can be written into
\begin{align}
    \label{eq_thm1_muhat}
    \hat\mu(x)\One_{\A_n}= & \sum_{i=1}^nw_h(x,X_i)\One_{\A_n}\big(Y_i-Z_i^{\top}\bbeta\big) \nonumber \\
    = & \sum_{i=1}^nw_h(x,X_i)\One_{\A_n}\mu(X_i) +  \sum_{i=1}^nw_h(x,X_i)\One_{\A_n}\sigma(X_i)\epsilon_i.
\end{align}
Therefore, we can decompose the first part in expression (\ref{eq_thm1_part11_twoparts}) into
\begin{align}
    \label{eq_thm1_part11_muhat}
    & \quad \max_{1\le j\le N}\sup_{|x-x_j|_{\infty}<\delta}\big|\hat\mu(x)-\hat\mu(x_j)\big|\One_{\A_n} \nonumber \\
    & = \max_{1\le j\le N}\sup_{|x-x_j|_{\infty}<\delta} \Big|\sum_{i=1}^n\big(w_h(x,X_i)-w_h(x_j,X_i)\big)\One_{\A_n}\mu(X_i) \nonumber \\
    & \quad + \sum_{i=1}^n\big(w_h(x,X_i)-w_h(x_j,X_i)\big)\One_{\A_n}\sigma(X_i)\epsilon_i\Big| \nonumber \\
    & \lesssim  \delta h^{-2d} \cdot \max_{1\le i\le n}|\epsilon_i|,
\end{align}
where the last inequality holds since $\mu(\cdot)$ is Lipschitz continuous, and Assumptions \ref{asm_sigma} and \ref{asm_kernel}, as well as expression (\ref{lm53}) in Lemma \ref{lemma_weight} yield
\begin{align}
    & \quad \max_{1\le j\le N}\sup_{|x-x_j|_{\infty}<\delta} \Big|\sum_{i=1}^n\big(w_h(x,X_i)-w_h(x_j,X_i)\big)\One_{\A_n}\Big| \nonumber \\
    & \lesssim  \max_{1\le j\le N}\sup_{|x-x_j|_{\infty}<\delta}\frac{1}{h^dn}\sum_{i=1}^n\max_{1\le k\le d}\Big|\partial K\Big(\frac{x_j-X_i}{h}\Big)/\partial x_{jk}\Big|\cdot|x-x_j|_{\infty} \nonumber \\
    & \lesssim \delta h^{-2d}.
\end{align}
Here, $x_{jk}$ is the $k$-th coordinate of $x_j$, $1\le k\le d$. Note that for $q\ge4$, $\EE|\epsilon_i|^q<\infty$ by Assumption \ref{asm_moment}. Then, since $\delta=h^{3d/2+1}n^{-(1/2+2/q)}$ and $\PP\big(\max_{1\le i\le n}|\epsilon_i|>n^{2/q}\big)\lesssim 1/n$, we achieve the desired result by inserting expressions (\ref{eq_thm1_part11_mu}) and (\ref{eq_thm1_part11_muhat}) into expression (\ref{eq_thm1_part11_twoparts}).

\end{proof}

\begin{lemma}[Tail probability of $T_{n,\delta}$]
    \label{lemma_TN_Iepsilon}
    For two statistics $T_{n,\delta}$ and $I_{\epsilon}$ as defined in expressions (\ref{eq_thm1_T_discrete}) and (\ref{eq_thm1_I_epsilon_discrete}), we have
    $$\PP\big(|T_{n,\delta}-I_{\epsilon}| \ge z+c'h^{d/2+2}n^{1/2}\big)\lesssim 2N\exp\Bigg\{\frac{-z^2}{zh/\sqrt{h^dn}+h^2}\Bigg\} + \PP(\A_n^{\text{c}}),$$
    where $c'>0$ is some constant, $h>0$ is the bandwidth parameter and $\PP(\A_n^{\text{c}})=O\big(n^{(-q/2+1)\vee(-\xi q)}\big)$ as indicated by Lemma \ref{lemma_weight}.
\end{lemma}
\begin{proof}[Proof of Lemma \ref{lemma_TN_Iepsilon}]
We denote the first part in expression (\ref{eq_thm1_muhat}) on the $\delta$-net by
\begin{equation}
    \label{eq_thm1_mu_tilde}
    \tilde \mu(x_j)=\sum_{i=1}^nw_h(x_j,X_i)\mu(X_i).
\end{equation}
Then, by the definition of $I_{\epsilon}$ in expression (\ref{eq_thm1_I_epsilon_discrete}) and the fact that $\sum_{i=1}^nw_h(x_j,X_i)=1$, we have
\begin{align}
    \label{eq_thm1_T_I}
    |T_{n,\delta}-I_{\epsilon}| 
    & \le \sqrt{h^dn}\max_{1\le j\le N}\big|\mu(x_j)-\tilde\mu(x_j)\big|/\sigma(x_j) \nonumber \\
    & = \sqrt{h^dn}\max_{1\le j\le N}\Big|\sum_{i=1}^nw_h(x_j,X_i)\big(\mu(X_i)-\mu(x_j)\big)\Big|/\sigma(x_j).
\end{align}
On the set $\A_n$ defined in Lemma \ref{lemma_weight}, we let $\gamma_i(x_j)=w_h(x_j,X_i)\One_{\A_n}\big(\mu(X_i)-\mu(x_j)\big)/\sigma(x_j)$. We apply the same decomposition technique in expression (\ref{eq_lm5_twoparts}) in Lemma \ref{lemma_m_GA} to expression (\ref{eq_thm1_T_I}), that is
\begin{align}
    \label{eq_thm1_mu_decompose}
    |T_{n,\delta}-I_{\epsilon}|\One_{\A_n} \le \sqrt{h^dn}\max_{1\le j\le N}\Big|\sum_{i=1}^n\big\{\gamma_i(x_j)-\EE[\gamma_i(x_j)\mid\F_{i-1}]\big\}\Big| + \sqrt{h^dn}\max_{1\le j\le N}\Big|\sum_{i=1}^n\EE[\gamma_i(x_j)\mid\F_{i-1}]\Big|.
\end{align}
For the expectation part in expression (\ref{eq_thm1_mu_decompose}), by Assumptions \ref{asm_sigma}, \ref{asm_kernel} and Lemma \ref{lemma_weight}, we have
\begin{align}
    \label{eq_thm1_mu_E}
    \EE[\gamma_i(x_j)\mid\F_{i-1}] & \lesssim  \frac{1}{h^dn}\int_{\T_d}K\Big(\frac{x_j-y}{h}\Big)\big[\mu(y)-\mu(x_j)\big]g(y\mid\F_{i-1})dy \nonumber \\
    & = \frac{1}{n}\int_{\T_d}K(z)\big[\mu(x_j-hz)-\mu(x_j)\big]g(x_j-hz\mid\F_{i-1})dz \nonumber \\
    & = \frac{1}{n}\int_{\T_d}K(z)\big[-hz\mu'(x_j) +h^2z^2\mu''(x_j)\big]\big[g(x_j\mid\F_{i-1}) - hzg'(x_j\mid\F_{i-1})\big]dz \nonumber \\
    & = O(h^2).
\end{align}
Since the sequence $\{\gamma_i(x_j)-\EE[\gamma_i(x_j)\mid\F_{i-1}]\}_i$ are martingale differences with respect to the filtration $\F_i$, by Freedman's inequality in Lemma \ref{lemma_freedman}, for any $z>0$,
\begin{equation}
    \label{eq_thm1_mu_E_approx}
    \PP\Big(\max_{1\le j\le N}\sqrt{h^dn}\Big|\sum_{i=1}^n\big\{\gamma_i(x_j)-\EE[\gamma_i(x_j)\mid\F_{i-1}]\big\}\Big| \ge z\Big) 
    \le  2Ne^{-z^2/(2zu+2v)},
\end{equation}
where $v$ is the upper bound for
\begin{align}
    \label{eq_thm1_mu_var}
    \sigma_{\gamma}^2 = & \max_{1\le j\le N}h^dn\sum_{i=1}^n\EE\big[\big(\gamma_i(x_j)-\EE[\gamma_i(x_j)\mid\F_{i-1}]\big)^2\mid \F_n\big] \nonumber \\
    = & \max_{1\le j\le N}h^dn\sum_{i=1}^n\EE\big[w_h^2(x_j,X_i)\One_{\A_n}\big(\mu(X_i)-\mu(x_j)\big)^2\mid \F_n\big](1+o(1)) \lesssim h^2,
\end{align}
and $u$ is the upper bound for
\begin{equation}
    \label{eq_thm1_mu_bd}
    \max_{1\le j\le N}\sqrt{h^dn}\big|\gamma_i(x_j)-\EE[\gamma_i(x_j)\mid\F_{i-1}]\big| \lesssim h/\sqrt{h^dn}.
\end{equation}
Hence, by expressions (\ref{eq_thm1_T_I}) and (\ref{eq_thm1_mu_E}) and Assumption \ref{asm_sigma}, we have
\begin{equation}
    \label{eq_thm1_mu_E_approx_result}
    \PP\Big(\max_{1\le j\le N}\sqrt{h^dn}\Big|\sum_{i=1}^n\big\{\gamma_i(x_j)-\EE[\gamma_i(x_j)\mid\F_{i-1}]\big\}\Big| \ge z\Big) \le 2N\exp\Bigg\{\frac{-z^2}{zh/\sqrt{h^dn}+h^2}\Bigg\}.
\end{equation}
This, together with expressions (\ref{eq_thm1_mu_decompose}) and (\ref{eq_thm1_mu_E}) completes the proof.
\end{proof}

\begin{lemma}[$m$-dependent approximation]
    \label{lemma_m}
    Under the conditions in Theorem \ref{thm1_GA}, for some $0< m\le n$, we have 
    $$\PP\big(|I_{\epsilon}-I_{\epsilon,m}|\ge z\big)\le n\log^q(n)z^{-q}(h^dn)^{-q/2}m^{-q\zeta}+ N e^{-z^2m^{2\zeta}} + \PP(\A_n^{\text{c}}),$$
    where $\PP(\A_n^{\text{c}})=O\big(n^{(-q/2+1)\vee(-\xi q)}\big)$ as indicated by Lemma \ref{lemma_weight}.
\end{lemma}
\begin{proof}[Proof of Lemma \ref{lemma_m}]
Recall the definition of $\epsilon_i$ in expression (\ref{eq_epsilon}). Note that on the set $\A_n$ defined in Lemma \ref{lemma_weight}, we can bound $|I_{\epsilon} - I_{\epsilon,m}|$ by
\begin{align}
    \label{eq_thm1_m}
    |I_{\epsilon} - I_{\epsilon,m}|\One_{\A_n} & \le \sqrt{h^dn}\max_{1\le j\le N}\Big|\sum_{i=1}^{n} w_h(x_j,X_i)\One_{\A_n}\sigma(X_i)(\epsilon_i-\epsilon_{i,m})\Big|/\sigma(x_j) \nonumber \\
    & =: \sqrt{h^dn}\max_{1\le j\le N} \Big|\sum_{l\le n-m}b_{j,l}\eta_{l}\Big|,
\end{align}
where 
\begin{equation}
    \label{eq_thm1_b}
    b_{j,l} = \Big(\sum_{i=1 \vee (l+m)}^{n} w_h(x_j,X_i)\One_{\A_n}\sigma(X_i)a_{i-l}\Big)/\sigma(x_j).
\end{equation}
Recall the filtration $\F_i = (\ldots, v_{i-1},v_{i})$. It can be shown that $b_{j,l}\eta_l$ are martingale differences for different $l$ with respect to $\G_i = (\F_n$, $\eta_l, l\leq i)$. Then, by Freedman's inequality in Lemma \ref{lemma_freedman}, we have, for any $u>0$, 
\begin{align}
    \label{lm5main}
    & \PP\Big(\sqrt{h^dn}\max_{1\le j\le N} \Big|\sum_{l\le n-m}b_{j,l}\eta_{l}\Big|\ge z\Big) \nonumber \\
    \le & \sum_{l \le n-m} \PP\big(\max_{1\le j\le N} \sqrt{h^dn}|b_{j,l}\eta_l|\ge u\big) + 2N e^{-z^2/(2zu+2v)},
\end{align}
where $v$ is the upper bound for 
$$ \theta^2 = \max_{1 \leq j \leq N}\sum_{l \leq n-m} h^dn\EE \big(|b_{j,l} \eta_l|^2 \mid \G_{l-1}\big).$$
Since $\eta_{l}$ are independent for different $l$, we have
\begin{align}
    \label{lm5ineq1}
    \theta^2 
    = \max_{1 \leq j \leq N} h^dn \sum_{l \leq n-m} b_{j,l}^2 
    \leq \max_{1 \leq j \leq N} h^dn \Big(\sum_{l \leq n-m} |b_{j,l}|\Big) \cdot \max_{l\leq n-m} |b_{j,l}|.
\end{align}
By expression (\ref{lm53}) in Lemma \ref{lemma_weight}, the definition of $b_{l,j}$ in expression (\ref{eq_thm1_b}) and Assumptions \ref{asm_sigma} and \ref{asm_dep_epsilon}, we obtain 
\begin{align}
    \label{lm5ineq11}
    \max_{l\leq n-m} |b_{j,l}| \lesssim \max_{l\leq n-m} \Big( \sum_{i=1 \vee (l+m)}^n w_h(x_j,X_i)\One_{\A_n}|a_{i-l}| \Big) \lesssim \frac{1}{h^dn} \max_{l\leq n-m} \sum_{i=l+m}^n |a_{i-l}| \lesssim \frac{1}{h^dn} m^{-\zeta}.
\end{align}
By Lemma \ref{lemma_weight}, we have $\sum_{i=1}^n w_h(x_j,X_i)\One_{\A_n} = O(1)$. Then, by Assumption \ref{asm_dep_epsilon}, we obtain 
\begin{align}
    \label{lm5ineq12}
    \max_{1 \leq j \leq N} \Big|\sum_{l \leq n-m} b_{j,l}\Big| \lesssim \max_{1 \leq j \leq N}\Big( \sum_{i=1}^n w_h(x_j,X_i)\One_{\A_n} \sum_{l\leq i-m}|a_{i-l}|\Big)\lesssim m^{-\zeta}.
\end{align}
Implementing inequality \eqref{lm5ineq11} and \eqref{lm5ineq12} to inequality \eqref{lm5ineq1} leads to
\begin{equation}
    \label{lm5ineq3}
    \theta^2\lesssim m^{-2\zeta}.
\end{equation}
Note that by expression (\ref{lm53}) in Lemma \ref{lemma_weight}, we obtain
\begin{equation}
\max_{1 \leq j \leq N} (h^dn)^{1/2} |b_{j,l}|  \lesssim (h^dn)^{-1/2} \big(m^{-\zeta}\One_{-m+1\leq l \leq n-m} + (1-l)^{-\zeta}\One_{l\leq -m}\big).
\end{equation}
This along with Markov's inequality gives
\begin{align}
    \label{lm5ineq2}
    \PP\Big(\max_{1 \leq j \leq N} (h^dn)^{1/2}|b_{j,l} \eta_l| \geq u\Big) &\leq \EE\Big(\max_{1 \leq j \leq N} (h^dn)^{q/2}|b_{j,l} \eta_l|^q \Big)/u^q \nonumber\\
    &\lesssim  (h^dn)^{-q/2} \big(m^{-q\zeta}\mathbf{1}_{-m+1\leq l \leq n-m} + (1-l)^{-q\zeta}\One_{l\leq -m}\big)/u^q.
\end{align}
By expressions \eqref{lm5ineq3} and \eqref{lm5ineq2}, for any $u>0$, we can rewrite inequality \eqref{lm5main} as follows
\begin{align*}
    & \quad \PP\Big((h^dn)^{1/2}\max_{1\le j\le N} \Big|\sum_{l\le n-m}b_{j,l}\eta_{l}\Big| \geq z\Big) \nonumber \\
    & \lesssim \sum_{l=-m+1}^{n-m} (h^dn)^{-q/2} m^{-q\zeta}/ u^q + \sum_{l\le -m} (h^dn)^{-q/2} (1-l)^{-q\zeta}/ u^q  + N e^{-z^2m^{2\zeta}}\\
    & \lesssim n(h^dn)^{-q/2}m^{-q\zeta}/u^q + N e^{-z^2/(2zu+2m^{-2\zeta})}.
\end{align*}
Taking $u=z/\log(n)$ yields
$$\PP\Big((h^dn)^{1/2}\max_{1\le j\le N} \Big|\sum_{l\le n-m}b_{j,l}\eta_{l}\Big| \geq z\Big) \lesssim  n\log^q(n)z^{-q}(h^dn)^{-q/2}m^{-q\zeta}+ N e^{-z^2m^{2\zeta}}.$$
By this and Lemma \ref{lemma_weight}, the desired result is achieved.
\end{proof}

\begin{lemma}[Gaussian approximation]
    \label{lemma_m_GA}
    Under the conditions in Theorem \ref{thm1_GA}, we have
    $$\sup_{u\in\RR}\big|\PP\big(I_{\epsilon,m}\le u\big)-\PP\big(I_{z,m}\le u\big) \big|\lesssim (h^dn)^{-1/6} \log^{7/6} (Nn) + (n^{2/q}/(h^dn))^{1/3} \log(Nn) + \PP(\A_n^{\text{c}}),$$
    where $\PP(\A_n^{\text{c}})=O\big(n^{(-q/2+1)\vee(-\xi q)}\big)$ as indicated by Lemma \ref{lemma_weight}.
\end{lemma}

\begin{proof}[Proof of Lemma \ref{lemma_m_GA}]
Recall the filtration $\F_i = (\ldots,v_{i-1},v_{i})$ and the set $\A_n$ defined in Lemma \ref{lemma_weight}. We denote 
\begin{equation}
    \label{eq_lm5_D}
    D_{j,l} = \Big(\sum_{i=1\vee l}^{n \wedge (l+m-1)} w_h(x_j,X_i)\One_{\A_n}\sigma(X_i)a_{i-l}\Big)/\sigma(x_j) ,
\end{equation}
and then, on the set $\A_n$, we can rewrite $I_{\epsilon,m}$ into
\begin{equation}
    \label{eq_lm5_I_epsilon_m}
    I_{\epsilon,m}\One_{\A_n} = \sqrt{h^dn}\max_{1 \leq j \leq N} \Big|\sum_{l=2-m}^n D_{j,l} \eta_{l}\Big|,
\end{equation}
where, for all $1\le j\le N$, conditioned on $\F_n$, $D_{j,l}\eta_{l}$ are independent for different $l$. We aim to apply the Gaussian approximation theorem on expression (\ref{eq_lm5_I_epsilon_m}). Since the proof is quite involved, we shall proceed the proof with two main steps.

\noindent\textbf{Step 1.} The goal of this first step is to show that, for some constant $c_0>0$, with probability greater than $1-O(n^{(-q/2+1)\wedge (-\xi  q)})$,
we have
$$ \min_{1 \leq j \leq N} \sum_{l=2-m}^n \EE \Big\{\big(\sqrt{h^dn} D_{j,l}\eta_{l} \big)^2 \mid \F_n\Big\}\geq c_0 .$$
To see this, it shall be noted that $D_{j,l}$ is $\F_{n}$-measurable. Since $\EE(\eta_l\mid \F_n)=0$ and $\EE(\eta_l^2\mid \F_n)=1$, we have
\begin{align}
    \label{eq_thm1_e1_step1}
    \sum_{l=2-m}^n \EE \Big\{\big(\sqrt{h^dn} D_{j,l}\eta_{l} \big)^2 \mid \F_n\Big\} & = h^dn\sum_{l=2-m}^n D_{j,l}^2 \nonumber \\
    & \ge h^dn \sum_{l=1}^{n-m}  \Big(\sum_{i= l}^{n \wedge (l+m-1)} w_h(x_j,X_i)\One_{\A_n}\sigma(X_i)a_{i-l}/\sigma(x_j) \Big)^2,
\end{align}
where we could decompose $w_h(x_j,X_i)$ into two parts, namely
\begin{equation}
    \label{eq_lm5_w_decompose}
    M_{ji}:=\Big(w_h(x_j,X_i)- \EE\big\{w_h(x_j,X_i)\mid\F_{i-1} \big\}\Big) \quad \text{and}\quad R_{ji}:=\EE\big\{w_h(x_j,X_i)\mid\F_{i-1} \big\}.
\end{equation}
We insert expression (\ref{eq_lm5_w_decompose}) back into (\ref{eq_thm1_e1_step1}) and define
\begin{equation}
    \label{eq_thm1_G}
    \G_{j,l}:=\Big( \sum_{i= l}^{l+m-1} a_{i-l}\sigma(X_i)M_{ji}\One_{\A_n}/\sigma(x_j)\Big)^2.
\end{equation}
Since the density function $g(\cdot)$ is bounded, by expression (\ref{lm52}) in Lemma \ref{lemma_weight}, we obtain
\begin{align}
    \label{lm6ineq}
    |R_{ji}\One_{\A_n}| \lesssim \Big|\int_{\II^d} n^{-1} K(z)g\big(x_j - hz\mid \F_{i-1}\big)dz\Big| \lesssim n^{-1},
\end{align}
which along with expressions (\ref{eq_thm1_e1_step1}), (\ref{eq_thm1_G}), Assumption \ref{asm_sigma} and the Cauchy-Schwarz inequality yields
\begin{align}
    \label{lm61main}
    \sum_{l=2-m}^n \EE \Big\{\big(\sqrt{h^dn} D_{j,l}\eta_{l} \big)^2 \mid \F_n\Big\} & \ge  h^dn \sum_{l=1}^{n-m} \Big( \G_{j,l}^{1/2}-O(1/n) \Big)^2 \nonumber\\
    & \ge h^dn \sum_{l=1}^{n-m}\G_{j,l} - O\Big(h^dn \Big(\sum_{l=1}^{n-m}\G_{j,l}/n\Big)^{1/2}\Big).
\end{align}
We shall investigate the two parts in expression (\ref{lm61main}) respectively. For the first part, note that if $i\neq i'$,
\begin{equation}
    \EE(M_{ji}M_{ji'})
    =\EE\big\{\EE\big(M_{ji}M_{ji'}\mid\F_n\big)\big\}=0.
\end{equation}
Therefore, the cross terms in $\EE\G_{j,l}$ is zero, which along with Assumptions \ref{asm_dep_epsilon}, \ref{asm_kernel} and expression (\ref{lm53}) gives, for some constant $c>0$,
\begin{align}
    \label{E_gil}
    h^dn \sum_{l=1}^{n-m} \EE \G_{j,l}
    & = h^dn\sum_{l=1}^{n-m}\sum_{i=l}^{l+m-1}|a_{i-l}|^2\One_{\A_n}\sigma^2(X_i) \EE(M_{ji}^2)\One_{\A_n}/\sigma^2(x_j) \nonumber\\
    & =  \Bigg(h^dn\sum_{l=1}^{n-m}\sum_{i=l}^{l+m-1} |a_{i-l}|^2\sigma^2(X_i)  \int_{\II}\Big(\frac{K(z)}{h^dng(x_j)}\Big)^2 g(x_j-hz) h^ddz /\sigma^2(x_j) \Bigg)(1+o(1)) \nonumber \\
    & \asymp c.
\end{align}
Next, we shall bound the difference between $\G_{j,l}$ and $\EE\G_{j,l}$, which can be decomposed into two parts, that is
\begin{align}
    \label{eq_lm5_twoparts}
    h^dn\sum_{l=1}^{n-m}\big( \G_{j,l} - \EE\G_{j,l}\big) & = h^dn\sum_{l=1}^{n-m} \big(\G_{j,l}-\EE(\G_{j,l}\mid\F_{l-1})\big) + h^dn \sum_{l=1}^{n-m}\big(\EE(\G_{j,l}\mid\F_{l-1})-\EE \G_{j,l}\big) \nonumber \\
    & =: h^dn (\III_{j,1} + \III_{j,2}).
\end{align}
The first part $\III_{j,1}$ is a sum of martingale differences with respect to $\F_l$. Note that by Assumption \ref{asm_dep_epsilon}, we have $\sum_{k=0}^{m-1}|a_{i,l}|=O(1)$, and it follows from expression (\ref{lm53}) in Lemma \ref{lemma_weight} that $\One_{\A_n}w_h(x_j,X_i)\lesssim 1/(h^dn)$. Therefore, by the definition of $\G_{j,l}$ in expression (\ref{eq_thm1_G}),
\begin{equation}
    \label{eq_lm5_fd_bd}
    \max_{1\le j\le N,1\le l\le n-m}h^dn\big|\G_{j,l}-\EE(\G_{j,l}\mid\F_{l-1})\big| \lesssim 1/(h^dn),
\end{equation}
and 
\begin{equation}
    \label{eq_lm5_fd_var}
    \max_{1\le j\le N}\sum_{l=1}^{n-m}\text{Var}\big(h^dn\G_{j,l}\mid\F_{l-1}\big) \lesssim \max_{1\le j\le N}\Big[\big(\max_{1\le l\le n-m}h^dn\G_{j,l}\big)\cdot\sum_{l=1}^{n-m}\EE(h^dn\G_{j,l}\mid\F_{l-1})\Big] \lesssim 1/(h^dn).
\end{equation}
Therefore, by the Freedman's inequality in Lemma \ref{lemma_freedman}, we obtain
\begin{equation}
    \label{eq_lm5_part1bd}
    \PP\big(\max_{1\leq j\leq N}h^dn|\III_{j,1}|\geq z\big) \lesssim N\exp\Big(-\frac{z^2}{z/(h^dn)+1/(h^dn)}\Big).
\end{equation}
Hence, with probability greater than $1-n^{-q}$, we have $\max_{1\le j\le N}h^dn|\III_{j,1}|\lesssim \sqrt{\log(N)/(h^dn)}$. For the part $\III_{j,2}$, recall the definition of $\G_{j,l}$ in expression (\ref{eq_thm1_G}) and let $$L_{j,i,l}=a_{i-l}\One_{\A_n}\sigma(X_i)M_{ji}/\sigma(x_j).$$ 
Then, we can rewrite $\G_{j,l}$ into
\begin{equation}
    \label{gil_decompose}
    \G_{j,l}=\sum_{i=1\vee l}^{n\wedge (l+m-1)}\Big(L_{j,i,l}^2+2L_{j,i,l}\sum_{s<l}L_{j,i,s}\Big).
\end{equation}
Note that $\EE\big(L_{j,i,l}\mid\F_{l-1}\big)=0$. Thus, we have
\begin{equation}
    \label{condE_gil_part2}
    \EE\big(\G_{j,l}\mid \F_{l-1}\big) = \sum_{i=l}^{n\wedge(l+m-1)}\EE\big(L_{j,i,l}^2\mid\F_{l-1}\big),
\end{equation}
and it follows from Lemma \ref{lemma_weight} that
\begin{align}
    \label{condE_gil_part1}
    \EE\Big(\sum_{i=l}^{l+m-1}L_{j,i,l}^2\mid\F_{l-1}\Big)
    = & \Bigg(\sum_{i=l}^{l+m-1}\frac{1}{h^d}|a_{i-l}|^2\sigma^2(X_i)\int_{\mathbb{I}} \frac{K^2(z)}{n^2g^2(x_j)}g(x_j-hz|\mathcal{F}_{l-1})dz\Bigg)/\sigma^2(x_j).
\end{align}
Then, it follows from Lemma 5.8 in \textcite{zhang2015gaussian} that, for $\xi >0$, we have
\begin{equation}
    \label{eq_lm5_part2_result}
    \PP\big(\max_{1\leq j \leq N} |\III_{j,2}|\geq z/(h^dn^2)\big) \lesssim z^{-q} n^{1\vee(q/2-\xi  q)} \log^{q/2}(N)+e^{-z^2/n}.
\end{equation}
Hence, with probability greater than $1-O(n^{-(\xi q)\wedge(-q/2+1)})$, we have $\max_{1\le j\le N}h^dn|\III_{j,2}|\lesssim \sqrt{\log(N)/n}$. Finally, by combining the results of $I_{j,1}$ and $I_{j,2}$, we achieve that, with probability greater than $1-O(n^{-(\xi q)\wedge(-q/2+1)})$,
$$h^dn\Big|\sum_{l=1}^{n-m}\big(\G_{j,l}-\EE\G_{j,l}\big)\Big| \lesssim \sqrt{\log(N)/(h^dn)}.$$
This, together with expression (\ref{E_gil}) gives, with probability greater than $1-O(n^{(-q/2+1)\wedge (-\xi  q)})$,
$$\min_{1\le j\le N}h^dn \sum_{l=1}^{n-m}\mathcal{G}_{j,l} \gtrsim 1,$$
which along with expression (\ref{lm61main}) completes the proof of Step 1.

\noindent\textbf{Step 2.} In this step we aim to show that, for $B_n=ch^{-d/2}$, some constant $c>0$large enough, $k=1,2$ and $\xi >0$, with probability greater than $1-O(n^{(-q/2+1)\wedge (-\xi  q)})$, we have
$$ \max_{1 \leq j \leq N}n^{(2+k)/2-1} \sum_{l=2-m}^n \EE \Big(\big|\sqrt{h^dn} D_{j,l}\eta_{l}\big|^{2+k}\mid\F_n\Big)\le B_n^{k}.$$
To prove this, we first recall that when conditioned on $\F_n$, $D_{j,l}\eta_l$ are independent over $l$, and by Lemma \ref{lemma_weight}, $\EE\big(|w_h(x_j,X_i)|\mid \F_{i-1}\big) \lesssim n^{-1}$. This, along with the decomposition of $w_h(x_j,X_i)$ in expression (\ref{eq_lm5_w_decompose}) and the definition of $\G_{j,l}$ in expression (\ref{eq_thm1_G}) gives, for $k=1,2$, 
\begin{align}
    \label{eq_lm5_m2_goal}
    & \quad n^{k/2}\sum_{l=2-m}^n \EE\Big(\big|\sqrt{h^dn}D_{j,l} \eta_l\big|^{2+k}\mid\F_n\Big) \nonumber \\
    & \lesssim n^{k/2}(h^dn)^{(2+k)/2}\sum_{l=2-m}^n|D_{j,l}|^{2+k} \nonumber \\
    & = n^{k/2}(h^dn)^{(2+k)/2}\sum_{l=2-m}^n \Big|\sum_{i=l}^{n\wedge (l+m-1)}w_h(x_j,X_i) a_{i-l}\One_{\A_n}\sigma(X_i)/\sigma(x_j)\Big|^{2+k}.
\end{align}
It follows from expression (\ref{lm53}) in Lemma \ref{lemma_weight} that $w_h(x_j,X_i)\One_{\A_n}\lesssim 1/(h^dn)$. Therefore, for $1\le l\le n-m$, we have
\begin{equation}
    \label{eq_lm5_m2}
    \sum_{l=2-m}^n\Big|\sum_{i=l}^{l+m-1}w_h(x_j,X_i) a_{i-l}\One_{\A_n}\sigma(X_i)/\sigma(x_j)\Big|^{2+k}
    \lesssim 1/(h^dn)^{1+k},
\end{equation}
which along with expression (\ref{eq_lm5_m2_goal}) yields
\begin{equation}
    \label{eq_lm5_m2_result}
    n^{k/2}\sum_{l=2-m}^n \EE\Big(\big|\sqrt{h^dn}D_{j,l} \eta_l\big|^{2+k}\mid\F_n\Big) \lesssim 1/(h^d)^{k/2}.
\end{equation}
Since $\EE\Big(\big|\sqrt{h^dn}D_{j,l} \eta_l\big|^q\mid\F_n\Big)\lesssim 1/(h^dn)^{(q-1)/2}$, for $q>2$, it follows from $w_h(x_j,X_i)\One_{\A_n}\lesssim 1/(h^dn)$ that
\begin{align}
    \label{eq_lm5_e2}
    \EE\Big(\max_{1\le j\le N}n^{q/2}\big|\sqrt{h^dn}D_{j,l}\eta_l\big|^q \mid \F_n\Big) & \le  n^{(q-2)/2}\sum_{1\le j\le N}\EE\Big(\big|\sqrt{h^dn}D_{j,l}\eta_l\big|^q \mid \F_n\Big)  \lesssim (h^d)^{-q/2}.
\end{align}
For any $r<q$, we denote
\begin{equation}
    \label{eq_lm5_Bn_step1}
    \tilde M_r:=\max_{1\le j\le N}\Big[n^{r/2}\sum_{l=2-m}^n\EE\Big(\big|\sqrt{h^dn}D_{j,l} \eta_l\big|^r\mid\F_n\Big)\Big]^{1/r},
\end{equation}
and let
\begin{equation}
    \label{eq_lm5_Bn}
    B_n=\max\Big\{\tilde M_3^3,\tilde M_4^2,\max_l\EE\Big(\max_j\big|\sqrt{h^dn}D_{j,l} \eta_l\big|^q\mid\F_n\Big) \Big\} \lesssim h^{-d}.
\end{equation}
The desired result of Step 2 is achieved.

Finally, we combine the results of Steps 1 and 2 and complete the proof by Proposition 2.1 in \textcite{chernozhukov_central_2017}.
\end{proof}

\subsection*{S.3 Proof of Proposition \ref{prop1}}

\begin{proof}[Proof of Proposition \ref{prop1}]
Recall that the estimator for the coefficient vector $\bbeta$ provided by (\ref{eq_betahat}): 
$$\hat\bbeta = \Big(\sum_{i=1}^n(Z_i-\tilde Z_i)(Z_i-\tilde Z_i)^{\top}\One_{X_i\in\T_d}\Big)^{-1}\Big(\sum_{i=1}^n(Y_i-\tilde Y_i)(Z_i-\tilde Z_i)\One_{X_i\in\T_d}\Big).$$
We aim to bound the estimation error $|\hat\bbeta-\bbeta|_{\infty}$, which can be written into
\begin{align}
    \label{eq_prop1_goal}
    |\hat\bbeta-\bbeta|_{\infty} \le & \Big| \Big(\sum_{i=1}^n(Z_i-\tilde Z_i)(Z_i-\tilde Z_i)^{\top}\One_{X_i\in\T_d}\Big)^{-1}\Big|_{\infty} \cdot \Bigg\{ \Big|\sum_{i=1}^n\big[\mu(X_i)-\tilde \mu(X_i)\big](Z_i-\tilde Z_i)\One_{X_i\in\T_d}\Big|_{\infty} \nonumber \\ 
    & + \Big|\sum_{i=1}^n\Big[\sigma(X_i)\epsilon_i-\sum_{t=1}^nw_h(X_i,X_t)\sigma(X_t)\epsilon_t\Big](Z_i-\tilde Z_i)\One_{X_i\in\T_d} \Big|_{\infty}\Bigg\} \nonumber \\
    =: & \Big| \Big(\sum_{i=1}^n(Z_i-\tilde Z_i)(Z_i-\tilde Z_i)^{\top}\One_{X_i\in\T_d}\Big)^{-1}\Big|_{\infty} \cdot \big(|\III_1|_{\infty}+|\III_2|_{\infty}\big).
\end{align}
Recall the definition of $\tilde\mu(X_i)$ in expression (\ref{eq_thm1_mu_tilde}). Similarly, we define
$$\tilde h(X_i) = \sum_{t=1}^nw_h(X_i,X_t)h(X_t),\quad \tilde u_i = \sum_{t=1}^nw_h(X_i,X_t)u_t.$$
For the part $\III_1$, by Assumption \ref{asm_iden}, we have
\begin{align}
    \label{eq_prop1_part1}
    |\III_1|_{\infty}
    \le & \Big|\sum_{i=1}^n\big[\mu(X_i)-\tilde \mu(X_i)\big]\big(h(X_i)-\tilde h(X_i)\big)\One_{X_i\in\T_d}\Big|_{\infty} \nonumber \\
    & + \Big|\sum_{i=1}^n\big[\mu(X_i)-\tilde \mu(X_i)\big]\big(u_i-\tilde u_i\big)\One_{X_i\in\T_d}\Big|_{\infty} =: |\III_{11}|_{\infty} + |\III_{12}|_{\infty}.
\end{align}
Since $w_h(X_i,X_t)=0$ when $|X_i-X_t|_{\infty}>h$, $\mu(\cdot)$ and $h(\cdot)$ are Lipschitz continuous, it follows that
\begin{equation}
    \label{eq_prop1_tildemu}
    \big|\mu(X_i)-\tilde\mu(X_i)\big| = \Big|\sum_{t=1}^nw_h(X_i,X_t)\big[\mu(X_i)-\mu(X_t)\big]\Big| \lesssim h,
\end{equation}
which together with the same result on $h(\cdot)$ part 
yields $|\III_{11}|_{\infty}/n=O(h^2)$. For $\III_{12}$, note that by the definition of $u_i$ in Assumption \ref{asm_iden}, we have $\EE(u_i\mid X_i)=0$. Then, by the weak dependency of $u_i$ imposed in Assumption \ref{asm_iden}, we have
\begin{align*}
    \EE\Big(\sum_{i=1}^n\big[\mu(X_i)-\tilde \mu(X_i)\big]u_i\One_{X_i\in\T_d}\mid \F_n\Big)^2 \lesssim \sum_{i=1}^n\big[\mu(X_i)-\tilde \mu(X_i)\big]^2\One_{X_i\in\T_d} \lesssim nh^2,
\end{align*}
which together with a similar argument for the other part in $\III_{12}$ gives $|\III_{12}|_{\infty}/n=O_{\PP}(h/\sqrt{n})$. Therefore, $|\III_1|_{\infty}/n=O_{\PP}\{h^2 + h/\sqrt{n}\}$.

For the part $\III_2$, we aim to give an upper bound for $\Big\lVert\sum_{i=1\vee k}^n\sigma(X_i)\epsilon_i(Z_i-\tilde Z_i)\One_{X_i\in\T_d}\Big\rVert_q^2$, for $q\ge2$. To this end, we shall note that, by the continuity of $h(X_i)$ and Assumptions \ref{asm_iden} and \ref{asm_sigma},
\begin{align}
    \label{eq_prop1_part21}
    & \quad \Big\lVert\sum_{i=1\vee k}^n\sigma(X_i)\epsilon_i\sum_{t=1}^nw_h(X_i,X_t)(h(X_i)-h(X_t))\One_{X_i\in\T_d}\Big\rVert_q^2 \nonumber \\
    & \lesssim \sum_{1\le k\le n}\Big\lVert\sum_{i=k}^n\P_k(\epsilon_i)\sigma(X_i)\sum_{t=1}^nw_h(X_i,X_t)(h(X_i)-h(X_t))\One_{X_i\in\T_d}\Big\rVert_q^2 \nonumber \\
    & \lesssim \sum_{1\le k\le n} h^2\cdot(k-1)^{-2\zeta+2} \nonumber \\
    & = O(nh^2),
\end{align}
where $\P_k(\cdot)=\EE(\cdot\mid\F_k)-\EE(\cdot\mid\F_{k-1})$ is the projection operator introduced by \textcite{Wu:2005}. Since $u_i$ is independent of $X_i$ and $\epsilon_i$, it follows from Assumptions \ref{asm_sigma}, \ref{asm_dep_epsilon}, \ref{asm_dep} and \ref{asm_iden} that,
\begin{align}
    \label{eq_prop1_part22}
    & \quad \Big\lVert\sum_{i=1\vee k}^n\sigma(X_i)\epsilon_i\sum_{t=1}^nw_h(X_i,X_t)u_i\One_{X_i\in\T_d}\Big\rVert_q^2 \nonumber \\
    & \lesssim \sum_{1\le k\le n}\Big\lVert\sum_{i=k}^n\P_k(\epsilon_i)\sigma(X_i)u_i\One_{X_i\in\T_d}\Big\rVert_q^2 \nonumber \\
    & = o(n),
\end{align}
which along with expression (\ref{eq_prop1_part21}) and a similar argument on the part involving $\tilde u_i$ gives
\begin{equation}
    \label{eq_prop1_part2_result}
    \Big\lVert\sum_{i=1\vee k}^n\sigma(X_i)\epsilon_i(Z_i-\tilde Z_i)\One_{X_i\in\T_d}\Big\rVert_q^2 =o(n).
\end{equation}
This, together with similar arguments for the remaining terms in the part $\III_2$ gives $|\III_2|_{\infty}/n=o(1/\sqrt{n})$.

Now we derive the upper bound for $\big| \big(\sum_{i=1}^n(Z_i-\tilde Z_i)(Z_i-\tilde Z_i)^{\top}\One_{X_i\in\T_d}/n\big)^{-1}\big|_{\infty}$. Note that by the definition of $Z_i$ in Assumption \ref{asm_iden},
\begin{align}
    \label{eq_prop1_zmatgoal}
    \sum_{i=1}^n(Z_i-\tilde Z_i)(Z_i-\tilde Z_i)^{\top} & = \sum_{i=1}^n\big(h(X_i)-\tilde h(X_i)\big)\big(h(X_i)-\tilde h(X_i)\big)^{\top} \nonumber \\
    & \quad + 2\sum_{i=1}^n\big(h(X_i)-\tilde h(X_i)\big)(u_i-\tilde u_i)^{\top}
    + \sum_{i=1}^n(u_i-\tilde u_i)(u_i-\tilde u_i)^{\top} .
\end{align}
For the first part in (\ref{eq_prop1_zmatgoal}), by the continuity of $h(\cdot)$ and Lemma \ref{lemma_weight}, we have, 
\begin{align}
    \label{eq_prop1_zmat1}
    \Big|\sum_{i=1}^n\big(h(X_i)-\tilde h(X_i)\big)\big(h(X_i)-\tilde h(X_i)\big)^{\top}\One_{X_i\in\T_d}\Big|_{\infty} = O_{\PP}(nh^2).
\end{align}
Then, it follows from the dependence decay of $u_i$ in Assumption \ref{asm_iden} that, with high chance,
\begin{align}
    \label{eq_prop1_zmat2}
    \Big|\sum_{i=1}^n\big(h(X_i)-\tilde h(X_i)\big)(u_i-\tilde u_i)^{\top}\One_{X_i\in\T_d}\Big|_{\infty} = O_{\PP}(\sqrt{n}h).
\end{align}
Note that
\begin{align}
    \label{eq_prop1_zmat3}
    & \quad \Big|\sum_{i=1}^n\EE_0\big[(u_i-\tilde u_i)(u_i-\tilde u_i)^{\top} \big]\One_{X_i\in\T_d}\Big|_{\infty} \nonumber \\
    & \lesssim  \Big|\sum_{i=1}^n\EE_0
    (u_i u_i^{\top})\One_{X_i\in\T_d} + \sum_{i=1}^n\EE_0
    (u_i\tilde u_i^{\top})\One_{X_i\in\T_d} + \sum_{i=1}^n\EE_0
    (\tilde u_i\tilde u_i^{\top})\One_{X_i\in\T_d}\Big|_{\infty}.
\end{align}
Due to the weak dependence of $u_i$ over $i$ in Assumption \ref{asm_iden}, we obtain, for $p\ge 4$, $1\le j_1,j_2\le l$,
\begin{align}
    \label{eq_prop1_zmat_proj}
    \Big\|\sum_{i=1}^n\EE_0(u_{i j_1} u_{i j_2})\One_{X_i\in\T_d}\Big\|_{p/2} & \le \sum_{ k\ge 0}\Big\|\sum_{i=1}^n\P_{i-k}\big(\EE_0(u_{i j_1} u_{i j_2})\big)\One_{X_i\in\T_d}\Big\|_{p/2} \nonumber \\
    & \lesssim \sum_{ k\ge 0}\Big(\sum_{i=1}^n\big\|\P_{i-k}\big(\EE_0(u_{i j_1} u_{i j_2})\big)\One_{X_i\in\T_d}\big\|_{p/2}^2\Big)^{1/2} \nonumber \\
    & = O(\sqrt{n}),
\end{align}
where $u_{ij}$ is the $j$-th element in $u_i$. This, together with a similar argument on the other two parts in expression (\ref{eq_prop1_zmat3}) gives
\begin{equation}
    \label{eq_prop1_zmat_E0}
    \Big|\sum_{i=1}^n\EE_0\big[(u_i-\tilde u_i)(u_i-\tilde u_i)^{\top} \big]\One_{X_i\in\T_d}\Big|_{\infty} = O_{\PP}(\sqrt{n}).
\end{equation}
Also, note that by the dependence decay of $u_i$ in Assumption \ref{asm_iden} and properties of weights in Lemma \ref{lemma_weight}, with probability $\PP(\A_n)$,
\begin{align}
    \EE(u_{i j_1}\tilde u_{i j_2})\One_{X_i\in\T_d} = \sum_{t=1}^nw_h(X_i,X_t)\EE(u_{i j_1}u_{t j_2})\One_{X_i\in\T_d} \lesssim 1/(h^dn),
\end{align}
for $1\le j_1,j_2\le n$, which gives
\begin{align}
    \label{eq_prop1_zmat_E}
    & \Big|\sum_{i=1}^n\EE\big[(u_i-\tilde u_i)(u_i-\tilde u_i)^{\top}\big]\One_{X_i\in\T_d} -n\Sigma_u\Big|_{\infty} = O_{\PP}\{1/(h^dn)\big\}. 
\end{align}
Since $\lambda_{\text{min}}(\Sigma_u)$ is lower bounded by Assumption \ref{asm_iden}, by inserting expressions (\ref{eq_prop1_zmat1}), (\ref{eq_prop1_zmat2}), (\ref{eq_prop1_zmat_E0}) and (\ref{eq_prop1_zmat_E}) back to (\ref{eq_prop1_zmatgoal}), we achieve that
\begin{equation}
    \label{eq_prop1_zmat_result}
    \Big| \Big(\sum_{i=1}^n(Z_i-\tilde Z_i)(Z_i-\tilde Z_i)^{\top}\One_{X_i\in\T_d}/n\Big)^{-1}\Big|_{\infty} =O_{\PP}(1).
\end{equation}
Hence, by inserting the upper bounds of $|\III_1|_{\infty}$ and $|\III_2|_{\infty}$ into expression (\ref{eq_prop1_goal}), since $h^4n\rightarrow0$, we achieve
$$|\hat\bbeta-\bbeta|_{\infty}=O_{\PP}\big\{1/\sqrt{n}\big\}.$$
\end{proof}

\subsection*{S.4 Proof of Proposition \ref{prop2}}

\begin{proof}[Proof of Proposition \ref{prop2}]
Recall the $\delta$-net $\{x_j\}_{j=1}^N$ defined in Lemma \ref{lemma_deltanet}. For each $x_j\in\T_d$, the estimator of $\sigma^2(x_j)$ is
$$\hat\sigma^2(x_j)=\sum_{t=1}^nw_h(x_j,X_t)(Y_t-Z_t^{\top}\hat\bbeta -\hat\mu^*(x_j))^2.$$
We define $\tilde\sigma^2(x_j) = \sum_{t=1}^n w_h(x_j,X_t) \sigma^2(X_t).$ Then, for any $1\le j\le N$,
\begin{align}
    \label{sigma1}
    |\hat\sigma^2(x_j)-\sigma^2(x_j)| \le |\hat\sigma^2(x_j)-\tilde\sigma^2(x_j)| + |\tilde\sigma^2(x_j)-\sigma^2(x_j)|.
\end{align}
For the second term in expression (\ref{sigma1}), note that $w_h(x_j,X_t)=0$ when $|x_j-X_t|_\infty >h$. Then, by Lipschitz continuity of $\sigma(\cdot)$ and Assumption \ref{asm_kernel}, we have
\begin{align}
    \label{eq_prop2_term2}
    \max_{1\le j\le N,\, x_j\in\T_d}|\sigma^2(x_j) - \tilde\sigma^2(x_j)| = \max_{1\le j\le N,\, x_j\in\T_d}\Big|\sum_{t=1}^n w_h(x_j,X_t)\big[\sigma^2(x_j)-\sigma^2(X_t)\big]\Big| = O(h).
\end{align}
Now we study the first term in \eqref{sigma1}, which can be written into
\begin{align}
    \label{eq_prop2_term1}
    |\hat\sigma^2(x_j)-\tilde\sigma^2(x_j)|
    & \lesssim |\hat\bbeta-\bbeta|_2^2\Big|\sum_{t=1}^nw_h(x_j,X_t)Z_t\Big|_2^2 + \Big|\sum_{t=1}^nw_h(x_j,X_t)\big[\mu(X_t) -\hat\mu^*(x_j)\big]^2\Big| \nonumber \\
    & \quad + \Big|\sum_{t=1}^nw_h(x_j,X_t)\sigma^2(X_t)(\epsilon_t^2-1)\Big| \nonumber \\
    & =: \III_{j,1} + \III_{j,2}+\III_{j,3}.
\end{align}
For the part $\III_{j,1}$,
by Assumptions \ref{asm_iden} (i) and (ii), it follows from Freedman's inequality in Lemma \ref{lemma_freedman} that, for any $z>0$,
\begin{align}
    \label{eq_prop2_freed}
    \PP\Big(\max_{1\le j\le N,\, x_j\in\T_d}\sum_{t=1}^n|w_h(x_j,X_t)|\cdot\EE_0|Z_t|_2^2 \ge z\Big) \lesssim N\log^{q/2}(z)/(h^dnz)^{q/2} + Ne^{\frac{-z^2}{2/(h^dn)}}.
\end{align}
Let $z=\log(N)/(h^dn)$. Also, notice that by Assumption \ref{asm_iden}, $\max_i\EE\big|\sum_{t=1}^nw_h(x_j,X_t)Z_t\big|_2^2=O(1)$. Then, it follows from expression (\ref{eq_prop2_freed})
and Proposition \ref{prop1} that
\begin{equation}
    \label{eq_prop2_part1_result}
    \max_{1\le j\le N,\, x_j\in\T_d} \III_{j,1}=O_{\PP}\{1/n\}.
\end{equation}
For the part $\III_{j,2}$, since $w_h(x_j,X_t)=0$ for any $|x_j-X_t|_{\infty} >h$ and $\mu(\cdot)$ is Lipschitz continuous, it follows that
\begin{equation}
    \label{eq_prop2_part2_step1}
    \max_{1\le j\le N,\, x_j\in\T_d}\big[\hat\mu(x_j)-\mu(x_j)\big]^2= \max_{1\le j\le N,\, x_j\in\T_d}\Big(\sum_{t=1}^n w_h(x_j,X_t)\big[\mu(X_t)-\mu(x_j)\big]\Big)^2=O(h^2).
\end{equation}
Moreover, by applying a similar argument in expression (\ref{eq_lemmaT_freed}), we achieve
\begin{align}
    \label{eq_prop2_part2_step2}
    \max_{1\le j\le N,\, x_j\in\T_d}\big[\hat\mu^*(x_j)-\hat\mu(x_j)\big]^2 & = \max_{1\le j\le N,\, x_j\in\T_d}\Big(\sum_{t=1}^nw_h(x_j,X_t)Z_t^{\top}(\hat\bbeta - \bbeta)\Big)^2 \nonumber \\
    & \lesssim  \max_{1\le j\le N,\, x_j\in\T_d}|\hat\bbeta - \bbeta|_2^2\Big|\sum_{t=1}^nw_h(x_j,X_t)Z_t\Big|_2^2 \nonumber \\
    & = O_{\PP}(1/n).
\end{align}
As a direct consequence of expressions (\ref{eq_prop2_part2_step1}) and (\ref{eq_prop2_part2_step2}), we have 
$$\max_{1\le j\le N,\, x_j\in\T_d}\III_{j,2} = O\{h^2+1/n\}.$$
Recall the projection operator $\P_k(\cdot)$ in expression (\ref{eq_prop1_part21}). For the part $\III_{j,3}$, we shall note that for $q\ge4$, since $\P_k(\epsilon_t^2)=0$ when $t<k$, by Lemma \ref{lemma_weight} and Assumption \ref{asm_dep_epsilon}, we have
\begin{align}
    \label{eq_prop2_part3_goal}
    & \quad \Big\lVert\max_{1\le j\le N,\, x_j\in\T_d}h^dn\Big|\sum_{t=1}^nw_h(x_j,X_t)\sigma^2(X_t)(\epsilon_t^2-1)\Big|\Big\rVert_q^2 \nonumber \\
    & \le \Big\lVert h^dn\sum_{t=1\vee k}^n\max_{1\le j\le N,\, x_j\in\T_d}\big|w_h(x_j,X_t)\sigma^2(X_t)\P_k(\epsilon_t^2)\big|\Big\rVert_q^2 \nonumber \\
    & \lesssim \Big(h^dn\sum_{t=1\vee k}^n\max_{1\le j\le N,\, x_j\in\T_d}|w_h(x_j,X_t)|(t-k)^{-2\zeta}\Big)^2 \nonumber \\
    & = O_{\PP}(1).
\end{align}
Then, by expression (\ref{eq_prop2_part3_goal}) and Lemma 5.8 in \textcite{zhang2015gaussian}, we have
\begin{align}
    \label{eq_prop2_part3}
    \max_{1\le j\le N,\, x_j\in\T_d}\III_{j,3} =O_{\PP}\Big\{\sqrt{\log(N)/(h^dn)}\Big\}. 
\end{align}
Recall that $\sigma(\cdot)$ is Lipschitz continuous by Assumption \ref{asm_sigma} and we let $\delta=h$, $N=n$. By combining the results of $\III_{j,1}$-$\III_{j,3}$ and expression (\ref{eq_prop2_term2}), it follows from the $\delta$-net approximation in Lemma \ref{lemma_deltanet} that
\begin{equation}
    \label{eq_prop2_result}
    \sup_{x\in\T_d}\big|\hat\sigma^2(x)-\sigma^2(x)\big| \le \max_{1\le j\le N}\sup_{|x-x_j|_{\infty}<\delta}\big|\hat\sigma^2(x_j)-\sigma^2(x_j)\big|=O_{\PP}\Big\{h+\frac{1}{n}+\sqrt{\frac{\log(n)}{h^dn}}\Big\}.
\end{equation}
This completes the proof.
\end{proof}

\subsection*{S.5 Proof of Proposition \ref{prop_longrun}}


\begin{proof}[Proof of Proposition \ref{prop_longrun}]
We aim to provide an upper bound for $\max_{1\le j_1,j_2\le n}|\hat Q_{j_1,j_2}^{(L)} - Q_{j_1,j_2}|$, which can be decomposed into two parts as follows:
\begin{align}
    \max_{1\le j_1,j_2\le n}|\hat Q_{j_1,j_2}^{(L)} - Q_{j_1,j_2}| & \le  \max_{1\le j_1,j_2\le n}|\hat Q_{j_1,j_2}^{(L)} - Q_{j_1,j_2}^{(L)}| + \max_{1\le j_1,j_2\le n}|Q_{j_1,j_2}^{(L)} - Q_{j_1,j_2}| =: \III_1 + \III_2.
\end{align}
For the part $\III_2$, under the same lines in expression (\ref{eq_thm1_longrun}), it follows from Lemma \ref{lemma_weight}, Assumptions \ref{asm_sigma} and \ref{asm_dep_epsilon} that
\begin{align}
    \label{eq_prop3_residual}
    \III_2 & = \max_{1\le j_1,j_2\le n}h^dn\Big|\sum_{|k|\ge L}\sum_{i=1\vee(1-k)}^{n\wedge(n-k)}c_{j_1,j_2,i,k}w_h(X_{j_1},X_i)w_h(X_{j_2},X_{i+k})\gamma(k)\One_{X_{j_1},X_{j_2}\in\T_d}\Big| \nonumber \\
    & \lesssim \max_{1\le j_1,j_2\le n}h^dn\max_{k}\big|w_h(X_{j_2},X_k)\sigma(X_k)/\sigma(X_{j_2})\One_{X_{j_2}\in\T_d}\big| \nonumber \\
    & \quad \cdot\Big|\sum_{k\ge L}\gamma(k)\sum_{i=1}^{n-k}w_h(X_{j_1},X_i)\sigma(X_i)/\sigma(X_{j_1})\One_{X_{j_1}\in\T_d}\Big| \lesssim L^{-\zeta},
\end{align}
where $c_{j_1,j_2,i,k}$ is defined in expression (\ref{eq_cov_Z}) and the constant $\zeta>1$ is defined in Assumption \ref{asm_dep_epsilon}.

Now we study the part $\III_1$. We first note that, for any $1\le j_1,j_2\le n$,
\begin{align}
    \label{eq_prop3_Lpart}
    & \quad |\hat Q_{j_1,j_2}^{(L)} - Q_{j_1,j_2}^{(L)}| \nonumber \\
    & = h^dn\Big|\sum_{k=-L+1}^{L-1}\sum_{i=1\vee (1-k)}^{n\wedge (n-k)}w_h(X_{j_1},X_i)w_h(X_{j_2},X_{i+k})\big[c_{j_1,j_2,i,k}\gamma(k) - \hat c_{j_1,j_2,i,k}\hat\gamma(k)\big]\One_{X_{j_1},X_{j_2}\in\T_d}\Big| \nonumber \\ 
    & \lesssim h^dn\max_{k}\big|w_h(X_{j_2},X_k)\One_{X_{j_2}\in\T_d}\big|\cdot\Big|\sum_{k=0}^{L-1}\sum_{i=1}^{n-k}\big[c_{j_1,j_2,i,k}\gamma(k) - \hat c_{j_1,j_2,i,k}\hat\gamma(k)\big]w_h(X_{j_1},X_i)\One_{X_{j_1}\in\T_d}\Big| \nonumber \\
    & =: h^dn\III_{11,j_2}\cdot\III_{12,j_1,j_2},
\end{align}
where the constant in $\lesssim$ is independent of $h,n,L$. For the part $\III_{11,j_2}$, it follows from Lemma \ref{lemma_weight} that $\max_{j_2}\III_{11,j_2}\lesssim 1/(h^dn)$. Then, for the part $\III_{12,j_1,j_2}$, we consider the decomposition
\begin{align}
    c_{j_1,j_2,i,k}\gamma(k) - \hat c_{j_1,j_2,i,k}\gamma(k) = \hat c_{j_1,j_2,i,k}\big[\gamma(k) - \hat\gamma(k)\big] + \gamma(k)(c_{j_1,j_2,i,k} - \hat c_{j_1,j_2,i,k}).
\end{align}
Note that by Assumption \ref{asm_sigma} and Proposition \ref{prop2}, we have $c_{j_1,j_2,i,k}\asymp1$ and $\hat c_{j_1,j_2,i,k}\asymp1$ with probability tending to 1. Also, for any $k\in\ZZ$, $\gamma(k)$ is bounded by Assumption \ref{asm_sigma}. Therefore, to bound $\III_{12,j_1,j_2}$, it suffices to investigate the term
\begin{align}
    \label{eq_prop3_Lpart2}
    \III_{12,j}^* & := \Big|\sum_{k=0}^{L-1}\sum_{i=1}^{n-k}\big[\gamma(k) - \hat\gamma(k)\big]w_h(X_j,X_i)\One_{X_j\in\T_d}\Big| \nonumber \\
    & \le \Big|\sum_{k=0}^{L-1}\big[\gamma(k) - \hat\gamma(k)\big]\Big|\cdot\max_k\Big|\sum_{i=1}^{n-k}w_h(X_j,X_i)\One_{X_j\in\T_d}\Big| \nonumber \\
    & \lesssim \Big|\sum_{k=0}^{L-1}\big[\gamma(k) - \hat\gamma(k)\big]\Big|,
\end{align}
where the last inequality holds by Lemma \ref{lemma_weight}. To this end, recall the autocovariance of the stationary errors with lag $k$, i.e., $\gamma(k) = \EE(\epsilon_1\epsilon_{1+k})$, and its empirical estimator $\hat\gamma(k) = n^{-1}\sum_{i=1}^{n-k}(\hat\epsilon_i\hat\epsilon_{i+k})$. Then, we can decompose $\hat\gamma(k)-\gamma(k)$ into the deviation part and the expectation part, that is
\begin{align}
    \label{eq_prop3_auto}
    \hat\gamma(k) - \gamma(k) 
    & = \frac{1}{n}\sum_{i=1}^{n-k}\big[\hat\epsilon_i\hat\epsilon_{i+k} - \EE(\epsilon_i\epsilon_{i+k})\big] - \frac{k}{n}\EE(\epsilon_1\epsilon_{1+k}).
\end{align}
By the weak dependence of $\epsilon_i$ in Assumption \ref{asm_dep_epsilon}, it follows from Cauchy-Schwarz inequality that
\begin{align}
    \label{eq_prop3_autocenter}
    \Big|\sum_{k=0}^{L-1}\frac{k}{n}\EE(\epsilon_1\epsilon_{1+k})\Big| 
    & \le \sum_{k=0}^{L-1}\frac{k}{n}\Big|\EE\Big[\Big(\sum_{l=0}^{\infty}\P_{1-l}\epsilon_1\Big)\Big(\sum_{l=0}^{\infty}\P_{1+k-l}\epsilon_{1+k}\Big)\Big]\Big| \nonumber \\
    & \le \sum_{k=0}^{L-1}\frac{k}{n}\sum_{l=0}^{\infty}\Big|\EE\big[\big(\P_{1-l}\epsilon_i\big)\big(\P_{1-l}\epsilon_{1+k}\big)\big]\Big| \nonumber \\
    & \le \sum_{k=0}^{L-1}\frac{k}{n}\sum_{l=0}^{\infty}\big\|\P_{1-l}\epsilon_i\big\|_2\big\|\P_{1-l}\epsilon_{1+k}\big\|_2 \nonumber \\
    & \le \sum_{k=0}^{L-1}\frac{k}{n}\Big\|\sum_{l=0}^{\infty}|a_{1-l}|\Big\|_4\Big\|\sum_{l=k}^{\infty}|a_{1-l}|\Big\|_4 \lesssim L^{-\zeta+2}/n,
\end{align}
for $\zeta>1$. To bound the deviation part in expression (\ref{eq_prop3_auto}), consider the decomposition 
\begin{align}
    \hat\epsilon_i\hat\epsilon_{i+k} - \EE(\epsilon_i\epsilon_{i+k}) & = \big[\hat\epsilon_i\hat\epsilon_{i+k} - \epsilon_i\epsilon_{i+k}\big] + \big[\epsilon_i\epsilon_{i+k} - \EE(\epsilon_i\epsilon_{i+k})\big] \nonumber \\
    & = (\hat\epsilon_i - \epsilon_i)\hat\epsilon_{i+k} + \epsilon_i(\hat\epsilon_{i+k} - \epsilon_{i+k}) + \big[\epsilon_i\epsilon_{i+k} - \EE(\epsilon_i\epsilon_{i+k})\big].
\end{align}
Recall the operator $\EE_0(X)=X-\EE(X)$. Since the errors $(\epsilon_i)$ are stationary and $\EE|\epsilon_i|^q<\infty$ for $q\ge4$ by Assumption \ref{asm_moment}, it follows from Assumption \ref{asm_dep_epsilon} and Lemma \ref{lemma_burkholder} that
\begin{align}
    \EE\Big|\frac{1}{n}\sum_{k=0}^{L-1}\sum_{i=1}^{n-k}\big[\epsilon_i\epsilon_{i+k} - \EE(\epsilon_i\epsilon_{i+k})\big]\Big| & = \frac{1}{n}\EE\Big|\sum_{k=0}^{L-1}\sum_{i=1}^{n-k}\sum_{l_1\le i}\sum_{l_2\le i+k}a_{i-l_1}a_{i+k-l_2}\EE_0(\eta_{l_1}\eta_{l_2})\Big| \nonumber \\
    & = \frac{1}{n}\EE\Big|\sum_{k=0}^{L-1}\sum_{l_1\le n-k}\Big[\sum_{i=1\vee l_1}^{n-k}a_{i-l_1}\sum_{l_2<l_1}a_{i+k-l_2}\EE_0(\eta_{l_1}\eta_{l_2})\Big]\Big| \nonumber \\
    & =: \frac{1}{n}\EE\Big|\sum_{k=0}^{L-1}\sum_{l_1\le n-k}D_{l_1,k}\Big| \lesssim L/n,
\end{align}
where the last inequality holds since $\{D_{l_1,k}\}_{l_1}$ are martingale differences with respect to the filtration $(\ldots,\eta_{l_1-1},\eta_{l_1})$, and therefore, we can apply Lemma \ref{lemma_burkholder} and similar techniques adopted in the proof of Theorem 1 in \textcite{li_ell2_2023} to achieve the desired bound. Next, it follows from Minkowski's inequality and Cauchy-Schwarz inequality that
\begin{align}
    \label{eq_prop3_epsilon_twoparts1}
    \EE\Big|\frac{1}{n}\sum_{k=0}^{L-1}\sum_{i=1}^{n-k}\epsilon_i(\hat\epsilon_{i+k} - \epsilon_{i+k})\Big| & \le \frac{1}{n}\sum_{k=0}^{L-1}\sum_{i=1}^{n-k}\|\epsilon_i\|_2\|\hat\epsilon_i - \epsilon_i\|_2.
\end{align}
Recall that $\hat\epsilon_i=\big(Y_i-Z_i^{\top}\hat\bbeta -\hat\mu^*(X_i)\big)/\hat\sigma(X_i)$, $\hat\mu^*(X_i)=\sum_{t=1}^nw_h(X_i,X_t)(Y_t-Z_t^{\top}\hat\bbeta)$, and $\tilde Y_i$, $\tilde Z_i$ are the kernel estimators of $Y_i$ and $Z_i$ respectively. Then, we have
\begin{align}
    \label{eq_prop3_epsilon_twoparts2}
    \|\hat\epsilon_i - \epsilon_i\|_2 & \le \Big\|\big(Y_i-Z_i^{\top}\hat\bbeta -\hat\mu^*(X_i)\big)\Big(\frac{1}{\hat\sigma(X_i)} - \frac{1}{\sigma(X_i)}\Big)\Big\|_2 + \Big\|\frac{(Y_i-\tilde Y_i) - (Z_i-\tilde Z_i)^{\top}\hat\bbeta - \sigma(X_i)\epsilon_i}{\sigma(X_i)}\Big\|_2 \nonumber \\
    & = \Big\|\big[(Y_i-\tilde Y_i) - (Z_i-\tilde Z_i)^{\top}\hat\bbeta\big]\Big(\frac{1}{\hat\sigma(X_i)} - \frac{1}{\sigma(X_i)}\Big)\Big\|_2 + \Big\|\frac{\mu(X_i) -\hat\mu^*(X_i) + Z_i^{\top}(\bbeta-\hat\bbeta)}{\sigma(X_i)}\Big\|_2 \nonumber \\
    & =: \tilde\III_{i,1} + \tilde\III_{i,2}.
\end{align}
Regarding the first part $\tilde\III_{i,1}$, we note that by the similar M/R decomposition technique applied in the proof of Lemma \ref{lemma_weight} and Proposition \ref{prop2}, it follows from Freedman's inequality in Lemma \ref{lemma_freedman} that
\begin{align}
    \max_{1\le i\le n, X_i\in\T_d}\tilde\III_{i,1} \lesssim h+\sqrt{\log(n)/(h^dn)}.
\end{align}
Similarly, by Proposition \ref{prop1} and the similar arguments in expressions (\ref{eq_prop2_term1})--(\ref{eq_prop2_part2_step2}), we can bound the second part $\tilde\III_{i,2}$ as follows:
\begin{align}
    \max_{1\le i\le n,X_i\in\T_d}\tilde\III_{i,2} \lesssim h +\sqrt{\log(n)/(h^dn)}+ 1/\sqrt{n}.
\end{align}
This, along with the result of $\max_{1\le i\le n, X_i\in\T_d}\tilde\III_{i,1}$ and expressions (\ref{eq_prop3_epsilon_twoparts1}) and (\ref{eq_prop3_epsilon_twoparts2}) yields
\begin{equation}
    \label{eq_prop3_autodeviation}
    \EE\Big|\frac{1}{n}\sum_{k=0}^{L-1}\sum_{i=1}^{n-k}\epsilon_i(\hat\epsilon_{i+k}-\epsilon_{i+k})\Big| \lesssim \frac{L^2}{n}\Big(h+\sqrt{\frac{\log(n)}{h^dn}} + \frac{1}{\sqrt{n}}\Big),
\end{equation}
with high chance. Finally, by inserting expressions (\ref{eq_prop3_autocenter}) and (\ref{eq_prop3_autodeviation}) back into (\ref{eq_prop3_Lpart2}), we achieve, for any $1\le j_1,j_2\le n$, with probability tending to 1,
\begin{align}
    & \quad |\hat Q_{j_1,j_2}^{(L)} - Q_{j_1,j_2}^{(L)}| \lesssim \frac{L^2}{n}\Big(h+\sqrt{\frac{\log(n)}{h^dn}} + \frac{1}{\sqrt{n}}\Big) + L/n,
\end{align}
which together with expression (\ref{eq_prop3_residual}) further gives
\begin{equation}
    \max_{1\le j_1,j_2\le n,X_{j_1},X_{j_2}\in\T_d}|\hat Q_{j_1,j_2}^{(L)} - Q_{j_1,j_2}| = O_{\PP}\Big\{\frac{L^2}{n}\Big(h+\sqrt{\frac{\log(n)}{h^dn}} + \frac{1}{\sqrt{n}}\Big) + L/n + L^{-\zeta}\Big\}.
\end{equation}
The desired result is achieved.
\end{proof}

\subsection*{S.6 Proof of Theorem \ref{thm2_GA}}
\begin{proof}[Proof of Theorem \ref{thm2_GA}]
Recall the definition of $T_n$ in expression (\ref{eq_thm1_T}) and we similarly define
\begin{equation}
    \label{eq_thm2_Tstar}
    T_n^* = \sqrt{h^dn}\sup_{x\in\T_d}\big|\hat\mu^*(x)-\mu(x)\big|/\hat\sigma(x).
\end{equation}
Also, recall that $\hat \Z_j$ is a centered Gaussian random field with conditional covariance matrix $\hat Q^{(L)}=(\hat Q_{j,j'}^{(L)})_{1\le j,j'\le n}$, where $\hat Q_{j,j'}^{(L)}$ is defined in expression (\ref{eq_Q_hat}). For $\Delta_1$--$\Delta_3$ defined in Theorem \ref{thm1_GA}, we denote $\Delta = \Delta_1+\Delta_2+\Delta_3$. Then, it follows from Theorem \ref{thm1_GA} that, for any $\alpha_1,\alpha_2>0$,
\begin{align}
    \label{eq_thm2_goal}
    & \quad \sup_{u\in\RR}\Big[\PP\big( T_n^*\le u\big) - \PP\big(\sup_{t\in\T_d}|\hat\Z_t|\le u\big) \Big] \nonumber \\ 
    & \le \sup_{u\in\RR}\Big[\PP\big( T_n^*\le u\big) - \PP\big(\sup_{t\in\T_d}|\Z_t|\le u\big) \Big] + \sup_{u\in\RR}\Big|\PP\big( \sup_{t\in\T_d}|\Z_t|\le u\big) - \PP\big(\sup_{t\in\T_d}|\hat\Z_t|\le u\big) \Big| \nonumber \\
    & \le \PP\big(|T_n^*-T_n|\ge \alpha\big) + \PP\Big(\big|\sup_{t\in\T_d}|\Z_t|-u\big|\le \alpha\Big) + \sup_{u\in\RR}\Big|\PP\big( \sup_{t\in\T_d}|\Z_t|\le u\big) - \PP\big(\sup_{t\in\T_d}|\hat\Z_t|\le u\big) \Big| + \Delta \nonumber \\
    & =: \sum_{k=1}^3\III_k +\Delta.
\end{align}
We shall investigate the parts $\III_1$--$\III_3$ separately. First, let $\alpha=c_n\sqrt{1/\log(n)}$, for some positive constant $c_n\rightarrow0$ as $n\rightarrow\infty$. As a direct consequence of Lemma \ref{lemma_est},
\begin{align}
    \label{eq_thm2_part1}
    \III_1 = \PP\big(|T_n^*-T_n|\ge c_n\log^{-1/2}(n)\big) =o(1).
\end{align}

For the part $\III_2$, recall that by expressions (\ref{eq_thm1_longrun}) and (\ref{eq_thm1_part5_result}), each entry in the conditional covariance matrix $Q^{(L)}$ is lower bounded, that is $\min_{1\le j\le N}Q_{j,j}^{(L)}\ge c$ for some constant $c>0$. Then, by applying the anti-concentration inequality in Lemma \ref{lemma_nazarov}, we have
\begin{equation}
    \label{eq_thm2_part2}
    \III_2\lesssim \alpha\sqrt{\log(n)} = o(1),
\end{equation}
where the constant in $\lesssim$ only depends on $c$.

Regarding the part $\III_3$, recall the $\delta$-net $\{x_j\}_{j=1}^N\subset\T_d$ defined in Lemma \ref{lemma_deltanet}. We first note that, for any $\alpha_1,\alpha_2>0$, we have
\begin{align}
    \label{eq_thm2_GS}
    & \quad \sup_{u\in\RR}\Big[\PP\big( \sup_{t\in\T_d}|\Z_t|\le u\big) - \PP\big(\sup_{t\in\T_d}|\hat\Z_t|\le u\big)\Big] \nonumber \\
    & \le \PP\Big(\big|\sup_{t\in\T_d}|\Z_t| - \max_{1\le j\le N}|\Z_{t_j}|\big| \ge \alpha_1 \Big) + \sup_{u\in\RR}\Big|\PP\big( \max_{1\le j\le N}|\Z_{t_j}|\le u\big) - \PP\big(\max_{1\le j\le N}|\hat\Z_{t_j}|\le u\big)\Big| \nonumber \\
    & \quad + \PP\Big(\big|\sup_{t\in\T_d}|\hat\Z_t| - \max_{1\le j\le N}|\hat\Z_{t_j}|\big| \ge \alpha_2 \Big) \nonumber \\
    & =: \III_{31} + \III_{32} + \III_{33}.
\end{align}
For the part $\III_{31}$, by Lemma \ref{lemma_deltanet}, we let $\alpha_1 = c_1h$, for some constant $c_1>0$, and obtain
\begin{equation}
    \III_{31} \lesssim 1/n + \PP(\A_n^c),
\end{equation}
where $\PP(\A_n^c) = O\big(n^{(-q/2+1)\vee(-\xi q)}\big)$, for $q\ge4$ and $\xi>0$, as implied by Lemma \ref{lemma_weight}. Similarly, we have $\III_{33}\lesssim 1/n+\PP(\A_n^c)$. For the part $\III_{32}$, it follows from Proposition \ref{prop_longrun} that $\max_{1\le j,j'\le N}|\hat Q_{j,j'}^{(L)} - Q_{j,j'}^{(L)}| = O_{\PP}\big\{n^{-1}L^2\big(h+\sqrt{\log(n)/(h^dn)} + 1/\sqrt{n}\big) + L/n + L^{-\zeta}\big\}$, for $\zeta>1$ and some large positive integer $L$. Let $L=\sqrt{n}$. Then, as a direct consequence of Lemma \ref{lemma_comparison}, we have, with probability tending to 1,
\begin{align}
    \III_{32} & \lesssim \Big[n^{-1}L^2\Big(h+\sqrt{\log(n)/(h^dn)} + 1/\sqrt{n}\Big) + L/n + L^{-\zeta}\Big]^{1/3}\log^{2/3}(N) \nonumber \\
    & \lesssim \Big(hL^2/n+L^2\sqrt{\log(n)/(h^dn^3)} + L/n + L^{-\zeta}\Big)^{1/3}\log^{2/3}(n).
\end{align}

By the results of $\III_1$--$\III_3$ and a similar argument for the other side of the inequality in expressions (\ref{eq_thm2_GS}) and (\ref{eq_thm2_goal}), respectively, we complete the proof.
\end{proof}

\begin{lemma}[Precision of $T_n^*$]
    \label{lemma_est}
    Under the Assumptions in Theorem \ref{thm1_GA}, Propositions \ref{prop1} and \ref{prop2}, for $T_n$ and $T_n^*$ defined in expressions (\ref{eq_thm1_T}) and (\ref{eq_thm2_Tstar}) respectively, we have
    $$\sup_{x\in\T_d}|T_n^*-T_n| = o_{\PP}\big\{1/\sqrt{\log(n)}\big\}.$$
\end{lemma}

\begin{proof}[Proof of Lemma \ref{lemma_est}]
Note that by the definitions of $T_n^*$ and $T_n$, we have, for any $x\in\T_d$,
\begin{align}
    \label{eq_lemmaT_goal}
    |T_n^*-T_n|\le (h^dn)^{1/2} \frac{|\hat\mu^*(x)-\hat\mu(x)|}{\hat\sigma(x)} + (h^dn)^{1/2}\frac{|\hat\mu(x)-\mu(x)|}{\sigma(x)} \cdot \frac{|\hat\sigma(x)-\sigma(x)|}{\hat\sigma(x)}=: \III_1(x) + \III_2(x).
\end{align}
We shall investigate the two parts $\III_1(x)$ and $\III_2(x)$ separately. For the part $\III_1(x)$, by the definitions of $\hat\mu(x)$ and $\hat\mu^*(x)$ in expressions (\ref{eq_mu_hat_sol}) and (\ref{eq_mu_hat_star}), we obtain
\begin{align}
    \label{eq_lemmaT_part1_goal}
    \III_1(x) & = (h^dn)^{1/2}\Big|\sum_{i=1}^nw_h(x,X_i)Z_i^{\top}(\hat\bbeta - \bbeta)\Big|/\hat\sigma(x) \nonumber \\
    & \lesssim (h^dn)^{1/2}|\hat\bbeta - \bbeta|_2\Big|\sum_{i=1}^nw_h(x,X_i)Z_i\Big|_2/\hat\sigma(x).
\end{align}
Recall the definition of $\delta$-net $\{x_j\}_{j=1}^N$ in Lemma \ref{lemma_deltanet}. By Assumptions \ref{asm_iden} (i) and (ii), it follows from Freedman's inequality in Lemma \ref{lemma_freedman} that, for any $z>0$,
\begin{align}
    \label{eq_lemmaT_freed}
    \PP\Big(\max_{1\le j\le N ,x_j\in\T_d}\sum_{i=1}^n|w_h(x_j,X_i)|\cdot\EE_0|Z_i|_2/\hat\sigma(x_j) \ge z\Big) \lesssim n\log^{q/2}(z)/(h^dnz)^{q/2} + ne^{\frac{-z^2}{2/(h^dn)}}.
\end{align}
Let $z=\sqrt{\log(n)/(h^dn)}$. 
Also, we shall note that by Assumptions \ref{asm_sigma}, \ref{asm_iden} and Proposition \ref{prop2}, it follows that $\max_j\EE\big|\sum_{i=1}^nw_h(x_j,X_i)Z_i\big|_2/\hat\sigma(x_j)=O_{\PP}(1)$. 
Then, by applying similar arguments of the $\delta$-net approximation in Lemma \ref{lemma_deltanet}, it follows from Assumption \ref{asm_sigma}
and Proposition \ref{prop1}
that
\begin{equation}
    \label{eq_lemmaT_part1_result}
    \sup_{x\in\T_d}\III_1(x) =O_{\PP}\Big\{h+\sqrt{h^d}\Big\},
\end{equation}
where the first term results from the error due to the $\delta$-net approximation. For the part $\III_2(x)$, by Theorem \ref{thm1_GA}, we have
\begin{equation}
    \label{eq_lemmaT_part21}
    \sup_{x\in\T_d}|\hat\mu(x)-\mu(x)|/\sigma(x) = O_{\PP}\Big\{\sqrt{\log(n)/(h^dn)}\Big\}.
\end{equation}
This, along with Assumption \ref{asm_sigma} and Proposition \ref{prop2} yields
\begin{equation}
    \label{eq_lemmaT_part2_result}
    \sup_{x\in\T_d} \III_2(x) =o_{\PP}\Big\{h\sqrt{\log(n)} +\frac{\sqrt{\log(n)}}{n}+\frac{\log(n)}{\sqrt{h^dn}} \Big\}.
\end{equation}
We insert expressions (\ref{eq_lemmaT_part1_result}) and (\ref{eq_lemmaT_part2_result}) into expression (\ref{eq_lemmaT_goal}), and by the conditions in Theorem \ref{thm1_GA}, we achieve the desired result.
\end{proof}

\end{document}